\documentclass[12pt,onecolumn]{IEEEtran}
\usepackage{amsmath,amsfonts}
\usepackage{algorithmic}
\usepackage{algorithm}
\usepackage{array}
\usepackage[caption=false,font=normalsize,labelfont=sf,textfont=sf]{subfig}
\usepackage{textcomp}
\usepackage{stfloats}
\usepackage{url}
\usepackage{verbatim}
\usepackage{graphicx}
\usepackage{cite}
\usepackage{ieee}

\usepackage[dvipsnames]{xcolor}
\usepackage[hidelinks]{hyperref}

\begin{document}

\title{Error Exponents for Oblivious Relaying and Connections to Source Coding with a Helper}


\author{Han Wu and Hamdi Joudeh
        \thanks{The authors are with Eindhoven University of Technology, the Netherlands. Email: \{h.wu1, h.joudeh\}@tue.nl.
        This work was supported in part by the European Research Council (ERC) under Grant 101116550.}
}



\maketitle

\begin{abstract}
The information bottleneck channel, also known as oblivious relaying, is a two-hop channel where a transmitter sends messages to a remote receiver via an intermediate relay node.
A codeword sent by the transmitter passes through a discrete memoryless channel to reach the relay, which then processes the noisy channel output and forwards it to the receiver through a noiseless rate-limited link.
The relay is oblivious, in the sense that it has no knowledge of the channel codebook used in transmission.
Previous works on oblivious relaying focus on characterizing achievable rates.
In this work, we study error exponents and explore connections to lossless source coding with a helper, also known as the Wyner-Ahlswede-Körner (WAK) problem.

We first establish an achievable error exponent for oblivious relaying under constant compositions codes.
A key feature of our analysis is the use of the type covering lemma to design the relay's compress-forward scheme.
We then show that employing constant composition code ensembles does not improve the rates achieved with their IID counterparts.
We also derive a sphere packing upper bound for the error exponent.
In the second part of this paper, we establish a connection between the information bottleneck channel and the WAK problem. We show that good codes for the latter can be produced through permuting codes designed for the former.
This is accomplished by revisiting Ahlswede's covering lemma, and extending it to achieve simultaneous covering of a type class by several distinct sets using the same sequence of permutations.
We then apply our approach to attain the best known achievable error exponent for the WAK problem, previously established by Kelly and Wagner.
As a byproduct of our derivations, we also establish error exponents and achievable rates under mismatched decoding rules.
\end{abstract}
\section{Introduction}
We study a basic two-hop network comprising a transmitter, a relay and a receiver.
The transmitter is connected to the relay through a discrete memoryless channel (DMC), denoted by $P_{Y|X}$, and the link between the relay and receiver is  noiseless but rate-limited with capacity $B$.
The goal is to send a message from the transmitter to the receiver, where the only connection between the two is via the relay.
To this end, the transmitter uses a channel codebook from which it sends a codeword representing the message to the relay.
The relay processes its noisy observation and forwards an index to the receiver.
From this index, the receiver attempts to retrieve the original message.
The complication here is that while the transmitter and receiver have access to the channel codebook in use over the DMC, the relay does not and hence is \emph{oblivious} to this codebook.
The setting is known as oblivious relaying \cite{sanderovichCommunicationDecentralizedProcessing2008, simeoneCodebookInformationInterference2011}, or equivalently, the information bottleneck (IB) channel \cite{caireInformationBottleneckOblivious2018, steinerBroadcastApproachInformation2021}.

In the process of analyzing a model for oblivious relaying, a key question that arises is how to rigorously model obliviousness at the relay.
An answer to this question was provided in the seminal work of Sanderovich \emph{et al.} \cite{sanderovichCommunicationDecentralizedProcessing2008} through a Bayesian formalization.
In particular, obliviousness is modeled by assuming that the codebook in use by the encoder at the transmitter and the decoder at the receiver is drawn at random from the class of all possible codebooks according to some prior distribution.
While the relay knows the prior distribution, it has no knowledge of the exact codebook being used, and therefore its processing strategy should be chosen such that it works for codebooks in the class with high probability.
Mathematically, this bears close resemblance to random coding as used in achievability proofs \cite{gamalNetworkInformationTheory2011,scarlettInformationtheoreticFoundationsMismatched2020}; or randomized encoding as used in arbitrarily varying channels \cite{csiszarInformationTheoryCoding2011}.
Nevertheless, the motivation here is different as the focus is on modeling the relay's lack of knowledge.
With this Bayesian approach, the task of modeling obliviousness now reduces to choosing a reasonable codebook prior distribution.

The IID prior is adopted in \cite{sanderovichCommunicationDecentralizedProcessing2008}, where all codeword symbols are independently drawn from the same distribution $P_X$ (i.e. IID random codebook ensemble).
This choice may reflect the relay's \emph{belief} that
the employed codebook is one that achieves, e.g., the capacity of the DMC $P_{Y|X}$, and hence its first-order empirical distribution must resemble the capacity-achieving distribution \cite{shamaishitzEmpiricalDistributionGood1997}.
This is also reminiscent of the discrete memoryless source (DMS) model in source coding \cite[Section 3]{shannonMathematicalTheoryCommunication1948},  which ignores higher-order structures.
Under the IID prior, the capacity of the oblivious relay channel described earlier is
\begin{equation}
    C_{\mathrm{IID}}(B) = \max_{P_{X}, P_{U|Y}} I(X;U) \qquad \text{s.t.} \quad I(Y;U) \leq B, \label{eq:introduction_capacity}
\end{equation}
where \(X \to Y \to U\) is a Markov chain.
This follows as a special case from \cite{sanderovichCommunicationDecentralizedProcessing2008}, where a more general model with multiple oblivious relays is considered.
This capacity formula, which can be seen as an instance of the IB problem \cite{tishbyInformationBottleneckMethod1999} (specifically if we fix the
input distribution \(P_X\) to match the source distribution in the IB problem), is the reason why the oblivious relaying setting is also known as the IB channel.
Henceforth we will use the two terms interchangeably.

In establishing \eqref{eq:introduction_capacity}, it becomes clear that obliviousness at the relay effectively limits the relay's processing to compress-forward schemes, and precludes the use of, e.g., decode-forward schemes.\footnote{If the relay is non-oblivious, i.e., it is cognizant of the codebook in use over the DMC \(P_{Y|X}\), then decode-forward achieves capacity, which in this case coincides with the cut-set bound $\min\{I(X;Y),B\}$.}
This limitation is particularly useful for modeling cloud radio access network (C-RAN) architectures, which feature distributed low-cost wireless access nodes, known as remote radio heads (RRHs), connected through wired front-haul links to a centralized cloud server \cite{simeoneCloudRadioAccess2016,pengFronthaulconstrainedCloudRadio2015}.
RRHs can only perform low-level basic processing, e.g., down-conversion and quantization, while more advanced signal processing and channel decoding tasks are performed by the central processor.
The oblivious relay model and compress-forward schemes are effective abstractions for RRHs and their limited functionality; and have been central for analyzing information-theoretic capacity limits for various C-RAN architectures, see, e.g., \cite{simeoneCodebookInformationInterference2011,aguerriCapacityCloudRadio2019,ensanCloudRadioAccess2021}.
Other extensions include, e.g., IB channels with state \cite{caireInformationBottleneckOblivious2018}, fading channels \cite{steinerBroadcastApproachInformation2021,xuInformationBottleneckRayleigh2021}, and multi-user downlink (broadcast) settings \cite{wangAchievabilityDownlinkCloud2018,patilGeneralizedCompressionStrategy2019,ghaddarLowcomplexityCodingTechniques2024}.
The IB channel under mismatched decoding or mismatched compressing rules is studied in \cite{dikshteinMismatchedObliviousRelaying2023}, while second-order achievable rates were recently derived in \cite{liuNonasymptoticObliviousRelaying2025}.
\subsection{Channel Reliability}
All aforementioned works focus on analyzing achievable code rates, or channel capacity, under the IID code ensemble.
Apart from channel capacity, another important figure of merit is the channel reliability function, or error exponent, which captures the exponential decay rate of the decoding error probability at the receiver.
For the DMC, lower and upper bounds for the reliability function, commonly known as the random coding exponent and sphere packing exponent, have been established in classical works by Gallager \cite{gallagerSimpleDerivationCoding1965} (who refined Fano's analysis), Shannon-Gallager-Berlekamp \cite{shannonLowerBoundsError1967}, Haroutunian \cite{haroutunianBoundsExponentProbability1968}, and Csiszár-Körner-Marton \cite{csiszarInformationTheoryCoding2011}, where the latter two rely on constant composition codes.
For the classical relay channel, error exponents have been studied in \cite{tanReliabilityFunctionDiscrete2015}.
However, for the IB channel with an oblivious relay, error exponents have received very little attention (apart from our preliminary work \cite{wuAchievableErrorExponent2024}).

In this work, we will establish an achievable random coding exponent for the IB channel, as well as a sphere packing upper bound.
The exponents we derive recover the corresponding exponents for the DMC when \(B\) is large.
Our analysis relies on the method of types, and therefore it is natural to use the constant composition code ensemble instead of the IID code ensemble commonly used in the oblivious relaying literature.
The use of the constant composition ensemble is also of independent interest, as it represents scenarios where the relay has knowledge of some high-order codebook structure used in transmission.
This naturally gives rise to the question of whether constant composition code ensembles can improve upon the IB channel capacity under IID codes given in \eqref{eq:introduction_capacity}, the same way they improve upon the rates achieved under mismatched decoding \cite{scarlettInformationtheoreticFoundationsMismatched2020}. We answer this question in the negative in this paper.
\subsection{Connections to Source Coding with a Helper}
For reasons that will become clear shortly, let us now turn our attention to the problem of almost lossless source coding with a helper, also known as the Wyner-Ahlswede-Körner (WAK) problem.
Here a transmitter wishes to describe a discrete memoryless source $X^n$ to a receiver, whose goal is to reconstruct this source. The receiver has access to side information provided by a helper, connected to the receiver through a rate-limited link of capacity $B$, and who observes a second source $Y^n$ correlated to $X^n$.

Let \(R_{\text{h}}(B)\) denote the minimum rate for the transmitter's description in the WAK setting described above. Wyner \cite{wynerSourceCodingSide1975} and Ahlswede and Körner \cite{ahlswedeSourceCodingSide1975} showed that this is given by
\begin{equation}
    R_{\text{h}}(B) = \min_{P_{U|Y}} H(X|U) \qquad \text{s.t.}  \quad  I(Y; U) \leq B, \label{eq:characterization_minimum_rate_coded_side_information}
\end{equation}
where \(X \to Y \to U\).
The IB channel capacity in \eqref{eq:introduction_capacity} is closely related to  this minimum rate, specifically if we fix the input distribution \(P_X\) in \eqref{eq:introduction_capacity} to match the source distribution in \eqref{eq:characterization_minimum_rate_coded_side_information}.
In fact, the WAK problem has also been recognized as an instance of information bottleneck problems \cite{zaidiInformationBottleneckProblems2020}.

Following the above observation, it is intriguing to ask the question of whether there exists a deeper level of connection between the IB channel and the WAK problem, beyond their common information-theoretic rate limits.
For example, can coding schemes developed for one problem be applied to the other? In this paper, we establish such a connection by showing that a class of \emph{good} codes which we construct for the IB channel can be transformed into a class of \emph{good} codes for the WAK problem, which in turn achieve the best known error exponent previously derived by Kelly and Wagner in \cite{kellyReliabilitySourceCoding2012}.

In establishing this code-level connection, we draw on an existing connection between special cases of the above problems. Suppose that the bottleneck capacity $B$ is large enough to describe $Y^n$ in an (almost) lossless fashion. This reduces the IB channel to the standard DMC, and the WAK problem to the Slepian-Wolf (SW) problem \cite{Slepian1973}.
Coding for the SW problem can be seen as partitioning the set of source sequences into bins, each of which constitutes a good channel code for the DMC.
This perspective was adopted by Ahlswede and Dueck in \cite{ahlswedeGoodCodesCan1982}, who showed that good constant composition codes for the DMC can be used to construct good partitions for the SW problem through permutations; and then utilized this observation to derive error exponents for the latter problem.\footnote{Similar results were derived by Csisz{\'a}r and K{\"o}rner \cite{csiszarGraphDecompositionNew1981} through a related yet different perspective that does not use permutations.}
Key to their construction is a result known as Ahlswede's covering lemma, which establishes a limit on the number of permutations required to cover a type class from a subset of sequences of the same type.
In this paper, we extend Ahlswede's covering lemma and further develop the Ahlswede-Dueck perspective, showing that good partitions for the WAK problem can also be constructed through permuting good codes for the IB channel.
\subsection{Contributions and Organization}
We now summarize the main technical contributions of this paper.
First, we establish an achievable error exponent for the IB channel under the constant composition ensemble, i.e., the prior at the relay is uniform on a certain type class.
As part of our coding scheme, we design a compress-forward scheme at the relay using the type covering lemma \cite{bergerRateDistortionTheory1971,csiszarInformationTheoryCoding2011}.
The error exponent is established through an intricate analysis of the intersection between conditional type classes.
We further show that the attained error exponent implies that \eqref{eq:introduction_capacity} is achievable, i.e., the IB channel capacity under the IID ensemble is also achievable with  the constant composition ensemble.
For the sake of generality, we carry out the analysis while assuming that the receiver employs a generalized \( \alpha \)-decoder \cite{csiszarGraphDecompositionNew1981}, allowing us to establish an achievable error exponent under mismatched decoding rules and recover an LM rate result derived in \cite{dikshteinMismatchedObliviousRelaying2023}.

Second, we provide a converse proof showing that under the constant composition ensemble, the rate in \eqref{eq:introduction_capacity} cannot be exceeded.
Together with the achievability result mentioned above, this establishes that \eqref{eq:introduction_capacity} is also the capacity of the IB channel under the constant composition ensemble.
In our proof, we analyze the behavior of the constant composition ensemble and establish several properties for its marginal and conditional distributions.
These properties reveal that as far as oblivious relaying is concerned, the constant composition ensemble asymptotically behaves similar to the IID ensemble (i.e., codes without structure), and its higher-order structure cannot help with processing at the oblivious relay.

Third, we establish a sphere packing upper bound for all achievable error exponents under the constant composition ensemble.
We accomplish this by following the approach of Kelly and Wagner \cite{kellyReliabilitySourceCoding2012}, which refines the standard sphere packing argument in the context of the WAK problem; and adapt it to the IB channel. For this, the constant composition converse proof mentioned above is essential.

Finally, we establish a code-level connection between the IB channel and the WAK problem.
In particular, we show that the helper in the WAK problem can be viewed as an oblivious relay, and good source partitions for the WAK problem can be produced through permuting good IB channel codes.
This is achieved by revisiting and extending Ahlswede's covering lemma, showing that a type class can be simultaneously covered by several distinct sets using a single sequence of permutations.
As a demonstration, we transform the coding scheme constructed for the IB channel in our current work to a coding scheme for the WAK problem, and show that it attains the best known achievable error exponent for the WAK problem, previously established in \cite{kellyReliabilitySourceCoding2012}.
Moreover, since the achievable error exponent for the IB channel is established under the generalized \(\alpha\)-decoder, this enables us to derive an achievable error exponent and LM rate for the WAK problem under mismatched decoding rules.

The rest of the paper is organized as follows. After describing key notations at the end of this section, in the next section we provide a formal description of the IB channel under consideration.
In Section \ref{sec:main_results}, we discuss the main results of this paper and provide some insights.
Sections \ref{sec:lower_bound_proof} to \ref{sec:source_coding_coded_side_information_exponent_lower_bound} are dedicated to proving the main results, while proofs of some technical lemmas are deferred to the appendices.
Concluding remarks and future directions are provided in  Section \ref{sec:conclusion}.
\subsection{Notation}
We describe the notation that will be used throughout the work.
Given a finite alphabet \(\mathcal{X}\), we use \(\mathcal{P}(\mathcal{X})\) to denote the set of all probability mass functions (pmfs) \(P_X\) on \(\mathcal{X}\).
We write \(\bm{x}=(x_1, x_2, \ldots, x_n)\) for an \(n\)-length sequence from \(\mathcal{X}^n\).
A random vector on \(\mathcal{X}^n\)  is denoted by \(\bm{X}=(X_1, X_2, \ldots, X_n)\).
Depending on the context, we may also write \(x^n\) and \(X^n\) instead of \(\bm{x}\) and \(\bm{X}\).
In the same way, we adopt the notation \(\bm{y} = (y_1, y_2, \ldots, y_n)\) or \(\bm{u}=(u_1, u_2, \ldots, u_n)\), and \(\bm{Y}\) or \(\bm{U}\), on \(\mathcal{Y}^n\) or \(\mathcal{U}^n\) respectively.
All alphabets in this work are finite.
Following convention, the hat symbol \(\hat{P}\) is used whenever we are looking at the empirical distribution induced by some deterministic sequences.
For a sequence \(\bm{x} \in \mathcal{X}^{n}\), we use \(\hat{P}_{\bm{x} }\) to denote its vector of relative frequencies of all symbols \(x \in \mathcal{X}\), i.e., its type.
\(\hat{P}_{\bm{x}\bm{y}}\) denotes the joint type of a sequence pair \((\bm{x}, \bm{y})\), while \(\hat{P}_{\bm{x} | \bm{y}}\) is the conditional type from \(\bm{y}\) to \(\bm{x}\) induced by  \(\hat{P}_{\bm{x}\bm{y}}\).
The set of all possible types \(\hat{P}_{\bm{x}}\) on \(\mathcal{X}^{n}\) is written as \(\mathcal{P}_n(\mathcal{X})\), while the set of all possible conditional types \(\hat{P}_{\bm{x} | \bm{y}}\) for sequences from \(\mathcal{Y}^n\) and \(\mathcal{X}^n\) is written as \(\mathcal{P}_n(\mathcal{X} | \mathcal{Y})\).
The type class \(\mathcal{T}_n(P_X)\) consists of all sequences \(\bm{x}\) that have the same type \(P_X \in \mathcal{P}_n(\mathcal{X})\).
For a given sequence \(\bm{y}\), the conditional type class \(\mathcal{T}_n(P_{X|Y} | \bm{y})\) is the set of all sequences \(\bm{x}\) such that the conditional type from \(\bm{y}\) to \(\bm{x}\) is \(P_{X|Y} \in \mathcal{P}_n(\mathcal{X} | \mathcal{Y})\).

The entropy of \(P_{X}\) is written as \(H(X)\) or \(H(P_X)\) and the conditional entropy between two random variables \(X\) and \(Y\) is denoted by \(H(Y|X)\) or \(H(P_{Y|X} | P_X)\), while the mutual information between \(X\) and \(Y\) is written as \(I(X;Y)\) or \(I(P_X, P_{Y|X})\).
\(D(Q_X\|P_X)\) is the KL-divergence between two pmfs \(Q_X\) and \(P_X\), and \(D(Q_{Y|X}\|P_{X|Y} | P_X)\) denotes the conditional KL-divergence.
Given an event \(\mathcal{A}\), we use \(P[\mathcal{A}]\) to denote the probability of \(\mathcal{A}\) under the probability measure \(P\), while \(\idc{\mathcal{A}}\) is the indicator function of \(\mathcal{A}\) and \(\abs{\mathcal{A}}\) is its cardinality or size.
Given two sets \(\mathcal{A}\) and \(\mathcal{B}\), we use \(\mathcal{A} - \mathcal{B}\) to denote the elements from \(\mathcal{A}\) but not in \(\mathcal{A} \cap \mathcal{B}\).
For a conditional distribution \(P_{Y|X}\) with \(X \overset{P_{Y|X}}{\to} Y\), we use \(P_X \cdot P_{Y|X}\) to denote the distribution of \(Y\) when the input distribution is \(P_X\).
For a Markov chain \(X \overset{P_{Y|X}}{\to} Y \overset{P_{U|Y}}{\to} U\), we use \(P_{Y|X} \cdot P_{U|Y}\) to denote the conditional distribution between \(X\) and \(U\) through the Markov chain.
We write \(a_n \ndot{=} b_n\) if \(\lim_{n \to \infty}\frac{1}{n}\log (a_n/b_n) = 0\) and \(a_n \ndot{\leq} b_n\) if \(\limsup_{n \to \infty}\frac{1}{n} \log (a_n/b_n) \leq 0\).
For a positive integer constant \(N\), we use \([N]\) to denote \(\{1,2,\ldots,N\}\).
Let \(|a|^{+} \triangleq \max\{0,a\}\).
The base of exponential and log functions is chosen as the natural base.
\section{Problem Setup}
\label{section:problem_setting}
We now provide a more detailed description of the information bottleneck (IB) channel.
As illustrated in Fig. \ref{fig:channel_model_oblivious_relaying}, the setting comprises a transmitter, an oblivious relay, and a receiver.
The task is to reliably transmit a message \(M\), uniformly distributed over the message set
\( [e^{nR}]\), to the receiver.

The relay's obliviousness is modeled by assuming the codebook \(\mathcal{C}_n\) used in transmission is drawn at random from a codebook ensemble.
The oblivious relay is cognizant of the random codebook ensemble, but not the exact codebook realization in use.
Let $ \bm{C} = (\bm{X}(1), \bm{X}(2), \ldots, \bm{X}(e^{nR}))$ denote the random codebook ensemble, where a fixed codebook \(\mathcal{C}_n\) is a realization of $ \bm{C}$.
We adopt the constant composition ensemble, where codewords in
$ \bm{C}$ are independently and uniformly distribution over the type class \(\mathcal{T}_n(P_X)\) for a certain type \(P_X \in \mathcal{P}_n(\mathcal{X})\).
Therefore, $ \bm{C}$ is uniformly distributed on the codebook set $\mathcal{T}_n(P_X)^{e^{nR}} $.

\begin{figure}[hbt!]
    \centering
    \scalebox{1}{




\begin{tikzpicture}
        \node at (0,0) [name = source] { $M$ } ;
        \node at (2,0) [draw, rectangle, name = tx, label= above:{\small Transmitter}, minimum height=0.8cm,minimum width=2cm] {$f_n(M, \bm{C})$};
        \node at (5,0) [draw, rectangle, name = channel, label= above:{\small Channel}, minimum height=0.8cm,minimum width=1.8cm] {$P_{Y|X}$};
        \node at (8,0) [draw, rectangle, name = relay, label= above:{\small \smash{Relay} }, minimum height=0.8cm,minimum width=1.8cm] {$\varphi_n$};
        \node at (12,0) [draw, rectangle, name = rx, label= above:{\small Receiver}, minimum height=0.8cm,minimum width=2cm] {$\phi_n(L, \bm{C})$};
        \node at (14,0) [name = output] { $\hat{M}$ } ;
        \node at (7, -1.5) [name = codebook] { $\bm{C}$} ;

        \draw [->] (source.east) -- (tx.west);
        \draw [->] (rx.east) -- (output.west);

        \draw [->] (tx.east) -- (channel.west);
        \draw [->] (channel.east) --  (relay.west);
        \draw [->] (relay.east) -- node[align = center, above]{\small $B$ } (rx.west);
        \draw (codebook.west) -- (2,-1.5);
        \draw [->] (2,-1.5) --(tx.south);
        \draw (codebook.east) -- (12,-1.5);
        \draw [->] (12,-1.5) --(rx.south);
\end{tikzpicture}}
    \captionsetup{justification=centering}
    \caption{Information Bottleneck Channel}
    \label{fig:channel_model_oblivious_relaying}
\end{figure}
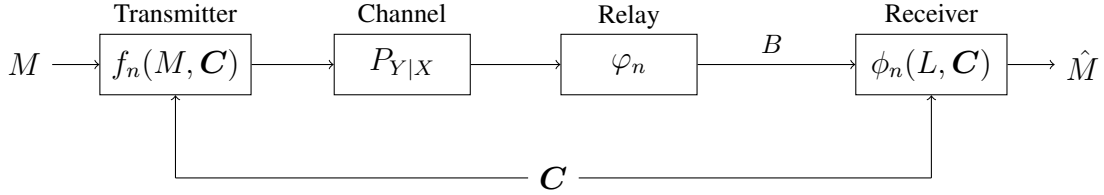

Given a random codebook selection \(\bm{C} = \mathcal{C}_n \), where \( \mathcal{C}_n = ( \bm{x}(1),\ldots,\bm{x}(e^{nR})) \), transmission proceeds as follows.
For a message \(M = m \in [e^{nR}]\), the transmitter assigns the codeword \(\bm{x}(m)\) from  \(\mathcal{C}_n\)
through the mapping \(f_n:[e^{nR}] \times  \mathcal{T}_n(P_X)^{e^{nR}}  \to \mathcal{X}^n\), and sends it over the channel.
The channel between the transmitter and the relay is a DMC \(P_{Y|X}\), i.e.,  the distribution of the channel output \(\bm{Y}\) at the relay follows the law
\begin{equation}
    P_{Y|X}^n(\bm{y} | \bm{x}(m)) = \prod_{i=1}^{n} P_{Y|X}(y_i | x_i(m)).
\end{equation}
The oblivious relay compresses its observation \(\bm{y}\) into \(l = \varphi_n(\bm{y}) \in [e^{nB}]\) and forwards it to the receiver through a noiseless link (i.e. bottleneck) of capacity \(B\),
where \(\varphi_n:\mathcal{Y}^n \to [e^{nB}]\) is the relay's mapping.
With knowledge of which codebook \(\mathcal{C}_n\) has been used by the transmitter, and the index \(l\) forwarded by the relay, the receiver attempts to determine which message has been sent and produces a message estimate \(\hat{M} = \hat{m}\), through a decoding mapping \(\phi_n:[e^{nB}] \times  \mathcal{T}_n(P_X)^{e^{nR}} \to [e^{nR}]\).

It should be noted that for any given message $m \in [e^{nR}]$ and index $l \in [e^{nB}]$, the encoding and decoding mappings $f_n(m,\bm{C})$ and $\phi_n(l,\bm{C})$ are random, due to the random codebook ensemble $\bm{C}$. Conditioned on \(\bm{C} = \mathcal{C}_n\), then $f_n(m,\mathcal{C}_n)$ and $\phi_n(l,\mathcal{C}_n )$ reduce to standard deterministic encoding and decoding rules.

The IB channel with bottleneck \(B\) will be written as \((P_{Y|X}, B)\).
The mapping vector \((f_n, \varphi_n, \phi_n)\) as described above is called an \((n, R, B)\)-code for the IB channel \((P_{Y|X}, B)\).
Given a codebook realization \(\bm{C} = \mathcal{C}_n\),
the decoding error probability of message \(m\) is defined as
\begin{equation}
    \lambda_m(n, R, B, \mathcal{C}_n) \triangleq \P \{ \hat{M} \neq M | M  = m, \bm{C} = \mathcal{C}_n\} \qquad \forall m \in [e^{nR}],
\end{equation}
where $\hat{M} = \phi_n \big( \varphi_n(\bm{Y}) ,\mathcal{C}_n \big)$.
The average decoding error probability over messages under \(\mathcal{C}_n\) is
\begin{equation}
    \bar{\lambda} (n, R, B, \mathcal{C}_n) \triangleq \frac{1}{e^{nR}}\sum_{m=1}^{e^{nR}}\lambda_m(n, R, B, \mathcal{C}_n).
\end{equation}
Since the relay is oblivious to the codebook realization \(\bm{C} = \mathcal{C}_n\), it instead seeks the compressor \(\varphi_n\) that minimizes the average decoding error probability over the entire random ensemble \(\bm{C}\).
Thus, the performance of an \((n, R, B)\)-code is measured through its ensemble-average decoding error probability
\begin{equation}
    \bar{\lambda}(n, R, B) \triangleq \E [ \bar{\lambda}(n, R, B, \bm{C}) ].
\end{equation}
We say that the rate \( R \) is achievable under constant composition codes if there exists a sequence of \((n, R, B)\)-codes such that $ \bar{\lambda}(n, R, B) \to 0$ as $n \to \infty$.
The \emph{capacity} \(C(B)\) is defined as the supremum of all achievable rates $R$  under constant composition codes.

Besides capacity, we are also interested in the exponential decay rate of \(\bar{\lambda}(n, R, B)\) for  \(R < C(B)\).
For the IB channel \((P_{Y|X}, B)\), the maximum achievable error exponent \(E(R, B)\), i.e., its reliability function, is the maximum \(\beta \geq 0\) for which there exists a sequence of \((n, R, B)\)-codes such that
\begin{equation}
    \liminf_{n \to \infty} -\frac{1}{n} \log \bar{\lambda}(n, R, B) \geq \beta, \quad \text{where} \ R < C(B).
\end{equation}
In this work, we will characterize the capacity \(C(B)\) under constant composition codes as well as establish lower and upper bounds for the reliability function \(E(R,B)\).
\begin{remark}
We may also define the ensemble-average error probability for message $m$ as
\begin{equation}
    \lambda_m(n, R, B) \triangleq \E [ \lambda_m(n, R, B, \bm{C}) ].
\end{equation}
which we use further on in the paper. It is easy to see that $ \bar{\lambda}(n, R, B) =  \frac{1}{e^{nR}}\sum_{m=1}^{e^{nR}} \lambda_m(n, R, B)$.
\end{remark}
\section{Main Results and Discussions}
\label{sec:main_results}
\subsection{Achievable Error Exponent and Rate}
We establish an achievable error exponent under constant composition codes, i.e.,  a lower bound for \(E(R, B)\).
To this end, consider an arbitrary auxiliary alphabet \(\mathcal{U}\) and define
\begin{align}
    E_{\textnormal{r}}(R, B, P_X) & \triangleq \min_{Q_Y } \max_{P_{U|Y}} \min_{ \substack{Q_{X|YU}: \\ Q_X = P_X} }  D(Q_{Y|X}  \| P_{Y|X} | P_X) + I_Q(X;U|Y) + \nonumber \\
    & \hspace{6cm} \big | I_Q(X;U) -R - |I_Q(Y;U) - B |^{+}  \big| ^{+}, \label{eq:achievability_definition_E_r_R_B_P_X}
\end{align}
where the inner minimization is over all \(Q_{X|YU}\) such that the joint distribution \(Q_{XYU} = Q_Y \times P_{U|Y} \times Q_{X|YU}\) satisfies \(Q_X = P_X\).
An interpretation of \(E_{\textnormal{r}}(R, B, P_X)\) is provided following the next theorem.
\begin{theorem}
    \label{thm:lower_bound}
    For the IB channel \((P_{Y|X}, B)\), we have
    \begin{equation}
        E(R,B) \geq \max_{P_X} E_{\textnormal{r}}(R,B ,P_X).
    \end{equation}
\end{theorem}
\begin{proof}
    See Section \ref{sec:lower_bound_proof}.
\end{proof}

We now briefly discuss the coding scheme employed to establish Theorem \ref{thm:lower_bound}, and provide some insights into the expression of \(E_{\textnormal{r}}(R, B, P_X)\).
The relay uses a compress-forward scheme based on \emph{type covering}, where each output type class \(\mathcal{T}_n(Q_Y)\) at the relay is covered using roughly \(e^{nI(Q_Y, P_{U|Y})}\) sequences from \(\mathcal{U}^n\) for some conditional type \(P_{U|Y}\).
Since the rate between the relay and the receiver is limited to \(B\), if \(I(Q_Y, P_{U|Y}) > B\), we partition the \(e^{nI(Q_Y, P_{U|Y})}\) sequences into \(e^{nB}\) bins  with bin size
\(e^{n|I(Q_Y, P_{U|Y}) - B |^{+}}\) and the relay forwards the bin index.
Note that \(P_{U|Y}\) can vary for different type classes \(\mathcal{T}_n(Q_Y)\).

Given a forwarded bin index, the receiver searches through all pairs of codewords and bin sequences from the codebook and the bin, and chooses a pair \((\bm{x}(m), \bm{u})\) that maximizes the empirical mutual information, i.e.,
MMI decoding.
This leads to the occurrence of \(I_Q(X;U) - R - |I_Q(Y;U) - B |^{+} \) in \(E_{\textnormal{r}}(R, B, P_X)\), reflecting the number of codeword-sequence pairs that can lead to an error, i.e., \(e^{n(R+ |I_Q(Y;U) - B |^{+})}\), and their  probability under the random ensemble, i.e., \(e^{-nI_Q(X;U)}\).

The conditional mutual information term \(I_Q(X;U|Y)\) in \(E_{\textnormal{r}}(R, B, P_X)\)
reflects the performance of the compress-forward strategy under the random codebook ensemble, i.e., it captures the correlation between the transmitted codeword \(X^n\) and its compress-forward sequence \(U^n\).
The more correlation between the two, i.e., the more informed the receiver is, the less likely the receiver will make a decoding error by deciding a different codeword is transmitted.
It is conditioned on \(Y\) since the relay has the knowledge of channel output \(Y^n\).
As for the sandwiched maximization over \(P_{U|Y}\), this reflects the fact that \(P_{U|Y}\) can be separately optimized for every output type class \(\mathcal{T}_n(Q_Y)\).

As a consequence of Theorem \ref{thm:lower_bound}, we obtain the following achievable rate.
\begin{corollary}
    \label{cor:achievable_rate}
    For the IB channel \((P_{Y|X}, B)\), we have
    \begin{equation}
        C(B) \geq \max_{P_X, P_{U|Y}} I(X;U) \qquad \text{s.t.} \quad I(Y;U) \leq B,
    \end{equation}
    where \(X \overset{P_{Y|X}}{\to} Y \overset{P_{U|Y}}{\to} U\) forms a Markov chain.
\end{corollary}
\begin{proof}
    See Section \ref{sec:proof_achievable_rate_corollary}.
\end{proof}
Corollary \ref{cor:achievable_rate} shows that the IB channel capacity under the IID ensemble in \eqref{eq:introduction_capacity} is also achievable with the constant composition ensemble, i.e., $C(B) \geq C_{\mathrm{IID}}(B)$ which is perhaps not surprising.
\begin{remark}[Mismatched decoding]
The proof of Theorem \ref{thm:lower_bound} is established under the generalized decoder, known as the \(\alpha\)-decoder \cite{csiszarGraphDecompositionNew1981}.
By specializing the generalized decoder, we obtain an achievable error exponent for the oblivious relaying setting under a mismatched decoding rule, and recover the LM-rates previously derived in \cite{dikshteinMismatchedObliviousRelaying2023}.
See Theorem \ref{thm:lower_bound_mismatch} and Corollary \ref{cor:lower_bound_rate_mismatch} in Section \ref{sec:lower_bound_error_exponent_mistmathced_decoder}.
\end{remark}
\subsection{Converse}
Having shown that $C(B) \geq C_{\mathrm{IID}}(B)$, we now address the question of whether $C(B)$ can be strictly greater than $C_{\mathrm{IID}}(B)$. We believe that this is not obvious or immediate for the following reasons. It has been shown in Gaussian settings that achievable rates are improved by using codebooks with some structure, e.g., BPSK instead of Gaussian ensembles \cite{sanderovichCommunicationDecentralizedProcessing2008}.
The intuition is that structure enables the oblivious relay to perform useful pre-processing, e.g., demodulation.
In DMC settings, constant composition ensembles have higher-order structure compared to their IID counterparts and result in better rates under, e.g., mismatched decoding rules \cite{scarlettInformationtheoreticFoundationsMismatched2020}.
It is therefore desirable to investigate whether constant composition codes are still capable of this for oblivious relaying.
In the following result, we answer this question in the negative.
\begin{theorem}
    \label{thm:weak_converse_constant_composition_ensemble}
    The capacity of the IB channel \((P_{Y|X}, B)\) under the constant composition ensemble is
    \begin{equation}
        C(B) = \max_{P_X, P_{U|Y}} I(X;U) \qquad \text{s.t.} \quad I(Y, U) \leq B, \label{eq:characterization_capacity_IB_channel}
    \end{equation}
    where \(X \overset{P_{Y|X}}{\to} Y \overset{P_{U|Y}}{\to} U\) forms a Markov chain and \(\abs{\mathcal{U}} \leq \abs{\mathcal{Y}} + 1\).
\end{theorem}
\begin{proof}
    See Section \ref{sec:weak_converse_for_capacity}.
\end{proof}

To establish Theorem \ref{thm:weak_converse_constant_composition_ensemble}, we investigate the marginal and conditional distributions of the constant composition ensemble.
We present several properties of the ensemble, listed in Section \ref{sec:weak_converse_properties_constant_composition_ensemble}.
These properties reveal that the higher-order structures of constant composition codes are weak, and the constant composition ensemble asymptotically behaves the same as the IID ensemble, i.e., codes without structure, as far as the capacity of the information bottleneck channel is concerned.
\subsection{Sphere Packing Bound}
Next, we provide an upper bound for \(E(R,B)\). For this purpose, define
\begin{equation}
    E_{\textnormal{sp}}(R, B, P_X) \triangleq \min_{Q_{Y}} \max_{ \substack{ P_{U|Y}: \\ I(Q_Y, P_{U|Y}) \leq B } }  \min_{ \substack{Q_{Y|X}: \\ P_X \cdot Q_{Y|X} = Q_{Y}, \\ I(P_X, Q_{Y|X} \cdot P_{U|Y} )  \leq R } } \! \! D(Q_{Y|X} \| P_{Y|X} | P_{X}) \label{eq:main_results_definition_E_sp}
\end{equation}
\begin{theorem}
    \label{thm:sphere_packing_bound}
    For the IB channel \((P_{Y|X}, B)\), every sequence of \((n, R, B)\)-codes with codeword composition being \(P_X\) satisfies
    \begin{equation}
        \limsup_{n \to \infty} -\frac{1}{n} \log \bar{\lambda}(n, R, B) \leq E_{\textnormal{sp}}(R, B, P_X),
    \end{equation}
    where \(\abs{\mathcal{U}} \leq \abs{\mathcal{X}}\abs{\mathcal{Y}} + \abs{\mathcal{Y}} + 1\).
    Therefore, we have
    \begin{equation}
        E(R, B) \leq \max_{P_X} E_{\textnormal{sp}}(R, B, P_X).
    \end{equation}
\end{theorem}
\begin{proof}
    See Section \ref{sec:sphere_packing_bound}.
\end{proof}
To establish Theorem \ref{thm:sphere_packing_bound}, we follow the approach of Kelly and Wagner \cite{kellyReliabilitySourceCoding2012}, developed in the context of the WAK problem, and adapt it to the oblivious relaying problem.
The  Kelly-Wagner approach refines Haroutunian's traditional proof of the sphere packing bound for DMCs \cite{haroutunianBoundsExponentProbability1968} (see also \cite{martonErrorExponentSource1974} and \cite{blahutHypothesisTestingInformation1974}).
In particular, compared to the traditional approach, the refinement can be seen through the sandwiched maximization over \(P_{U|Y}\) in \eqref{eq:main_results_definition_E_sp}.
Note that the converse for the capacity under constant composition codes in Theorem \ref{thm:weak_converse_constant_composition_ensemble} is a cornerstone for establishing the sphere packing bound in Theorem \ref{thm:sphere_packing_bound}.
\subsection{Connections to the WAK Problem}
We now establish a connection between the IB channel and the WAK problem.
Before starting, we first provide a more detailed description of the WAK problem.
Consider a joint pmf \(P_{XY} \in \mathcal{P}(\mathcal{X} \times \mathcal{Y})\).
As seen in Fig. \ref{fig:the_wak_problem}, we have a DMS pair \((X^n, Y^n)\) following the distribution
\begin{equation}
    P_{X^nY^n}(x^n, y^n) = \prod_{i=1}^{n}P_{XY}(x_i, y_i).
\end{equation}
We can interpret \(X^n\) as a source and \(Y^n\) as its side information.
A transmitter observes the source \(X^n\) and describes it to a receiver through an encoder \(f^{\prime}_n: \mathcal{X}^n \to [e^{nR}]\).
A helper observes the side information \(Y^n\) and independently provides its description through another encoder \(\varphi_n': \mathcal{Y}^n \to [e^{nB}]\).
A receiver reconstructs \(\hat{X}^n\) through a decoder \(\phi_n^{\prime}:[e^{nR}] \times [e^{nB}] \to \mathcal{X}^n\) after receiving the two descriptions.

\begin{figure}[hbt!]
    \centering
    \scalebox{1}{




\begin{tikzpicture}
        \node at (0,0) [draw, rectangle,  name = source, label= above:{\small Source },  minimum height=0.8cm, minimum width=1.8cm]  {$X^n$};

        \node at (3,0) [draw, rectangle,  name = channel, label= above:{\small Channel }, minimum height=0.8cm,minimum width=1.8cm] {$P_{Y | X} $};
        \node at (6,0) [draw, rectangle,  name = helper, label= above:{\small \smash{Helper} }, minimum height=0.8cm,minimum width=1.8cm] {$\varphi_n^{\prime}$};
        \node at (10,0) [draw, rectangle,  name = receiver, label= above:{\small Receiver}, minimum height=0.8cm,minimum width=1.8cm] {$\phi_n^{\prime}$};
        \node at (5,-1.5) [draw, rectangle,  name = transmitter, label= above:{\small Transmitter}, minimum height=0.8cm,minimum width=1.8cm] {$f_n^{\prime}$};
        \node at (12,0) [name = output] { $\hat{X}^n$ } ;

        \draw [->] (source.east) -- (channel.west);
        \draw [->] (channel.east) --  (helper.west);
        \draw [->] (helper.east) -- node[align = center, above]{\small $B$ } (receiver.west);
        \draw [->] (0,-1.5) -- (transmitter.west);
        \draw  (0,-1.5) --(source.south);
        \draw (transmitter.east) -- (10,-1.5);
       \draw [->] (10,-1.5) --(receiver.south);
       \draw [->] (receiver.east) -- (output.west);

\end{tikzpicture}}
    \captionsetup{justification=centering}
    \caption{WAK Problem}
    \label{fig:the_wak_problem}
\end{figure}
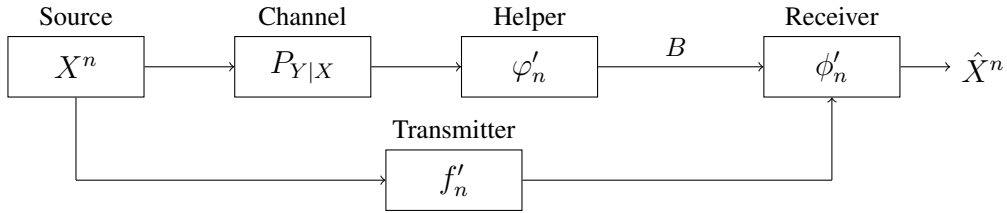

We call the mapping vector \((f^{\prime}_n, \varphi_n', \phi^{\prime}_n)\) an \((n, R, B)\)-code for the DMS pair \((X^n, Y^n)\).
The performance of an \((n, R, B)\)-code is measured through the decoding error probability
\begin{equation}
    \lambda^{\prime}(n,R,B) \triangleq \P \{ \hat{X}^n \neq X^n \}.
\end{equation}

We say that rate \(R\) is achievable if there exists a sequence of \((n,R,B)\)-codes such that \(\lambda^{\prime}(n, R, B) \to 0\).
The optimal (i.e. minimum) achievable rate was found in \cite{wynerSourceCodingSide1975,ahlswedeSourceCodingSide1975}  to be equal to \(R_{\text{h}}(B)\) described in \eqref{eq:characterization_minimum_rate_coded_side_information}.
In this work, we are interested in the reliability function (error exponent) \(E_{\text{h}}(R, B)\), that is the maximum \(\beta \geq 0\) for which there exists a sequence of \((n, R, B)\)-codes such that
\begin{equation}
    \liminf_{n \to \infty} -\frac{1}{n} \log \lambda^{\prime}(n, R, B) \geq \beta, \quad \text{where } R > R_{\text{h}}(B).
\end{equation}
It has been observed in \cite{zaidiInformationBottleneckProblems2020} that solving \eqref{eq:characterization_minimum_rate_coded_side_information} is equivalent to solving \eqref{eq:characterization_capacity_IB_channel} (if we ignore the optimization over \(P_X\) in \eqref{eq:characterization_capacity_IB_channel}).
In this work, we will further explore the connections between the two problems.

In particular, we show that the helper in the WAK problem can be viewed as an oblivious relay.
Further, good codes for the WAK problem can be produced by permuting codes developed for the IB channel.
To demonstrate the above connection, we construct a code for the WAK problem by permuting the code developed  for the IB channel in Theorem 1, and show that it attains the best known achievable error exponent of the WAK problem, previously established by Kelly and Wagner \cite[Theorem 1]{kellyReliabilitySourceCoding2012}.
\begin{theorem}
    \label{thm:source_coding_coded_side_information_exponent_lower_bound}
    For the DMS pair \((X^n, Y^n)\), we have
    \begin{align}
        E_{\textnormal{h}}(R, B) & \geq  \min_{ \substack{  Q_{Y} } }  \max_{P_{U|Y}}  \min_{ \substack{Q_{X|YU}: \\ H(Q_X) \geq R } }  D(Q_{XY}  \| P_{XY}) + I_Q(X;U|Y) + \nonumber \\
        & \hspace{7cm} \big | R- H_Q(X|U) - |I_Q(Y;U) - B |^{+}  \big| ^{+}.
    \end{align}
\end{theorem}
\begin{proof}
    See Section \ref{sec:source_coding_coded_side_information_exponent_lower_bound}.
\end{proof}
We discuss a difficulty encountered when producing codes for the WAK problem through permutations.
Ahlswede and Dueck \cite{ahlswedeGoodCodesCan1982} employed Ahlswede's covering lemma to design the encoder \(f^{\prime}_n\) for the SW problem, which is effectively a sequence of permutations.
A key technique in their proof is to adapt the receiver's decoding regions to the permuted codebooks, i.e., the codebook \(f_n^{\prime}(m)^{-1}\) (see equation (31) in \cite{ahlswedeGoodCodesCan1982}).
However, this technique cannot be directly applied to the WAK problem, because here the side information \(Y^n\) is compressed by an oblivious helper that has no knowledge of \(f_n^{\prime}(m)^{-1}\), i.e., the helper cannot adapt its compress-forward strategy to the permuted codebook \(f_n^{\prime}(m)^{-1}\).

To address this issue, we will revisit and extend Ahlswede's covering lemma to show that a type class can be simultaneously covered by several distinct sets using a single sequence of permutations.
This new simultaneous covering result enables us to find a good encoder \(f_n^{\prime}\), i.e., a sequence of permutations, for the WAK problem that can cope with the lack of adaptability at the helper.
The connection between the IB channel and the WAK problem shows that good codes can still be produced through permutations even if the coordination of the permuting process is disrupted at an intermediate node.
\begin{remark}[Mismatched decoding]
Since Theorem \ref{thm:source_coding_coded_side_information_exponent_lower_bound} is obtained through employing the coding scheme developed in Theorem \ref{thm:lower_bound}, by specializing the \(\alpha\)-decoder, we can immediately derive an achievable error exponent and rate for the WAK problem under mismatched decoding rules. As far as we are aware, these have not been derived before.
See Theorem \ref{thm:source_coding_coded_side_information_mismatch_exponent} and Corollary \ref{cor:source_coding_coded_side_information_mismatch_rate} in Section \ref{sec:source_coding_coded_side_information_mismatch}.
\end{remark}
\section{Achievable Error Exponent and Rate}
\label{sec:lower_bound_proof}
In this section, we present a coding scheme for the IB channel and establish an achievable error exponent, leading to proving Theorem \ref{thm:lower_bound} and Corollary \ref{cor:achievable_rate}.
Consider an arbitrary auxiliary alphabet \(\mathcal{U}\).
In the coding scheme we present, the relay and receiver will share a common codebook with codewords selected from the set \(\mathcal{U}^n\).
They use this codebook for compress-forward at the relay and to decode at the receiver. In particular, the relay assigns a codeword \(\bm{u} \in \mathcal{U}^n\) to every received channel output \(\bm{y} \in \mathcal{Y}^n\), while the receiver uses \(\bm{u}\) to decide which message is sent.\footnote{In case the receiver gets the index of the bin containing $\bm{u}$, it will decode $\bm{u}$ and the message jointly.}
To distinguish it from the channel codebook, we call the codebook shared between the relay and receiver the \emph{bottleneck codebook}, and denote it by \(\mathcal{B}_n\).
Loosely speaking, this can also be thought of as a \emph{quantization} codebook.

The bottleneck codebook \(\mathcal{B}_n\) is constructed through the well-known type covering lemma, presented below.
The type covering lemma is originally due to Berger \cite{bergerRateDistortionTheory1971}. The version we adopt here appears in other literature, e.g., \cite[Lemma 3.34]{moserAdvancedTopicsInformation2024}.
For a joint pmf \(Q_{YU}\), we will write \(Q_Y\) and \(Q_U\) for its marginal distributions, as well as \(Q_{Y|U}\) and \(Q_{U|Y}\) for its conditional distributions, when there is no ambiguity.
\begin{lemma}
\label{lemma:type_covering}
For every joint type \(Q_{YU} \in \mathcal{P}_n(\mathcal{Y} \times \mathcal{U})\), there exists a subset \(\mathcal{A}_n \subset \mathcal{T}_n(Q_U)\) with
\begin{equation}
    |\mathcal{A}_n| \ndot{\leq} e^{nI(Q_Y, Q_{U|Y})}
\end{equation}
such that for every \(\bm{y} \in \mathcal{T}_n(Q_Y)\) we can find a \(\bm{u} \in \mathcal{A}_n\) satisfying \(\hat{P}_{\bm{y}  \bm{u}} = Q_{YU}\).
\end{lemma}
\begin{proof}
This follows by modifying the proof of \cite[Lemma 9.1]{csiszarInformationTheoryCoding2011}, while considering sequences with the exact joint type instead of jointly typical sequences. See \cite[Lemma 3.34]{moserAdvancedTopicsInformation2024}.
\end{proof}
\subsection{Bottleneck Codebook}
\label{sec:lower_bound_bottleneck_codebook_construction}
The bottleneck codebook \(\mathcal{B}_n\) comprises an array of (sub) codebooks \(\mathcal{B}_n = \{\mathcal{B}_n(Q_{Y})\}_{Q_{Y} \in \mathcal{P}_n(\mathcal{Y})}\), or simply written as \(\{\mathcal{B}_n(Q_{Y})\} \), where \(\mathcal{B}_n(Q_{Y})\) is used for observed channel outputs of type \(Q_{Y} \in \mathcal{P}_n(\mathcal{Y})\).
That is, depending on the type \(Q_{Y}\) of the observed channel output,  the relay adopts different codebooks \(\mathcal{B}_n(Q_{Y})\) for compress-forward.
The bottleneck codebook \(\{\mathcal{B}_n(Q_{Y})\} \) is constructed as follows.
\begin{enumerate}
    \item For every type \(Q_{Y} \in \mathcal{P}_n(\mathcal{Y})\), we select a conditional type \(P_{U|Y} \in \mathcal{P}_n(\mathcal{U} | \mathcal{Y})\). Note that \(P_{U|Y}\) can vary for different \(Q_{Y}\), and hence when necessary, we write \(P_{U|Y}\) as \(P_{U|Y, Q_{Y}}\) to emphasize this dependence. Denote by \(P_{Y|U}\) the reverse conditional type induced by \(Q_Y\) and \(P_{U|Y}\). In the same fashion, we write this as \(P_{Y|U, Q_Y}\) when necessary.
    \item For every pair \((Q_{Y}, P_{U|Y})\), we select a set \(\mathcal{A}_n(Q_Y)\) according to Lemma \ref{lemma:type_covering}, i.e., \(\mathcal{A}_n(Q_Y)\) covers the entire type class \(\mathcal{T}_n(Q_Y)\) under \(P_{U|Y}\).
    \item For every type \(Q_{Y}\), we partition \(\mathcal{A}_n(Q_Y)\) into \(e^{nB}\) subsets (bins) of roughly equal size. The arrangement of elements from \(\mathcal{A}_n(Q_Y)\) into bins is arbitrary. Bins are denoted by \(\mathsf{B}_i\), \(i \in [e^{nB}]\), and
    \begin{equation}
        \abs{\mathsf{B}_i} \ndot{ \leq } e^{n|I(Q_{Y}, P_{U|Y}) - B |^{+} } \qquad \forall i \in [e^{nB}], \label{eq:lower_bound_bin_size}
    \end{equation}
    where the operation \(|a|^{+} \) is introduced due to the possible scenario that the size of \(\mathcal{A}_n(Q_Y)\) is asymptotically less than \(e^{nB}\), i.e., \(I( Q_Y, P_{U|Y} ) < B\). In this case, a bin may contain a single sequence. The codebook is chosen to be the collection of the bins, i.e., \(\mathcal{B}_n(Q_{Y}) = ( \mathsf{B}_1, \mathsf{B}_2, \ldots, \mathsf{B}_{e^{nB}})\).
\end{enumerate}
\subsection{Encoding and Decoding}
\label{sec:lower_bound_encoding_decoding}
Given message \(M = m\) and codebook \(\bm{C} = \mathcal{C}_n\), the transmitter sends codeword \(\bm{x}(m)\) from \(\mathcal{C}_n\).
After receiving a channel output \(\bm{Y} = \bm{y}\), the relay first examines the type of \(\bm{y}\) and determines the bottleneck codebook \(\mathcal{B}_n(\hat{P}_{\bm{y}})\) to be used.
Compress-forward at the relay then proceeds as follows.
\begin{enumerate}
    \item The relay searches through the entire \(\mathcal{B}_n(\hat{P}_{\bm{y}})\) and identifies a codeword \(\bm{u} \in \mathcal{B}_n(\hat{P}_{\bm{y}})\) such that \(\bm{y} \in \mathcal{T}_n(P_{Y|U, \hat{P}_{\bm{y}}} | \bm{u} )\), where we recall that \(P_{Y|U, \hat{P}_{\bm{y}}}\) denotes the reverse conditional type selected for the type \(\hat{P}_{\bm{y}}\) when constructing \(\mathcal{B}_n(\hat{P}_{\bm{y}})\). Since we construct the bottleneck codebooks under Lemma \ref{lemma:type_covering}, the existence of such codeword is guaranteed.
    \item If multiple candidates \(\bm{u}\)  satisfy \(\bm{y} \in \mathcal{T}_n(P_{Y|U, \hat{P}_{\bm{y}}} | \bm{u} )\), the relay selects one of them arbitrarily.
    \item The relay sends the index of the bin that contains \(\bm{u}\), i.e., it sends \(l \in [e^{nB}]\) if \(\bm{u} \in \mathsf{B}_l\).
\end{enumerate}
The relay also describes the type \(\hat{P}_{\bm{y}}\) to the receiver by sending another index besides \(l\).
Since there are at most \((1+n)^{\abs{\mathcal{Y}}}\) possible types \(Q_Y\) (i.e., a polynomial number in \(n\)), including the type index does not break the rate limit \(B\) asymptotically.

With knowledge of \(\hat{P}_{\bm{y}}\), the receiver knows that \(\mathcal{B}_n(\hat{P}_{\bm{y}})\) is used by the relay.
Given a forwarded index \(l\), it also knows that  the codeword \(\bm{u}\) covering the channel output \(\bm{y}\) is from the bin \(\mathsf{B}_l\) inside \(\mathcal{B}_n(\hat{P}_{\bm{y}})\).
Combining this with knowledge of the channel codebook \(\bm{C} = \mathcal{C}_n\), it decides that message \(\hat{m}\) is sent if
\begin{equation}
    \hat{m} = \argmax_{ \bm{x}(m) \in \mathcal{C}_n, \bm{u} \in \mathsf{B}_l } g(\hat{P}_{\bm{x}(m)}, \hat{P}_{\bm{u} | \bm{x}(m)}),
\end{equation}
where \(g: \mathcal{P}(\mathcal{X}\times \mathcal{U}) \to \mathbb{R} \) is a fixed continuous function
known as an \(\alpha\)-decoder \cite{csiszarGraphDecompositionNew1981} or generalized decoder.
In other words, the receiver searches through the entire codebook \(\mathcal{C}_n\) and bin \(\mathsf{B}_l\); identifies the unique pair \((\bm{x}(\hat{m}), \bm{u})\) that maximizes \(g(\hat{P}_{\bm{x}(m)}, \hat{P}_{\bm{u} | \bm{x}(m)})\); and decides \(\hat{m}\) is sent.
Examples of \(g\) include
\begin{equation}
    g(\hat{P}_{\bm{x}(m)}, \hat{P}_{\bm{u} | \bm{x}(m)}) = \sum_{x, u} \hat{P}_{\bm{x}(m) \bm{u}}(x,u) \log q(x, u)
\end{equation}
for some decoding metric \(q(x, u)\), commonly known as  the mismatched decoder under \(q(x,u)\); and
\begin{equation}
    g(\hat{P}_{\bm{x}(m)}, \hat{P}_{\bm{u} | \bm{x}(m)}) = I(\hat{P}_{\bm{x}(m)}, \hat{P}_{\bm{u} | \bm{x}(m)})
\end{equation}
is the maximum empirical mutual information (MMI) decoder.
\subsection{Error Probability}
Suppose that \(M = 1\) and \(\bm{C} = \mathcal{C}_n\), and hence \(\bm{x}(1) \in \mathcal{C}_n\) is sent.
Let the channel output received at the relay be \(\bm{y}\), and thus the bottleneck codebook for compress-forward is \(\mathcal{B}_n(\hat{P}_{\bm{y}})\).
Denote by \(\bm{u}(\bm{y})\) the sequence selected at the relay, i.e., \(\bm{y} \in \mathcal{T}_n(P_{Y|U, \hat{P}_{\bm{y}}} | \bm{u}(\bm{y}) )\).
Let \(l\) be the index forwarded to the receiver, and hence \(\bm{u}(\bm{y}) \in \mathsf{B}_l\) within \(\mathcal{B}_n(\hat{P}_{\bm{y}})\).
Since we use constant composition codes with codeword type \(P_X\), given the index \(l \in [e^{nB}]\), the receiver seeks  \(\bm{x}(\hat{m}) \in \mathcal{C}_n\) and  \(\bm{u} \in \mathsf{B}_l\) that maximize \(g(P_X, \hat{P}_{\bm{u} | \bm{x}(m)})\).
A decoding error occurs if and only if there exists some \(\bm{u}^{\prime} \in \mathsf{B}_l\) and \(\bm{x}(j)\in \mathcal{C}_n\) with \(j \neq 1\) such that
\begin{equation}
    \label{eq:lower_bound_error_exponenet_error_event}
    g(P_X, \hat{P}_{\bm{u}^{\prime} | \bm{x}(j) }  ) \geq \max_{\bm{u} \in \mathsf{B}_l }g(P_X, \hat{P}_{\bm{u} | \bm{x}(1) }  ),
\end{equation}
because the right hand side of \eqref{eq:lower_bound_error_exponenet_error_event} is the maximum value of \(g(P_X, \hat{P}_{\bm{u}^{\prime} | \bm{x}(1) }  )\) over the entire bin \(\mathsf{B}_l\) for \(\bm{x}(1)\).
Due to \(\bm{u}(\bm{y}) \in \mathsf{B}_l\), by relaxing the maximum over the entire bin, if a decoding error occurs at the receiver, then we must have
\begin{equation}
\label{eq:relaxed error event}
    g(P_X, \hat{P}_{\bm{u}^{\prime}| \bm{x}(j) }  ) \geq g(P_X, \hat{P}_{\bm{u}( \bm{y}) | \bm{x}(1) } )
\end{equation}
for some \(\bm{u}^{\prime} \in \mathsf{B}_l\) and \(j \neq 1\).
As a result, we consider a channel output \(\bm{y}\) to be ``erroneous'' if its covering sequence \(\bm{u}(\bm{y})\) and forwarded index \( l = \varphi_n(\bm{y})\) satisfy \eqref{eq:relaxed error event} for some \( \bm{u}^{\prime} \in \mathsf{B}_{\varphi_n(\bm{y})}\) and \(j \neq 1\),
as these include channel outputs at the relay that can possibly lead to a decoding error at the receiver.

We analyze the probability of this relaxed ``error'' event, which naturally provides an upper bound on the true decoding error probability of the coding scheme.
The relaxed ``error'' region of \(\bm{x}(1)\) regarding the other codeword \(\bm{x}(j) \in \mathcal{C}_n\) with \(j \neq 1\) is defined as
\begin{equation}
    \mathcal{Y}^n[\bm{x}(1), \bm{x}(j)] \triangleq \{ \bm{y} \in \mathcal{Y}^n :  \exists \bm{u}^{\prime} \in \mathsf{B}_{\varphi_n(\bm{y})},  g(P_X, \hat{P}_{\bm{u}^{\prime} | \bm{x}(j)}) \geq g (P_X, \hat{P}_{\bm{u}( \bm{y}) | \bm{x}(1) } ) \}. \label{eq:lower_bound_decoding_error_region}
\end{equation}
Thus, for message \(M =1 \), the ensemble-average decoding error probability  is upper bounded as
\begin{align}
    & \lambda_1(n, R, B) \nonumber \\
    & \leq \E_{\bm{C}}\Big[ \sum_{\bm{y} \in \mathcal{Y}^n} P_{Y|X}^n(\bm{y} | \bm{X}(1)) \times \mathbbmss{1} \Big\{ \bm{y} \in \bigcup_{j \neq 1} \mathcal{Y}^n [\bm{X}(1), \bm{X}(j) ] \Big\} \Big] \\
    & = \E_{\bm{X}(1)}\Big[ \sum_{\bm{y} \in \mathcal{Y}^n} P_{Y|X}^n(\bm{y} | \bm{X}(1)) \times \P \Big\{ \bm{y} \in \bigcup_{j \neq 1} \mathcal{Y}^n [\bm{X}(1), \bm{X}(j) ] \Big\} \Big] \\
    & \leq \E_{\bm{X}(1)}\Big[ \sum_{\bm{y} \in \mathcal{Y}^n} P_{Y|X}^n(\bm{y} | \bm{X}(1)) \times \min \Big\{ 1, e^{nR} \times \P \big\{ \bm{y} \in \mathcal{Y}^n [\bm{X}(1), \bm{X}(2) ] \big| \bm{X}(1) \big\}  \Big\}  \Big] \label{eq:lower_bound_error_exponent_two_codewords_error_probability} \\
    & = \E_{\bm{X}(1)}\Big[ \sum_{Q_Y \in \mathcal{P}_n(\mathcal{Y})} \sum_{\bm{y} \in \mathcal{T}_n(Q_Y)} P_{Y|X}^n(\bm{y} | \bm{X}(1)) \times \min \Big\{ 1, e^{nR} \times \P \big\{ \bm{y} \in \mathcal{Y}^n [\bm{X}(1), \bm{X}(2) ] \big| \bm{X}(1) \big\}  \Big\}  \Big], \label{eq:lower_bound_error_exponent_two_codewords_error_probability_conditional_type}
\end{align}
where \eqref{eq:lower_bound_error_exponent_two_codewords_error_probability} follows from the truncated union bound and independent generation of codewords under the same distribution, i.e., for any fixed \(\bm{X}(1) = \bm{x}(1)\), it holds that
\begin{equation}
    \P \Big\{ \bm{y} \in \bigcup_{j \neq 1} \mathcal{Y}^n [\bm{x}(1), \bm{X}(j) ] \Big\} \leq \min \left\{  1, e^{nR} \times \P \left\{ \bm{y} \in \mathcal{Y}^n [\bm{x}(1), \bm{X}(2) ] \right\}  \right\}.
\end{equation}
Given any fixed \(\bm{x}(1)\) and \(\bm{y}\), the probability \(\P \left \{ \bm{y} \in \mathcal{Y}^n [\bm{x}(1), \bm{X}(2) ]  \right \}\) results from the random generation of codeword \(\bm{X}(2)\).
Hence, we see that
\begin{align}
    & \P \left \{ \bm{y} \in \mathcal{Y}^n [\bm{x}(1), \bm{X}(2) ]  \right \} \nonumber \\
    & = \P \{ \bm{X}(2): \exists \bm{u}^{\prime} \in \mathsf{B}_{\varphi_n(\bm{y})},  g(P_X, \hat{P}_{\bm{u}^{\prime} | \bm{X}(2)}) \geq g (P_X, \hat{P}_{\bm{u}( \bm{y}) | \bm{x}(1) } )   \} \\
    & \leq \sum_{\bm{u}^{\prime} \in \mathsf{B}_{\varphi_n(\bm{y})}} \P \{ \bm{X}(2): g(P_X, \hat{P}_{\bm{u}^{\prime} | \bm{X}(2)}) \geq g (P_X, \hat{P}_{\bm{u}( \bm{y}) | \bm{x}(1) } )    \}, \label{eq:lower_bound_error_exponent_two_codeword_union_bound_over_bins}
\end{align}
where \eqref{eq:lower_bound_error_exponent_two_codeword_union_bound_over_bins} is due to the union bound over the bin.
Since \(\bm{X}(2)\) is uniformly distributed over the type class \(\mathcal{T}_n(P_X)\), it follows that for any \(\bm{u}^{\prime} \in \mathsf{B}_{\varphi_n(\bm{y})}\), we have
\begin{align}
    &  \P \{ \bm{X}(2): g(P_X, \hat{P}_{\bm{u}^{\prime} | \bm{X}(2)}) \geq g (P_X, \hat{P}_{\bm{u}( \bm{y}) | \bm{x}(1) } )  \} \nonumber \\
    & = \sum_{ \substack{ Q_{U|X}:  \ P_X \cdot Q_{U|X} = \hat{P}_{\bm{u}^{\prime}},  \\ g(P_X, Q_{U|X}) \geq g (P_X, \hat{P}_{\bm{u}( \bm{y}) | \bm{x}(1) } ) } } \frac{ \sum_{\bm{x} \in \mathcal{T}_n(P_X) } \idc{  \bm{u}^{\prime} \in \mathcal{T}_n(Q_{U|X} | \bm{x}) } }{ |\mathcal{T}_n(P_X)| } \label{eq:lower_bound_error_exponent_uniform_distribution_codeword} \\
    & \leq \sum_{ \substack{ Q_{U|X}:  \  P_X \cdot Q_{U|X} = \hat{P}_{\bm{u}^{\prime}},  \\ g(P_X, Q_{U|X}) \geq g (P_X, \hat{P}_{\bm{u}( \bm{y}) | \bm{x}(1) } ) } } (n+1)^{\abs{\mathcal{X} } } e^{-nI(P_X, Q_{U|X})} \label{eq:lower_bound_error_exponent_reverse_conditional_type} \\
    & \ndot{=} \max_{ \substack{ Q_{U|X}:  \  P_X \cdot Q_{U|X} = \hat{P}_{\bm{u}^{\prime}},  \\ g(P_X, Q_{U|X}) \geq g (P_X, \hat{P}_{\bm{u}( \bm{y}) | \bm{x}(1) } ) } } e^{-n I(P_X, Q_{U|X} ) } \\
    & = e^{-n E_0(P_X, \hat{P}_{\bm{u}( \bm{y}) | \bm{x}(1) } ) }, \label{eq:lower_bound_error_exponent_bound_E_0}
\end{align}
where in \eqref{eq:lower_bound_error_exponent_uniform_distribution_codeword} we only need to consider conditional types \(Q_{U|X}\) such that \(P_X \cdot Q_{U|X} = \hat{P}_{\bm{u}^{\prime}}\), due to the fixed \(\bm{u}^{\prime}\) and constant composition codewords \(\bm{X}(2)\); \eqref{eq:lower_bound_error_exponent_reverse_conditional_type} can be seen from considering the reverse conditional type \(Q_{X|U}\) induced by \(P_X\) and \(Q_{U|X}\); in \eqref{eq:lower_bound_error_exponent_bound_E_0} for any pair \((P_X, P_{U|X})\) we define
\begin{equation}
    E_0(P_X, P_{U|X}) \triangleq \min_{ \substack{ Q_{U|X}:  \ P_X \cdot Q_{U|X} = P_X \cdot P_{U|X}, \\ g(P_X, Q_{U|X}) \geq g(P_X, P_{U|X} ) }  } I ( P_X, Q_{U|X}), \label{eq:lower_bound_error_exponent_generalized_decoder_E_0}
\end{equation}
and notice that \(\hat{P}_{\bm{u}^{\prime}} = \hat{P}_{\bm{u}(\bm{y})} = P_X \cdot \hat{P}_{\bm{u}( \bm{y}) | \bm{x}(1) }\).
Because the upper bound in \eqref{eq:lower_bound_error_exponent_bound_E_0} holds for any \(\bm{u}^{\prime}\in \mathsf{B}_{\varphi(\bm{y})}\), it follows that given \(\bm{y} \in \mathcal{T}_n(Q_Y)\),
\begin{align}
    & \min \big\{ 1, e^{nR} \times \! \! \sum_{\bm{u}^{\prime} \in \mathsf{B}_{\varphi_n(\bm{y})}} \! \P \{ \bm{X}(2): g(P_X, \hat{P}_{\bm{u}^{\prime} | \bm{X}(2)}) \geq \hat{P}_{\bm{u}( \bm{y}) | \bm{x}(1) }   \}  \big\} \nonumber \\
    & \ndot{\leq} \min \left\{ 1, e^{nR} \times \exp{-n( E_0(P_X, \hat{P}_{\bm{u}( \bm{y}) | \bm{x}(1) }  )  - |I(Q_Y, P_{U|Y}) - B |^{+} )} \right\}\label{eq:lower_bound_error_exponent_E_0_bin_size}\\
    & = \exp{-n \big | E_0(P_X, \hat{P}_{\bm{u}( \bm{y}) | \bm{x}(1) }  ) -R - |I(Q_Y, P_{U|Y}) - B |^{+}  \big| ^{+} }, \label{eq:lower_bound_error_exponent_compress_forward_exponent }
\end{align}
where in \eqref{eq:lower_bound_error_exponent_E_0_bin_size} we consider the upper bound in \eqref{eq:lower_bound_error_exponent_bound_E_0} and then the sum over the bin in \eqref{eq:lower_bound_error_exponent_two_codeword_union_bound_over_bins} is reduced to a product with the size of \(\mathsf{B}_{\varphi_n(\bm{y})}\), i.e., \(\exp{|I(Q_Y, P_{U|Y}) - B |^{+}  }\).
Thus, substituting \eqref{eq:lower_bound_error_exponent_compress_forward_exponent } back into \eqref{eq:lower_bound_error_exponent_two_codewords_error_probability_conditional_type} yields that for any fixed \(\bm{X}(1) = \bm{x}(1)\),
\begin{align}
    & \sum_{Q_Y \in \mathcal{P}_n(\mathcal{Y})} \sum_{\bm{y} \in \mathcal{T}_n(Q_Y)} P_{Y|X}^n(\bm{y} | \bm{x}(1)) \times \min \left\{ 1, e^{nR} \times \P \left\{ \bm{y} \in \mathcal{Y}^n [\bm{x}(1), \bm{X}(2) ] \right\}  \right\} \nonumber \\
    & \ndot{\leq} \sum_{Q_Y \in \mathcal{P}_n(\mathcal{Y})} \sum_{\bm{y} \in \mathcal{T}_n(Q_Y)} P_{Y|X}^n(\bm{y} | \bm{x}(1)) \times \nonumber \\
    & \hspace{5cm} \exp{-n \big | E_0(P_X, \hat{P}_{\bm{u}( \bm{y}) | \bm{x}(1) }  ) -R - |I(Q_Y, P_{U|Y}) - B |^{+}  \big| ^{+} } \\
    & = \sum_{Q_Y \in \mathcal{P}_n(\mathcal{Y})}   \sum_{\tilde{\bm{u}} \in \mathcal{A}_n(Q_Y)} \sum_{ \substack{\bm{y} \in \mathcal{T}_n(Q_Y) : \\ \bm{u}(\bm{y}) = \tilde{\bm{u}}} } P_{Y|X}^n(\bm{y} | \bm{x}(1)) \times \nonumber \\
    & \hspace{5cm} \exp{-n \big | E_0(P_X, \hat{P}_{ \tilde{\bm{u}} | \bm{x}(1) }  ) -R - |I(Q_Y, P_{U|Y}) - B |^{+}  \big| ^{+} } \label{eq:lower_bound_error_probability_codebook_construction} \\
    & = \sum_{Q_Y \in \mathcal{P}_n(\mathcal{Y})}  \sum_{\tilde{\bm{u}} \in \mathcal{A}_n(Q_Y)} \sum_{  \substack{ Q_{Y|X}: \\ P_X \cdot Q_{Y|X} = Q_Y } } \sum_{ \substack{\bm{y} \in \mathcal{T}_n(Q_{Y|X} | \bm{x}(1)) : \\ \bm{u}(\bm{y}) = \tilde{\bm{u}} } } P_{Y|X}^n(\bm{y} | \bm{x}(1)) \times \nonumber \\
    & \hspace{5cm} \exp{-n \big | E_0(P_X, \hat{P}_{ \tilde{\bm{u}} | \bm{x}(1) }  ) -R - |I(Q_Y, P_{U|Y}) - B |^{+}  \big| ^{+} }, \label{eq:lower_bound_error_probability_sum_over_Q_Y_X}
\end{align}
where in \eqref{eq:lower_bound_error_probability_codebook_construction} we recall that the bottleneck codebook \(\mathcal{B}_n(Q_Y)\) is constructed through \(\mathcal{A}_n(Q_Y)\) and hence we must have \(\bm{u}(\bm{y}) = \tilde{\bm{u}}\) for a certain \(\tilde{\bm{u}} \in \mathcal{A}_n(Q_Y)\).

Notice that the inner term in \eqref{eq:lower_bound_error_probability_sum_over_Q_Y_X}, i.e.,
\begin{equation}
    \exp{-n \big | E_0(P_X, \hat{P}_{ \tilde{\bm{u}} | \bm{x}(1) }  ) -R - |I(Q_Y, P_{U|Y}) - B |^{+}  \big| ^{+} } \nonumber,
\end{equation}
which can be understood as the receiver's decoding error probability conditioned on $\bm{y} \in \mathcal{T}_n(Q_{Y|X} | \bm{x}(1))$ and $\bm{u}(\bm{y}) = \tilde{\bm{u}}$, only depends on \(\tilde{\bm{u}}\) and \(Q_Y\) (recall that \(P_{U|Y}\) is preselected for \(Q_Y\)).
Thus, by pulling it out of the two innermost sums in \eqref{eq:lower_bound_error_probability_sum_over_Q_Y_X},
we are interested in the probability
\begin{equation}
    \sum_{  \substack{ Q_{Y|X}: \\ P_X \cdot Q_{Y|X} = Q_Y } } \sum_{ \substack{\bm{y} \in \mathcal{T}_n(Q_{Y|X} | \bm{x}(1)) : \\ \bm{u}(\bm{y}) = \tilde{\bm{u}} } } P_{Y|X}^n(\bm{y} | \bm{x}(1)), \nonumber
\end{equation}
i.e., the probability of channel outputs \(\bm{y} \in \mathcal{T}_n(Q_Y)\) resulting in \(\bm{u}(\bm{y}) = \tilde{\bm{u}}\).
Since \(\bm{u}(\bm{y}) = \tilde{\bm{u}}\) occurs only if \(\bm{y} \in \mathcal{T}_n(P_{ Y | U} | \tilde{\bm{u}})\) (recall that \(P_{Y|U}\) is the reverse conditional type selected for \(Q_Y\)), it holds that
\begin{equation}
    \{  \bm{y} \in \mathcal{T}_n(Q_{Y|X} | \bm{x}(1)) :  \bm{u}(\bm{y}) = \tilde{\bm{u}}  \} \subset  \mathcal{T}_n(Q_{Y|X} | \bm{x}(1)) \cap \mathcal{T}_n(P_{ Y | U} | \tilde{\bm{u}}).
\end{equation}
Therefore, the cardinality of this set is bounded as
\begin{align}
    \big|  \{  \bm{y} \in \mathcal{T}_n(Q_{Y|X} | \bm{x}(1)) : \bm{u}(\bm{y}) = \tilde{\bm{u}}  \}  \big| & \leq   \abs{  \mathcal{T}_n(Q_{Y|X} | \bm{x}(1)) \cap \mathcal{T}_n(P_{ Y | U} | \tilde{\bm{u}})  } \\
    & \leq \sum_{ Q^{\prime}_{XYU} } e^{nH_{Q^{\prime}}(Y|XU)} \label{eq:lower_bound_intersection_two_conditional_class}
\end{align}
where \eqref{eq:lower_bound_intersection_two_conditional_class} follows from \cite[Problem 2.10]{csiszarInformationTheoryCoding2011}, and the joint type \(Q_{XYU}^{\prime}\)
must satisfy
\begin{equation}
    Q_{XY}^{\prime} = \hat{P}_{\bm{x}(1)} \times Q_{Y|X} = P_X \times Q_{Y|X}, \qquad Q_{YU}^{\prime} = \hat{P}_{\tilde{\bm{u}} } \times P_{Y|U} = Q_Y \times P_{U|Y}
\end{equation}
as well as
\begin{equation}
    Q_{XU}^{\prime} = \hat{P}_{\bm{x}(1)\tilde{\bm{u}} } = P_X \times \hat{P}_{\tilde{\bm{u}} | \bm{x}(1)},
\end{equation}
in which we recall that \( \hat{P}_{\tilde{\bm{u}} } = Q_Y \cdot P_{U|Y}\) and \(P_{Y|U}\) is the reverse conditional type.
On the other hand, for every \(\bm{y} \in \mathcal{T}_n(Q_{Y|X} | \bm{x}(1)) \), we have
\begin{equation}
    P_{Y|X}^n( \bm{y} | \bm{x}(1) ) = \exp{-n(D(Q_{Y|X} \| P_{Y|X} | P_X) + H(Q_{Y|X} | P_X))}.
\end{equation}
Consequently, we see that
\begin{align}
    & \sum_{  \substack{ Q_{Y|X}: \\ P_X \cdot Q_{Y|X} = Q_Y } } \sum_{ \substack{\bm{y} \in \mathcal{T}_n(Q_{Y|X} | \bm{x}(1)) : \\ \bm{u}(\bm{y}) = \tilde{\bm{u}} } } P_{Y|X}^n(\bm{y} | \bm{x}(1)) \nonumber \\
    & \leq \sum_{  \substack{ Q_{Y|X}: \\ P_X \cdot Q_{Y|X} = Q_Y } }\sum_{ Q_{XYU}^{\prime} }\exp{-n(D(Q_{Y|X} \| P_{Y|X} | P_X) + H(Q_{Y|X} | P_X) - H_{Q^{\prime}}(Y|XU))} \\
    & = \sum_{ Q_{XYU} }\exp{-n(D(Q_{Y|X} \| P_{Y|X} | P_X) + H(Q_{Y|X}  | P_X) -H_{Q}(Y|XU) )} \label{eq:lower_bound_intersection_set_probability} \\
    & = \sum_{ Q_{XYU} }\exp{-n(D(Q_{Y|X} \| P_{Y|X} | P_X ) + I_Q(Y;U | X) ) } \\
    & \ndot{ = } \max_{ Q_{XYU} } \exp{-n(D(Q_{Y|X} \| P_{Y|X} | P_X ) + I_Q(Y;U | X) ) }, \label{eq:lower_bound_sum_over_Q_XYU}
\end{align}
where in \eqref{eq:lower_bound_intersection_set_probability} we combine the two sums, i.e., we drop the restriction  \(Q_{XY}^{\prime} = P_X \times Q_{Y|X}\) and the sum over \(Q_{XYU}\) now consists of all \(Q_{XYU}\) satisfying
\begin{equation}
    Q_{XU} = P_X \times \hat{P}_{ \tilde{\bm{u}} | \bm{x}(1)} \quad \text{and}  \quad Q_{YU} = Q_Y\times P_{U|Y}. \label{eq:lower_bound_Q_XYU_conditions}
\end{equation}
Incorporating \eqref{eq:lower_bound_sum_over_Q_XYU} into \eqref{eq:lower_bound_error_probability_sum_over_Q_Y_X}, we exponentially upper bound \eqref{eq:lower_bound_error_probability_sum_over_Q_Y_X} by
\begin{align}
    \sum_{Q_Y \in \mathcal{P}_n(\mathcal{Y})}  \sum_{\tilde{\bm{u}} \in \mathcal{A}_n(Q_Y)} e^{-nf(\hat{P}_{\tilde{\bm{u}} | \bm{x}(1)}  ) }, \label{eq:lower_bound_probability_of_u_times_error}
\end{align}
where to shorten notation, we define
\begin{align}
    f(\hat{P}_{\tilde{\bm{u}} | \bm{x}(1)}  )  \triangleq \min_{Q_{XYU}}  D(Q_{Y|X}  \| P_{Y|X} | P_X) + I_Q(Y;U|X) + \big | E_0(P_X, \hat{P}_{ \tilde{\bm{u}} | \bm{x}(1) }  ) -R - |I_Q(Y;U) - B |^{+}  \big| ^{+}, \nonumber
\end{align}
in which \(Q_{XYU}\) satisfies \eqref{eq:lower_bound_Q_XYU_conditions}.
Substituting \eqref{eq:lower_bound_probability_of_u_times_error} into \eqref{eq:lower_bound_error_exponent_two_codewords_error_probability_conditional_type}, we obtain
\begin{align}
    \lambda_1(n, R, B)  & \ndot{\leq}
    \sum_{Q_Y \in \mathcal{P}_n(\mathcal{Y})} \sum_{ \tilde{\bm{u}} \in \mathcal{A}_n(Q_Y) } \E_{\bm{X}(1)} \left[  e^{-nf(\hat{P}_{\tilde{\bm{u}}| \bm{X}(1)}  ) } \right]. \label{eq:lower_bound_interchange_Q_Y_and_X_1}
\end{align}
Recall that \( \hat{P}_{\tilde{\bm{u}} } = Q_Y \cdot P_{U|Y} \triangleq Q_U\).
For every \(\tilde{\bm{u}} \in \mathcal{A}_n(Q_Y)\), the inner term can be evaluated through
\begin{align}
    \E_{\bm{X}(1)} \left[  e^{-nf(\hat{P}_{\tilde{\bm{u}}| \bm{X}(1)}  ) } \right] & = \sum_{Q_{U|X}} \P \{ \bm{X}(1)  : \hat{P}_{\tilde{\bm{u}} | \bm{X}(1)} = Q_{U|X}   \} \times e^{-nf(Q_{U|X})} \label{eq:lower_bound_probability_u_x}\\
    & \ndot{=}  \sum_{Q_{U|X}} e^{-n(I(P_X, Q_{U|X})+ f(Q_{U|X}))} \label{eq:lower_bound_probability_revert_type_u_x} \\
    & \ndot{=}  \max_{Q_{U|X}} e^{-n(I(P_X, Q_{U|X})+ f(Q_{U|X}))}, \label{eq:lower_bound_not_depend_on_u}
\end{align}
where in \eqref{eq:lower_bound_probability_u_x} \(Q_{U|X}\) satisfies \(P_X \cdot Q_{U|X} = Q_U\); in \eqref{eq:lower_bound_probability_revert_type_u_x} the probability is obtained by considering the reverse conditional type \(Q_{X|U}\).
Observe that \eqref{eq:lower_bound_not_depend_on_u} does not depend on \(\tilde{\bm{u}}\).
Therefore, we proceed with
\begin{align}
    \lambda_1(n, R, B)  & \ndot{\leq} \sum_{Q_Y \in \mathcal{P}_n(\mathcal{Y})} \sum_{ \tilde{\bm{u}} \in \mathcal{A}_n(Q_Y) } \E_{\bm{X}(1)} \left[  e^{-nf(\hat{P}_{\tilde{\bm{u}}| \bm{X}(1)}  ) } \right]\\
    & \ndot{=} \sum_{Q_Y \in \mathcal{P}_n(\mathcal{Y})} |\mathcal{A}_n(Q_Y)| \times  \max_{Q_{U|X} } e^{-n(I(P_X, Q_{U|X})+ f(Q_{U|X}))} \\
    & \ndot{\leq} \sum_{Q_Y \in \mathcal{P}_n(\mathcal{Y})} e^{nI(Q_Y, P_{U|Y})} \times  \max_{Q_{U|X} } e^{-n(I(P_X, Q_{U|X})+ f(Q_{U|X}))} \\
    & = \sum_{Q_Y \in \mathcal{P}_n(\mathcal{Y})} \max_{Q_{U|X} } e^{-n(I(P_X, Q_{U|X}) - I(Q_Y, P_{U|Y}) + f(Q_{U|X}))} \label{eq:lower_bound_expurgation_u_0_bound}
\end{align}
Recall that by definition, we have
\begin{align}
    f( Q_{U|X}  )  = \min_{Q_{XYU}}  D(Q_{Y|X}  \| P_{Y|X} | P_X) + I_Q(Y;U|X) + \big | E_0(P_X, Q_{U|X}  ) -R - |I_Q(Y;U) - B |^{+}  \big| ^{+}, \label{eq:lower_bound_f_Q_U_X}
\end{align}
where \(Q_{XYU}\) satisfies $Q_{XU} = P_X \times Q_{U|X} $ and $Q_{YU} = Q_Y\times P_{U|Y}$.
Substituting \eqref{eq:lower_bound_f_Q_U_X} into \eqref{eq:lower_bound_expurgation_u_0_bound}, we get
\begin{align}
    & \lambda_1(n, R, B) \nonumber \\
    & \ndot{ \leq } \sum_{Q_Y \in \mathcal{P}_n(\mathcal{Y})} \max_{Q_{U|X} } e^{-n(I(P_X, Q_{U|X}) - I(Q_Y, P_{U|Y}) + f(Q_{U|X}))} \\
    & = \sum_{Q_Y \in \mathcal{P}_n(\mathcal{Y})} \max_{Q_{U|X} } \max_{Q_{XYU}} \exp \big\{ -n \big(  D(Q_{Y|X}  \| P_{Y|X} | P_X) + I_Q(X;U|Y) + \nonumber \\
    & \hspace{6.5cm} \big | E_0(P_X, Q_{U|X}  ) -R - |I_Q(Y;U) - B |^{+}  \big| ^{+}   \big) \big\} \label{eq:lower_bound_mutual_information_identity} \\
    & = \sum_{Q_Y \in \mathcal{P}_n(\mathcal{Y})} \max_{Q_{XYU}} \exp \big\{ -n \big(  D(Q_{Y|X}  \| P_{Y|X} | P_X) + I_Q(X;U|Y) + \nonumber \\
    & \hspace{6.5cm} \big | E_0(P_X, Q_{U|X}  ) -R - |I_Q(Y;U) - B |^{+}  \big| ^{+}   \big) \big\} \label{eq:lower_bound_combine_two_maximization} \\
    & = \sum_{Q_Y \in \mathcal{P}_n(\mathcal{Y})} \min_{P_{U|Y}} \max_{Q_{XYU}} \exp \big\{ -n \big(  D(Q_{Y|X}  \| P_{Y|X} | P_X) + I_Q(X;U|Y) + \nonumber \\
    & \hspace{6.5cm} \big | E_0(P_X, Q_{U|X}  ) -R - |I_Q(Y;U) - B |^{+}  \big| ^{+}   \big) \big\} \label{eq:lower_bound_optimize_over_P_U_Y} \\
    & \ndot{=} \max_{Q_Y \in \mathcal{P}_n(\mathcal{Y})} \min_{P_{U|Y}} \max_{Q_{XYU}} \exp \big\{ -n \big(  D(Q_{Y|X}  \| P_{Y|X} | P_X) + I_Q(X;U|Y) + \nonumber \\
    & \hspace{6.5cm} \big | E_0(P_X, Q_{U|X}  ) -R - |I_Q(Y;U) - B |^{+}  \big| ^{+}   \big) \big\} \label{eq:lower_bound_binning_error_bound_Q_XYU} \\
    & = \max_{Q_Y \in \mathcal{P}_n(\mathcal{Y})} \min_{P_{U|Y}}  \max_{ \substack{Q_{X|YU} : \\ Q_X = P_X}   } \exp \big\{ -n \big(  D(Q_{Y|X}  \| P_{Y|X} | P_X) + I_Q(X;U|Y) + \nonumber \\
    & \hspace{6.5cm} \big | E_0(P_X, Q_{U|X}  ) -R - |I_Q(Y;U) - B |^{+}  \big| ^{+}   \big) \big\} \label{eq:lower_bound_binning_error_bound}
\end{align}
where \eqref{eq:lower_bound_mutual_information_identity} is due to the following identity
\begin{equation}
\label{eq:E_MI_term}
    I(X;U) -I(Y;U) + I(Y;U|X) = I(X;U|Y);
\end{equation}
in \eqref{eq:lower_bound_combine_two_maximization} we combine the two maximizations, i.e.,  now the maximization is over all \(Q_{XYU} \) satisfying
\begin{equation}
    Q_{X} = P_X \quad \text{and}  \quad Q_{YU} = Q_Y\times P_{U|Y};
\end{equation}
in \eqref{eq:lower_bound_optimize_over_P_U_Y} we assume the optimal \(P_{U|Y}\) is selected for every \(Q_Y \in \mathcal{P}_n(\mathcal{Y})\) when constructing the bottleneck codebook;
\eqref{eq:lower_bound_binning_error_bound_Q_XYU} and \eqref{eq:lower_bound_binning_error_bound} are the same but expressed differently, i.e., in \eqref{eq:lower_bound_binning_error_bound} we consider all joint distributions \(Q_{XYU} = Q_Y \times P_{U|Y} \times Q_{X|YU}\) satisfying \(Q_X = P_X\).

We conclude the above analysis with some observations and remarks as follows.
\begin{enumerate}
    \item Using different bottleneck codebooks for different types \(Q_{Y}\) allows us to select \(P_{U|Y}\) depending on \(Q_{Y}\) to minimize the overall decoding error probability, yielding the sandwiched minimization in \eqref{eq:lower_bound_binning_error_bound}.
    \item The term \(I_Q(X;U|Y)\) follows from \eqref{eq:E_MI_term}, restated as
    $-I_Q(Y;U) + I_Q(X;U) + I_Q(Y;U|X)$, where \(I_Q(Y;U)\) is due to the total number of sequences required for type covering under \((Q_Y, P_{U|Y})\); \(I_Q(X;U)\) is due to the probability of the event \(\hat{P}_{ \tilde{\bm{u}} | \bm{X}(1)} = Q_{U|X}\); and \(I_Q(Y;U|X)\) is caused by the probability of channel outputs \(\bm{y} \in \mathcal{T}_n(Q_{Y|X} | \bm{x}(1))\) satisfying \(\bm{u}(\bm{y}) = \tilde{\bm{u}}\).
    Intuitively speaking, \(I_Q(X;U|Y)\) captures the correlation between the sent codeword \(\bm{x}(1)\) and the compress-forward sequence \(\bm{u}(\bm{y})\).
    The more informative \(\bm{u}(\bm{y})\) is towards \(\bm{x}(1)\), the less likely the receiver will make a decoding error, which results in a better achievable error exponent.
    \item We make use of binning when constructing the bottleneck codebooks and employ the union bound over the bin in the analysis. The bin size \(|I(Q_{Y}, P_{U|Y}) - B|^{+}\) hence appears in the exponent.
\end{enumerate}
\begin{remark}
    If we do not make use of binning when constructing the bottleneck codebooks, i.e., only considering \(P_{U|Y}\) such that \(I(Q_Y, P_{U|Y}) \leq B\), then we obtain an achievable exponent of
    \begin{equation}
    \min_{Q_Y \in \mathcal{P}_n(\mathcal{Y})} \max_{ \substack{P_{U|Y} : \\ I(Q_Y, P_{U|Y}) \leq B } } \min_{ \substack{Q_{X|YU} : \\ Q_X = P_X}  }  D(Q_{Y|X}  \| P_{Y|X} | P_X) + I_Q(X;U|Y) + | E_0(P_X, Q_{U|X}  ) -R   |^{+}.   \label{eq:lower_bound_error_exponent_without_binning}
    \end{equation}
    Binning in the compress-forward scheme enables us to take \(\{ P_{U|Y} : I(Q_Y, P_{U|Y}) > B\}\) into account, which produces a generally better error exponent in \eqref{eq:lower_bound_binning_error_bound} (compared to \eqref{eq:lower_bound_error_exponent_without_binning}).
   The binning scheme used here has its roots in the classical Wyner-Ziv scheme \cite{wynerRatedistortionFunctionSource1976}. The idea of utilizing binning in tandem with covering to achieve better error exponents dates back at least to \cite{shimokawaErrorBoundHypothesis1994}.
   It was also adopted in, e.g., Kelly and Wagner \cite{kellyReliabilitySourceCoding2012} and Tan \cite{tanReliabilityFunctionDiscrete2015} later on.
   In particular, both papers also employed a decoder that considers the maximization over an entire bin, which is similar to the one used here.
\end{remark}
\begin{remark}
 A common approach in the literature on multiterminal lossy source coding is to randomly generate
 \(e^{nI(Q_Y, P_{U|Y})}\) sequences for compress-forward, see, e.g., \cite{kellyReliabilitySourceCoding2012}.
    Here, by adopting the type covering lemma for compress-forward,  we effectively separate the two phases: the random sequence generation phase for covering  and the error probability analysis phase.
    Thus, when analyzing the decoding error probability, we can avoid considering error events arising from the random generation.
    Instead, the error exponent is established through investigating the intersection between conditional type classes.
\end{remark}
\subsection{Error Exponent}
The upper bound on the ensemble-average decoding error probability we just derived holds for any message \(m \in [e^{nR}]\), not necessarily \(m = 1\).
Therefore, we obtain
\begin{equation}
    \liminf_{n \to \infty} -\frac{1}{n} \log \bar{\lambda}(n, R, B) \geq E_{\textnormal{r}}(R, B, P_X, g),
\end{equation}
where we have
\begin{align}
    E_{\textnormal{r}}(R, B,P_X, g) & \triangleq \min_{Q_Y } \max_{P_{U|Y}} \min_{\substack{Q_{X|YU}: \\ Q_X = P_X}}  D(Q_{Y|X}  \| P_{Y|X} | P_X) + I_Q(X;U|Y) + \nonumber \\
    & \hspace{6cm} \big | E_0(P_X, Q_{U|X}  ) -R - |I_Q(Y;U) - B |^{+}  \big| ^{+}, \label{eq:lower_bound_error_exponent_generalized_decoder_exponent}
\end{align}
in which \(Q_{XYU} = Q_Y \times P_{U|Y} \times Q_{X|YU}\) satisfies \(Q_X = P_X\).
Thus, by optimizing over \(P_X\) and generalized decoders \(g\), we conclude that
\begin{equation}
    E(R,B) \geq \max_{P_X}\max_{g}E_{\textnormal{r}}(R, B, P_X, g).
\end{equation}
Before solving \(\max_{g}E_{\textnormal{r}}(R, B, P_X, g)\), we first have a look at \(\max_{g}E_0(P_X, Q_{U|X} )\).
Recall that
\begin{equation}
    E_0(P_X, Q_{U|X} ) = \min_{ \substack{ Q_{U|X}^{\prime} : P_X \cdot Q_{U|X}^{\prime} = P_X \cdot Q_{U|X}, \\ g(P_X, Q_{U|X}^{\prime}) \geq g(P_X, Q_{U|X}  ) }  } I ( P_X, Q_{U|X}^{\prime}). \label{eq:lower_bound_error_exponent_special_case_E_0}
\end{equation}
Consequently, we have
\begin{equation}
    E_0(P_X, Q_{U|X} ) \leq I(P_X, Q_{U|X} ),
\end{equation}
since we can choose \(Q_{U|X}^{\prime} = Q_{U|X} \).
The equality is achieved if \(g(P_X, Q_{U|X}) = I(P_X, Q_{U|X})\), i.e., if the MMI decoder is adopted.
Hence, it is evident that the MMI decoder is the optimal \(\alpha\)-decoder, i.e.,
\begin{equation}
    \max_g E_{\textnormal{r}}(R, B, P_X, g) = E_{\textnormal{r}}(R, B, P_X),
\end{equation}
which proves Theorem \ref{thm:lower_bound}.

\subsection{Achievable Rate}
\label{sec:proof_achievable_rate_corollary}
Here we prove Corollary \ref{cor:achievable_rate}.
Define
\begin{equation}
    R_0 \triangleq \max_{P_X, P_{U|Y}} I(P_X, P_{U|X}) \qquad \text{s.t.} \quad I(P_{Y}, P_{U|Y}) \leq B,
\end{equation}
where \(X \overset{P_{Y|X}}{\to} Y \overset{P_{U|Y}}{\to} U\) forms a Markov chain.
Assume \((P_X^\ast, P_{U|Y}^\ast)\) achieves  \(R_0\).
Hence, $R_0 = I(P_X^\ast, P_{Y|X} \cdot P_{U|Y}^{\ast})$ and $I(P_{Y}^{\ast}, P_{U|Y}^{\ast}) \leq B$.
We need to show that all rates up to \(R_0\) are achievable, i.e., for all rates \(R < R_0\), we have $\max_{P_X} E_{\textnormal{r}}(R, B, P_X) > 0$.
Recall that
\begin{align}
    E_{\textnormal{r}}(R, B, P_X) & \triangleq \min_{Q_Y } \max_{P_{U|Y}}  \min_{ \substack{Q_{X|YU}: \\ Q_X = P_X} }   D(Q_{Y|X}  \| P_{Y|X} | P_X) + I_Q(X;U|Y) + \nonumber \\
    & \hspace{6cm} \big | I_Q(X;U) -R - |I_Q(Y;U) - B |^{+}  \big| ^{+},
\end{align}
where \(Q_{XYU} = Q_Y \times P_{U|Y} \times Q_{X|YU}\) satisfies \(Q_X = P_X\).
For all rates \(R< R_0\), we definitely have \(\max_{P_X} E(R, B, P_X) > 0\) if we can show that the following inequality holds
\begin{equation}
    \min_{Q_Y } \min_{ \substack{ Q_{X|YU}: \\ Q_X = P_X^{\ast}} }   D(Q_{Y|X}  \| P_{Y|X} | P_X^{\ast}) + I_Q(X;U|Y) + \big | I(P_X^{\ast}, Q_{U|X}) -R - |I(Q_Y, P_{U|Y}^{\ast}) - B |^{+}  \big| ^{+} >0,
\end{equation}
where \(Q_{XYU} = Q_Y \times P_{U|Y}^{\ast} \times Q_{X|YU}\) satisfies \(Q_X = P_X^{\ast}\).
The rest of the proof is reminiscent of a similar proof in \cite{merhavGeneralizedStochasticLikelihood2017}.
Consider the identity
\begin{equation}
    |a|^{+} = \max_{\rho \in [0,1]}\rho a, \ \text{where} \ a \in \mathbb{R}.
\end{equation}
Hence, it suffices to show that for all \(R < R_0\), we have
\begin{equation}
    \min_{Q_{XYU} } \max_{\rho \in [0,1]}  D(Q_{Y|X}  \| P_{Y|X} | P_X^{\ast}) + I_Q(X;U|Y) + \rho \big ( I(P_X^{\ast}, Q_{U|X}) -R - |I(Q_Y, P_{U|Y}^{\ast}) - B |^{+} \big) >0, \label{eq:achievable_rates_sufficiency}
\end{equation}
where \(Q_{XYU} = Q_Y \times P_{U|Y}^{\ast} \times Q_{X|YU}\) satisfies \(Q_X = P_X^{\ast}\).
\eqref{eq:achievable_rates_sufficiency} states that for every such \(Q_{XYU}\), there exists a \(\rho \in [0,1]\) such that
\begin{equation}
    D(Q_{Y|X}  \| P_{Y|X} | P_X^{\ast}) + I_Q(X;U|Y) + \rho \big ( I(P_X^{\ast}, Q_{U|X}) -R - |I(Q_Y, P_{U|Y}^{\ast}) - B |^{+}  \big) >0,
\end{equation}
i.e.,
\begin{equation}
    R <  \frac{ D(Q_{Y|X} \| P_{Y|X} | P_{X}^{\ast}) + I_Q(X;U|Y) }{ \rho }   +   I(P_X^{\ast}, Q_{U|X}) - |I(Q_Y, P_{U|Y}^{\ast}) - B |^{+}. \label{eq:achievable_rates_R_upper_bound_rho}
\end{equation}
Thus, \eqref{eq:achievable_rates_sufficiency} is equivalent to
\begin{align}
    R & < \min_{Q_{XYU}} \max_{\rho \in [0,1]}  \frac{ D(Q_{Y|X} \| P_{Y|X} | P_{X}^{\ast}) + I_Q(X;U|Y) }{ \rho }   +   I(P_X^{\ast}, Q_{U|X}) - |I(Q_Y, P_{U|Y}^{\ast}) - B |^{+} \label{eq:achievable_rates_R_upper_bound_optimize} \\
    & = I(P_X^{\ast}, P_{Y|X} \cdot P_{U|Y}^{\ast} ) - |I(P_{Y}^{\ast}, P_{U|Y}^{\ast}) - B  |^{+} \label{eq:achievable_rates_Q_Y_X_equal_P_Y_X} \\
    & = R_0,
\end{align}
where in \eqref{eq:achievable_rates_R_upper_bound_optimize} \(Q_{XYU} = Q_Y \times P_{U|Y}^{\ast} \times Q_{X|YU}\) satisfies \(Q_X = P_X^{\ast}\), while the minimization over \(Q_{XYU}\)  is because \eqref{eq:achievable_rates_R_upper_bound_rho} holds for every such \(Q_{XYU}\) and the maximization over \(\rho \in [0,1]\) is due to the existence of such \(\rho\);
\eqref{eq:achievable_rates_Q_Y_X_equal_P_Y_X} holds since the minimization in \eqref{eq:achievable_rates_R_upper_bound_optimize} is achieved when \(Q_{XYU}\) satisfies \(Q_{Y|X} = P_{Y|X}\) as well as \(I_Q(X;U|Y) = 0\), i.e., \(P_{U|X} = P_{Y|X} \cdot P_{U|Y}^{\ast}\) and \(Q_Y = P_{Y}^{\ast}\), due to the maximization over \(\rho \in [0,1]\).
Therefore, \eqref{eq:achievable_rates_sufficiency} indeed holds for all \(R < R_0\) since the two are equivalent, which completes the proof.

\subsection{Mismatched Decoding}
\label{sec:lower_bound_error_exponent_mistmathced_decoder}
Dikshtein \etal \cite{dikshteinMismatchedObliviousRelaying2023} considered the problem of oblivious relaying under a mismatched decoding rule.
In such problem, the receiver is required to reconstruct a certain sequence \(\bm{u} \in \mathcal{U}^n\) for every forwarded index \(l\) from the relay, and decode under a mismatched decoder, i.e.,
\begin{equation}
    \hat{m} = \argmax_{ \hat{x}(m) \in \mathcal{C}_n, \bm{u} } g(\hat{P}_{\bm{x}(m)}, \hat{P}_{\bm{u} | \bm{x}(m)}),
\end{equation}
where
\begin{equation}
    g(\hat{P}_{\bm{x}(m)}, \hat{P}_{\bm{u} | \bm{x}(m)} ) = \sum_{x,u} \hat{P}_{\bm{x}(m) \bm{u}}(x,u)\log q(x, u),
\end{equation}
for some decoding metric \(q(x, u)\).
Therefore, we have this immediate result.
\begin{theorem}
    \label{thm:lower_bound_mismatch}
    For the IB channel \((P_{Y|X}, B)\) under a mismatched decoding rule, we have
    \begin{equation}
        \liminf_{n \to \infty} -\frac{1}{n} \log \bar{\lambda}(n, R, B) \geq \max_{P_X} E_{\textnormal{r}}(R, B, P_X, g),
    \end{equation}
    where \(E_{\textnormal{r}}(R, B, P_X, g)\) is given by \eqref{eq:lower_bound_error_exponent_generalized_decoder_exponent} and \eqref{eq:lower_bound_error_exponent_special_case_E_0}.
\end{theorem}
Following the proof of Corollary \ref{cor:achievable_rate}, it can be verified that this exponent recovers the following achievable rate provided in \cite[Theorem 1]{dikshteinMismatchedObliviousRelaying2023}.
\begin{corollary}
    \label{cor:lower_bound_rate_mismatch}
    For the IB channel \((P_{Y|X}, B)\) under a mismatched decoding rule, all rates up to \(C_{\textnormal{LM}}(B)\) are achievable, where
    \begin{equation}
        C_{\textnormal{LM}}(B) = \max_{P_X, P_{U|Y}} E_0(P_X, P_{U|X}) \qquad \text{s.t.} \qquad  I(P_{Y}, P_{U|Y}) \leq B,
    \end{equation}
    in which \(X \overset{P_{Y|X}}{\to} Y \overset{P_{U|Y}}{\to} U\) forms a Markov chain and \(E_0(P_X, P_{U|X})\) is given by \eqref{eq:lower_bound_error_exponent_generalized_decoder_E_0}.
\end{corollary}
The proof of \cite[Theorem 1]{dikshteinMismatchedObliviousRelaying2023} (i.e., \cite[Appendix B]{dikshteinMismatchedObliviousRelaying2023}) relies on joint typicality and does not incorporate binning. On the other hand, to establish the achievable error exponent in Theorem \ref{thm:lower_bound_mismatch}, we use an improved scheme that employs binning and more refined analysis based on the method of types. Nevertheless, the resulting LM rate in Corollary \ref{cor:lower_bound_rate_mismatch} is identical to the one in \cite{dikshteinMismatchedObliviousRelaying2023}.

\section{Converse}
\label{sec:weak_converse_for_capacity}
This section is dedicated to the proof of Theorem \ref{thm:weak_converse_constant_composition_ensemble}.
Before starting, we first introduce some definitions and notation that will be used down the line.
\subsection{Definitions and Notation}
For convenience, in this section we write \(x^n = (x_1, x_2, \ldots, x_n)\) for a deterministic sequence from \(\mathcal{X}^n\) and \(X^n = (X_1, X_2, \ldots, X_n)\) for a random sequence.
Given a certain type \(P_X \in \mathcal{P}_n(\mathcal{X})\), define the following distribution on $\mathcal{X}^n$
\begin{equation}
    P_{X^n}(x^n) \triangleq \frac{\idc{x^n \in \mathcal{T}_n(P_X)}}{|\mathcal{T}_n(P_X)|}. \label{eq:weak_converse/definition_ensemble_distribution}
\end{equation}
Then, every \(X^n(i)\)  in the constant composition ensemble $\bm{C} = (X^n(1), X^n(2), \ldots, X^n(e^{nR}))$ with codeword composition \(P_X\) independently follows the same distribution \(P_{X^n}\).

Given any distribution \(P_X\) on a finite set \(\mathcal{X}\), we denote its support by \(\supp(P_X)\), i.e., $\supp(P_X) = \{x \in \mathcal{X}: P_X(x)>0\}$.
For a sequence \(x^n = (x_1, x_2, \ldots, x_n)\) and \(i \in [n]\), we call the subsequence \(x^i = (x_1, x_2, \ldots, x_i)\) its prefix and the remaining subsequence \(x^{n}_{i+1} = (x_{i+1}, x_{i+2}, \ldots, x_n)\) its suffix.
We denote the type of its prefix \(x^{i}\) by \(\hat{P}_{x^i}\) and the type of its suffix \(x_{i+1}^n\) by \(\hat{P}_{x^{n}_{i+1}}\).
For every \(x^n \in \mathcal{T}_n(P_X)\) and \(i \in [n]\), it is clear that
\begin{equation}
    i\hat{P}_{x^i}(a) + (n-i)\hat{P}_{x^{n}_{i+1}}(a) = nP_X(a), \qquad \forall a \in \mathcal{X}. \label{eq:weak_converse_constant_ensemble_Q_hat}
\end{equation}
We denote by \(\mathcal{S}_i(\mathcal{X})\) the set of all possible prefix types \(\hat{P}_{x^{i}}\) under the condition \(x^n \in \mathcal{T}_n(P_X)\).
Hence, we have \(\mathcal{S}_i(\mathcal{X}) \subseteq \mathcal{P}_i(\mathcal{X})\).
Note that the set of all possible suffix types \(\hat{P}_{x^{n}_{i+1} }\) in \( \mathcal{T}_n(P_X)\) is the same as \(\mathcal{S}_{n-i}(\mathcal{X})\), since any \(\hat{P}_{x^{n}_{i+1} }\)   can also be a prefix type \(\hat{P}_{x^{n-i} }\).

Given a constant \(\delta \geq 0\) and pmf \(P_X\), we write \(Q_X \overset{\delta}{\sim} P_X\) if
\begin{align}
    \abs{Q_X(a) - P_X(a)} \leq \delta P_X(a), \qquad \forall a \in \mathcal{X}.
\end{align}
We say a sequence \(x^n\) is \(P_X\)-typical with \(\delta\) if its type \(\hat{P}_{x^n}\) satisfies \(\hat{P}_{x^n} \overset{\delta}{\sim} P_X\).
The set of typical sequences \(x^n \in \mathcal{X}^n\) is denoted by \(\mathcal{T}_n^{\delta}(P_X)\).
The notion of typicality adopted here is known as robust typicality \cite{gamalNetworkInformationTheory2011}.
The reason for not using, e.g., strong typicality \cite{csiszarInformationTheoryCoding2011}, will be clear further on.

\subsection{Preliminaries}
We now begin the proof of Theorem \ref{thm:weak_converse_constant_composition_ensemble}.
Consider a sequence of \((n, R, B)\)-codes, or equivalently a sequence of mappings \((f_n, \varphi_n, \phi_n)\) as defined in Section \ref{section:problem_setting}, satisfying \(\bar{\lambda}(n, R, B) \to 0\).
Conditioned on \(\bm{C} = \mathcal{C}_n\), where
\(\mathcal{C}_n = (x^n(1),\ldots,x^n(e^{nR}))\), we write
$x^n(M) \triangleq f_n(M,\mathcal{C}_n)$. Recall that $M$ is uniform on $[e^{nR}]$.
The codeword \(x^n(M)\) passes through the DMC \(P_{Y|X}\) to reach the relay.
Let the random output at the relay be \(Y^n\), and denote by \(L\) the index forwarded from the relay to the receiver, i.e.,   \(L = \varphi_n(Y^n)\).
At the decoder side, we write the estimated message as $\hat{M} = \phi_n(L,\mathcal{C}_n) $.
Thus, conditioned on a codebook  \(\bm{C} = \mathcal{C}_n\), we have the Markov chain
\begin{equation}
    M \to x^n(M) \to Y^n \to L \to \hat{M}.
\end{equation}
From Fano's inequality, conditioned on any codebook \(\bm{C} = \mathcal{C}_n\), we have
\begin{equation}
    H(M|L, \bm{C} = \mathcal{C}_n) \leq H(M|\hat{M}, \bm{C} = \mathcal{C}_n) \leq 1 + \bar{\lambda}(n, R, B, \mathcal{C}_n)nR, \label{eq:Fano_inequality}
\end{equation}
where \(\bar{\lambda}(n, R, B, \mathcal{C}_n)\) is the average decoding error probability of codebook \(\mathcal{C}_n\) and the first inequality is due to the chain rule.
After averaging over the ensemble \(\bm{C}\), we obtain
\begin{equation}
    H(M|L, \bm{C})  \leq 1 + \bar{\lambda}(n , R, B) nR \triangleq n \epsilon_n. \label{eq:Fano_inequality_ensemble}
\end{equation}
To proceed, we first follow the footsteps of the converse proof in \cite[Theorem 2]{sanderovichCommunicationDecentralizedProcessing2008} and write
\begin{align}
    nR & = H(M) \\
    &=I(M;L, \bm{C}) + H(M | L, \bm{C})\\
    &\leq I(M;L, \bm{C}) + n\epsilon_n \label{eq:weak_converse_use_fano_inequality}\\
    &= I(M ; \bm{C}) + I(M;L | \bm{C}) + n \epsilon_n \\
    &= I(M ;L | \bm{C}) + n \epsilon_n  \label{eq:appendix_weak_converse_IID_codebook_message_independence}\\
    &\leq I(M, \bm{C} ; L) + n\epsilon_n \\
    &\leq I(X^n(M) ;L) + n\epsilon_n \label{eq:appendix_weak_converse_IID_chain_rule} \\
    &=I(X^{n};L) + n \epsilon_n  \label{eq:appendix_weak_converse_IID_codeword_distribution} \\
    &=H(X^{n})- H(X^{n}|L) + n\epsilon_n \\
    &\leq \sum_{i=1}^{n} \left( H(X_i) - H(X_i|L, X^{i-1}) \right)+n\epsilon_n \\
    &\leq \sum_{i=1}^{n} \left( H(X_i) - H(X_i|L, Y^{i-1}, X^{i-1}) \right)+n\epsilon_n \\
    &=  \sum_{i=1}^{n} I(X_i;L, Y^{i-1}, X^{i-1}) + n \epsilon_n,
\end{align}
where \eqref{eq:weak_converse_use_fano_inequality} follows from \eqref{eq:Fano_inequality_ensemble}; \eqref{eq:appendix_weak_converse_IID_codebook_message_independence} is because the random ensemble is independent of the message, i.e., \(I(M;\bm{C}) = 0\);
\eqref{eq:appendix_weak_converse_IID_chain_rule} is due to the chain rule, in which \(X^n(M)\) is the random codeword due to the random message as well as the random ensemble \(\bm{C}\);
in \eqref{eq:appendix_weak_converse_IID_codeword_distribution} we notice that \(X^n(M)\) follows the same distribution as \(X^n\), i.e., \(X^n (M) \sim P_{X^n}\) (recall the definition in \eqref{eq:weak_converse/definition_ensemble_distribution}).
On the other hand, we have
\begin{align}
    nB & \geq H(L) \\
    &\geq I(L;Y^{n}, X^{n}) \\
    &= \sum_{i=1}^{n} I(L;Y_i, X_i | Y^{i-1}, X^{i-1}) \\
    &\geq \sum_{i=1}^{n} I(L;Y_i| Y^{i-1}, X^{i-1}).
\end{align}

Our proof will divert from the proof of \cite[Theorem 2]{sanderovichCommunicationDecentralizedProcessing2008} from now on.
The reason for this is the different prior distribution of codebooks we selected to model the obliviousness.
In \cite{sanderovichCommunicationDecentralizedProcessing2008}, the IID ensemble with codeword distribution \(P_X^n\) is considered, while we consider the constant composition ensemble.
In our case, \(X^n\) follows the distribution \(P_{X^n}\) rather than the IID distribution \(P_X^n\), so \(X^{i-1}\) and \(X^{i}\) are not independent of each other.
Therefore, under the constant composition ensemble, we cannot assert that \(Y_i\) is independent of \((Y^{i-1}, X^{i-1})\) and that the Markov chain \(X_i \to Y_i \to (L, Y^{i-1}, X^{i-1})\) holds, which are key steps of the proof in \cite{sanderovichCommunicationDecentralizedProcessing2008}.
As a result, we are unable to proceed in the standard manner of identifying an auxiliary random variable and Markov chain.
To address this issue, we will investigate the behavior of the conditional distribution \(P_{X_i | X^{i-1}}\) under \(P_{X^n}\).
As a result, we establish several properties for \(P_{X^n}\) and \(P_{X_i | X^{i-1}}\) which are essential for our converse proof, and may also be of independent interest.

\subsection{Properties of the Constant Composition Distribution}
\label{sec:weak_converse_properties_constant_composition_ensemble}
The first property of \(P_{X^n}\) concerns its marginal distributions, which has appeared in, e.g., \cite[Lemma 5.9]{moserAdvancedTopicsInformation2024} in a different context.
Here, we provide a different proof from the one in
\cite{moserAdvancedTopicsInformation2024}.
\begin{lemma}[Marginal Distribution]
    \label{lemma:marginal_ditribution_constant_composition_ensemble}
    The marginal distribution of \(P_{X^n}\) satisfies \(P_{X_i} = P_X\) for every \(i \in [n]\).
\end{lemma}
\begin{proof}
See Appendix \ref{appendix:proof_marginal_distribution_lemma}.
\end{proof}
The next result looks into the conditional distribution \(P_{X_{i+1}|X^{i}}\) under joint distribution \(P_{X^n}\).
\begin{lemma}[Conditional Distribution]
    \label{lemma:conditional_distribution_constant_composition_ensemble}
    Under \(P_{X^n}\) and for every  \(i \in [n]\), the marginal prefix distribution \(P_{X^{i}}\) is supported on the set of \(x^i\) satisfying that there exists a suffix type \(Q_X^{\ast} \in \mathcal{S}_{n-i}(\mathcal{X})\) such that
    \begin{equation}
        i\hat{P}_{x^{i} }(a) + (n-i)Q_X^{\ast}(a) = nP_X(a), \qquad \forall a \in \mathcal{X}.
    \end{equation}
    Further, given a prefix \(x^i \in \supp(P_{X^i})\) with its corresponding suffix type being \(Q_X^{\ast}\), we have
    \begin{equation}
        P_{X_{i+1}|X^{i}}(a|x^i) = Q_X^{\ast}(a), \qquad \forall a \in \mathcal{X}.
    \end{equation}
\end{lemma}
\begin{proof}
    See Appendix \ref{appendix:proof_conditional_distribution_lemma}.
\end{proof}

The following immediate corollary of Lemma \ref{lemma:conditional_distribution_constant_composition_ensemble} reveals the behavior of \(P_{X_{i+1} | X^{i}}\) for certain \(x^i\).
\begin{corollary}[Almost Independent]
    \label{cor:typical_conditional_distribution_constant_composition_ensemble}
    Given any \(i \in [n]\) and \(\delta \geq 0\), and for every prefix \(x^{i} \in \supp(P_{X^i})\) whose suffix type \(Q_X^{\ast}\) satisfies \( Q_X^{\ast} \overset{\delta}{\sim} P_X\), we have \(P_{X_{i+1} |X^i }(\cdot| x^{i}) \overset{\delta}{\sim} P_X \).
\end{corollary}

If a prefix \(x^i \in \supp(P_{X^i})\) is such that \( Q_X^{\ast} \overset{\delta}{\sim} P_X\), i.e., its suffix  belongs to \( \mathcal{T}_{n-i}^{\delta}(P_X)\), then Corollary \ref{cor:typical_conditional_distribution_constant_composition_ensemble} shows that the conditional distribution \(P_{X_{i+1} | X^{i}}( \cdot | x^{i})\) will behave similarly to \(P_X\), i.e., almost independent.
Therefore, it is of interest to know the probability of such prefix \(x^i\) under \(P_{X^n}\), which we investigate next.

\begin{lemma}[Typical Subsequence]
    \label{lemma:weak_converse_constant_ensemble_window_typical}
    Under \(P_{X^n}\), for every \(\delta > 0 \),  \(i \in [n]\) and \(k \in [n-i+1]\), we have
    \begin{equation}
        \P\{X^n \in \mathcal{T}_n(P_X): X^{k+i-1}_{k} \notin \mathcal{T}_i^{\delta}(P_X)\} \leq 2\abs{\mathcal{X}}e^{\abs{\mathcal{X}}\log(n+1) - i\delta^2 P_{\min}^2},
    \end{equation}
    where \(P_{\min} = \min_{a \in \mathcal{X}: P_X(a)>0}P_X(a)\).
\end{lemma}
\begin{proof}
    See Appendix \ref{appendix:proof_window_typical_constant_ensemble}.
\end{proof}
\begin{remark}
    If we fix a sliding window \([k:k+i-1]\), then Lemma \ref{lemma:weak_converse_constant_ensemble_window_typical} provides an upper bound on the probability of observing a non-typical subsequence \(x_{k}^{k+i-1} = (x_k, x_{k+1}, \ldots, x_{k}^{k+i-1})\).
\end{remark}
The next corollary lower bounds the probability of sequences \(x^n \in \mathcal{T}_n(P_X)\) whose prefix \(x^i\) and suffix \(x^{n}_{i+1}\) are both \(P_X\)-typical, which follows immediately from Lemma \ref{lemma:weak_converse_constant_ensemble_window_typical} and the union bound.
\begin{corollary}
    \label{cor:weak_converse_constant_ensemble_prefix_suffix_typical}
    Under \(P_{X^n}\), for every \(i\) with \(\sqrt{n} \leq i \leq n-\sqrt{n}\), we have
    \begin{align}
        \P\{X^n \in \mathcal{T}_n(P_X): X^{i} \in \mathcal{T}_i^{\delta_n}(P_X), X^{n}_{i+1} \in \mathcal{T}_{n-i}^{\delta_n}(P_X)\}\geq 1- 4\abs{\mathcal{X}}e^{\abs{\mathcal{X}}\log(n+1) - n^{\frac{1 }{4}}P_{\min}^2},
    \end{align}
    where \(\delta_n = n^{-\frac{1}{8}}\) and \(P_{\min} = \min_{a \in \mathcal{X}: P_X(a)>0}P_X(a)\).
\end{corollary}
\begin{proof}
See Appendix \ref{appendix:proof_prefix_suffix_typical_constant_ensemble}.
\end{proof}
Corollary \ref{cor:weak_converse_constant_ensemble_prefix_suffix_typical} shows that under \(P_{X^n}\) and for \(\sqrt{n}  \leq i \leq n - \sqrt{n}\), we have a high probability to observe a sequence whose prefix \(x^i\) and suffix \(x^{n}_{i+1}\) are both \(P_X\)-typical.
It can be seen from the proof that \(\delta_n=n^{-\frac{1}{8}}\) and \(\sqrt{n}\) are selected arbitrarily, merely as an example to show the concentration of probability as \(n\) grows.
\begin{remark}
    \label{rem:weak_converse_almost_independent_both_direction}
    Corollary \ref{cor:typical_conditional_distribution_constant_composition_ensemble} is also true for the set \(\{x^n \in \mathcal{T}_n(P_X): x^{i} \in \mathcal{T}_i^{\delta_n}(P_X), x^{n}_{i+1} \in \mathcal{T}_{n-i}^{\delta_n}(P_X)\}\), i.e., it holds for both directions \(P_{X_{i+1} | X^{i}}\) and \(P_{X_{i} | X^{n}_{i+1}}\).
    Corollary \ref{cor:weak_converse_constant_ensemble_prefix_suffix_typical} reveals that this set also has high probability.
\end{remark}
\subsection{Main Proof}
\label{sec:weak_converse_Fano_inequality}
We first present an auxiliary result, which is key for establishing the converse proof.
\begin{lemma}
    \label{lemma:weak_converse_replacing}
    Consider three random variables \((X, Y, Z) \sim P_{XYZ}\) where \(P_{XYZ} = P_{X}P_{Y|X}P_{Z|YX}\).
    Assume there exists a subset \(\mathcal{E} \subset \mathcal{X}\) such that \(P_{Y|X}(\cdot | x ) \overset{\delta}{\sim} P_Y\) for every \(x \in \mathcal{E}\).
    Let \((X, \tilde{Y}, \tilde{Z}) \sim \tilde{P}_{XYZ}\) where \(\tilde{P}_{XYZ} = P_XP_YP_{Z|YX}\).
    Then, there exits a continuous function \(\epsilon:[0,1) \to \mathbb{R}\) with \(\epsilon(0) = 0\) such that
    \begin{equation}
        | H(Y| Z, X) - H(\tilde{Y} | \tilde{Z}, X) | < (\epsilon(\delta) + 1-P_{X}[\mathcal{E}]  ) \log|\mathcal{Y}|.\label{eq:weak_converse_replacing_statement}
    \end{equation}
\end{lemma}
\begin{proof}
Intuitively speaking, if \(\delta \approx 0\) (i.e., \(\epsilon(\delta) \approx 0\)) and \(P_X[\mathcal{E}] \approx 1\), then the two distributions \(P_{XYZ}\) and \(\tilde{P}_{XYZ}\) are the same, so \eqref{eq:weak_converse_replacing_statement} naturally holds.
A detailed proof is provided in Appendix \ref{appendix:weak_converse_replacing_lemma}.
\end{proof}
Now we can continue the converse proof, which relies on Remark \ref{rem:weak_converse_almost_independent_both_direction} and Lemma \ref{lemma:weak_converse_replacing}.
Remark \ref{rem:weak_converse_almost_independent_both_direction} shows that for the majority of prefixes \(x^{i-1}\) (in the sense of high probability), we have \(P_{X_{i} | X^{i-1}}(\cdot | x^{i-1}) \overset{\delta_n}{\sim} P_{X}\).
Hence, we construct a new joint distribution by replacing \(P_{X_{i} | X^{i-1}}\) with \(P_X\).
Lemma \ref{lemma:weak_converse_replacing} tells us that the two joint distributions are asymptotically the same (due to \(\delta_n = n^{-\frac{1}{8}} \to 0\) and Corollary \ref{cor:weak_converse_constant_ensemble_prefix_suffix_typical}).
Since \(X^{i}\) and \(X^{i-1}\) are independent under the constructed distribution, we can then apply the familiar converse technique on it for the auxiliary random variable and Markov chain.
Details are presented next.

Fix an arbitrary constant \(\tau \in (0,1)\).
Recall that we have previously arrived at
\begin{align}
    nR &\leq \sum_{i=1}^{n} I(X_i; L, Y^{i-1}, X^{i-1}) + n\epsilon_n, \\
    nB &\geq \sum_{i=1}^{n} I(L;Y_i| Y^{i-1}, X^{i-1}).
\end{align}
Observe that
\begin{align}
    R & \leq \frac{1}{n}\sum_{i=1}^{n} I(X_i; L, Y^{i-1}, X^{i-1}) + \epsilon_n \\
    & \leq \frac{1}{n}\sum_{i=\sqrt{n} + 1}^{n-\sqrt{n}} I(X_i; L, Y^{i-1}, X^{i-1}) + \frac{2\sqrt{n}}{n} \log \abs{\mathcal{X}} + \epsilon_n \label{eq:weak_converse_R_discarding_sqrt_n}\\
    & \leq \frac{1}{n}\sum_{i=\sqrt{n} + 1}^{n-\sqrt{n}} I(X_i; L, Y^{i-1}, X^{i-1}) + \tau + \epsilon_n \label{eq:weak_converse_R_less_than_tau}\\
    & = -\frac{1}{n}\sum_{i=\sqrt{n} + 1 }^{n-\sqrt{n}} H(X_i| L, Y^{i-1}, X^{i-1}) + \frac{1}{n}\sum_{i=\sqrt{n} +1 }^{n-\sqrt{n}} H(X_i) + \tau  + \epsilon_n,
\end{align}
where  \eqref{eq:weak_converse_R_discarding_sqrt_n} is due to \(I(X_i; L, Y^{i-1}, X^{i-1}) \leq \log \abs{\mathcal{X}}\); and in \eqref{eq:weak_converse_R_less_than_tau},  $\frac{2\sqrt{n}}{n} \log \abs{\mathcal{X}} \leq \tau$ for large \(n\).

For every \(i \in [n]\), consider the underlying distribution of the whole system
\begin{align}
    P_{X_i, Y_i, L, Y^{i-1}, X^{i-1}} &= P_{X^{i-1}} \times P_{X_i | X^{i-1}} \times P_{L, Y^{i-1}, Y_{i} | X_i, X^{i-1}},
\end{align}
where due to the DMC \(P_{Y|X}\) and the processing at the relay we have
\begin{equation}
    P_{L, Y^{i-1}, Y_{i} | X_i, X^{i-1}} = P_{Y^{i-1} | X^{i-1} } \times P_{Y_i | X_i} \times P_{L | Y^{i-1}, Y_i}.
\end{equation}
Now, define an auxiliary distribution
\begin{align}
    \tilde{P}_{X_i, Y_i, L, Y^{i-1}, X^{i-1}} \triangleq P_{X^{i-1}} \times P_{X_i} \times P_{L, Y^{i-1}, Y_{i} | X_i, X^{i-1}}.
\end{align}
As we can see, the only difference between the two distributions is the replacement of \(P_{X_i | X^{i-1}}\) with \(P_{X_i}\).
We will denote by \((\tilde{X}_i, \tilde{Y}_i, \tilde{L}_i, Y^{i-1}, X^{i-1})\) the random vector associated with \(\tilde{P}_{X_i, Y_i, L, Y^{i-1}, X^{i-1}}\), where we notice that \(X^{i-1}\) and \( Y^{i-1}\) remain unchanged after replacement.
After marginalizing over \(\tilde{Y}_i\), we obtain
\begin{align}
    P_{X_i, L, Y^{i-1}, X^{i-1}} & = P_{X^{i-1}} \times P_{X_i | X^{i-1}} \times P_{L, Y^{i-1} | X_i, X^{i-1}} \\
    \tilde{P}_{X_i, L, Y^{i-1}, X^{i-1}} & = P_{X^{i-1}} \times P_{X_i} \times P_{L, Y^{i-1} | X_i, X^{i-1}}.
\end{align}
We apply Lemma \ref{lemma:weak_converse_replacing} to the two distributions by choosing \(X\) to be \(X^{i-1}\), \(Y\) to be \(X_i\), and \(Z\) to be \((L, Y^{i-1})\).
The subset on the domain of \(X^{i-1}\), i.e., \(\supp(P_{X^{i-1}})\), is chosen to be
\begin{equation}
    \mathcal{E}_{i-1} \triangleq \{x^{i-1} \in \supp (P_{X^{i-1}}) : x^{i-1}  \in \mathcal{T}_{i-1}^{\delta_n}(P_X), x^{n}_{i} \in \mathcal{T}_{n-i+1}^{\delta_n}(P_X)\},
\end{equation}
where we consider all prefixes \(x^{i-1} \in \supp (P_{X^{i-1}}) \) such that both the prefix itself and its suffix are \(P_X\) typical (recall that under \(P_{X^n}\) every prefix has a unique suffix type).
Recall from Lemma \ref{lemma:marginal_ditribution_constant_composition_ensemble} that \(P_{X_i}\) has the same distribution as \(P_X\).
Hence, we have \(P_{X_i | X^{i-1}}(\cdot | x^{i-1}) \overset{\delta_n}{\sim} P_{X_i}\) for every \(x^{i-1} \in \mathcal{E}_{i-1}\) due to Corollary \ref{cor:typical_conditional_distribution_constant_composition_ensemble}.
On the other hand, because \(P_{X^{i-1}}\) is a marginal distribution of \(P_{X^n}\), it is clear that
\begin{align}
    P_{X^{i-1} }[ \mathcal{E}_{i-1} ] = P_{X^n} [x^n \in \mathcal{T}_n(P_X): x^{i-1} \in \mathcal{T}_{i-1}^{\delta_n}(P_X), x^{n}_{i} \in \mathcal{T}_{n-i+1}^{\delta_n}(P_X)],
\end{align}
i.e., \(P_{X^{i-1}}[\mathcal{E}_{i-1}] \to 1\) as \(n \to \infty\) due to Corollary \ref{cor:weak_converse_constant_ensemble_prefix_suffix_typical}.
Since \(P_{X^{i-1}}[\mathcal{E}_{i-1}] \to 1\) and \(\delta_n = n^{-\frac{1}{8}} \to 0\), from Lemma \ref{lemma:weak_converse_replacing} we see that for every \( \sqrt{n} + 1 \leq i \leq n - \sqrt{n}\),
\begin{equation}
    -H(X_i| L, Y^{i-1}, X^{i-1}) \leq -H(\tilde{X}_i| \tilde{L}_i, Y^{i-1}, X^{i-1})  + \tau,
\end{equation}
if \(n\) is sufficiently large.
Thus, for sufficiently large \(n\),
\begin{equation}
    -\frac{1}{n}\sum_{i=\sqrt{n} + 1 }^{n-\sqrt{n}} H(X_i| L, Y^{i-1}, X^{i-1}) \leq -\frac{1}{n}\sum_{i=\sqrt{n} +1 }^{n-\sqrt{n}} H(\tilde{X}_i| \tilde{L}_i, Y^{i-1}, X^{i-1}) + \tau.
\end{equation}
Therefore, we conclude that
\begin{align}
    R &\leq -\frac{1}{n}\sum_{i=\sqrt{n} +1 }^{n-\sqrt{n}} H(\tilde{X}_i| \tilde{L}_i, Y^{i-1}, X^{i-1}) + \frac{1}{n}\sum_{i=\sqrt{n} +1 }^{n-\sqrt{n}} H(X_i) +2\tau + \epsilon_n \\
    & = -\frac{1}{n}\sum_{i=\sqrt{n} +1 }^{n-\sqrt{n}} H(\tilde{X}_i| \tilde{L}_i, Y^{i-1}, X^{i-1}) + \frac{1}{n}\sum_{i=\sqrt{n} +1 }^{n-\sqrt{n}} H(\tilde{X}_i)  + 2\tau + \epsilon_n  \label{eq:weak_converse_R_same_marginal_distribution}\\
    &= \frac{1}{n}\sum_{i=\sqrt{n} +1 }^{n-\sqrt{n}} I(\tilde{X}_i; \tilde{L}_i, Y^{i-1}, X^{i-1}) + 2\tau + \epsilon_n \\
    &\leq \frac{1}{n}\sum_{i=1}^{n} I(\tilde{X}_i; \tilde{L}_i, Y^{i-1}, X^{i-1}) + 2\tau + \epsilon_n,
\end{align}
where in \eqref{eq:weak_converse_R_same_marginal_distribution}, we make use of Lemma \ref{lemma:marginal_ditribution_constant_composition_ensemble} again, i.e., noticing that \(X_i\) and \(\tilde{X}_i\) have the same distribution under \(P_{X_i, M_Y, Y^{i-1}, X^{i-1}}\) and \(\tilde{P}_{X_i, M_Y, Y^{i-1}, X^{i-1}}\) respectively.

We now turn our attention to the bound on \(B\). Notice that
\begin{align}
    B &\geq \frac{1}{n}\sum_{i=1}^{n} I(L;Y_i| Y^{i-1}, X^{i-1})  \\
    &\geq \frac{1}{n}\sum_{i=\sqrt{n} + 1 }^{n-\sqrt{n}} I(L;Y_i| Y^{i-1}, X^{i-1}) \\
    & = - \frac{1}{n}\sum_{i=\sqrt{n} + 1 }^{n-\sqrt{n}} H(Y_i| L, Y^{i-1}, X^{i-1}) + \frac{1}{n}\sum_{i=\sqrt{n} +1 }^{n-\sqrt{n}} H(Y_i | Y^{i-1}, X^{i-1}).
\end{align}
Recall the two distributions \(P_{X_i, Y_i, L, Y^{i-1}, X^{i-1}}\) and \(\tilde{P}_{X_i, Y_i, L, Y^{i-1}, X^{i-1}}\).
After marginalizing over \(X_i\), we obtain
\begin{align}
    P_{Y_i, L, Y^{i-1}, X^{i-1}} & = P_{X^{i-1}} \times P_{Y_i | X^{i-1}} \times P_{L, Y^{i-1} | Y_i, X^{i-1}} \\
    \tilde{P}_{Y_i, L, Y^{i-1}, X^{i-1}} & = P_{X^{i-1}} \times P_{Y_i} \times P_{L, Y^{i-1} | Y_i, X^{i-1}}, \label{eq:weak_converse_B_replacing}
\end{align}
where \eqref{eq:weak_converse_B_replacing} is because we notice \(Y_i\) is independent of \(X^{i-1}\) under \(\tilde{P}_{X_i, Y_i, L, Y^{i-1}, X^{i-1}}\).
Since \(P_{Y_i X_i | X^{i-1}} = P_{X_i | X^{i-1}}P_{Y_i | X_i}\), for every \(x^{i-1} \in \mathcal{E}_{i-1}\) we have \( P_{Y_i X_i | X^{i-1}} (\cdot | x^{i-1}) \overset{\delta_n}{ \sim } P_{X_i}P_{Y_i | X_i} \) and hence \(P_{Y_i | X^{i-1}}(\cdot | x^{i-1}) \overset{\delta_n}{\sim} P_{Y_i}\) through marginalizing over \(X_i\).
Thus, similarly we can assert that when \(n\) is sufficiently large,
\begin{equation}
    - \frac{1}{n}\sum_{i=\sqrt{n} + 1 }^{n-\sqrt{n}} H(Y_i| L, Y^{i-1}, X^{i-1}) \geq - \frac{1}{n}\sum_{i=\sqrt{n} + 1 }^{n-\sqrt{n}} H(\tilde{Y}_i| \tilde{L}_i, Y^{i-1}, X^{i-1}) - \tau.
\end{equation}
The same reasoning also applies to \(H(Y_i | Y^{i-1}, X^{i-1})\) with
\begin{equation}
    \frac{1}{n}\sum_{i=\sqrt{n} +1 }^{n-\sqrt{n}} H(Y_i | Y^{i-1}, X^{i-1}) \geq \frac{1}{n}\sum_{i=\sqrt{n} + 1 }^{n-\sqrt{n}} H(\tilde{Y}_i | Y^{i-1}, X^{i-1}) - \tau.
\end{equation}
Therefore, we conclude that
\begin{align}
    B & \geq - \frac{1}{n}\sum_{i=\sqrt{n} + 1 }^{n-\sqrt{n}} H(\tilde{Y}_i| \tilde{L}_i, Y^{i-1}, X^{i-1})  + \frac{1}{n}\sum_{i=\sqrt{n} + 1 }^{n-\sqrt{n}} H(\tilde{Y}_i | Y^{i-1}, X^{i-1}) - 2 \tau\\
    & = - \frac{1}{n}\sum_{i=\sqrt{n} + 1 }^{n-\sqrt{n}} H(\tilde{Y}_i| \tilde{L}_i, Y^{i-1}, X^{i-1}) + \frac{1}{n}\sum_{i=\sqrt{n} + 1}^{n-\sqrt{n}} H(\tilde{Y}_i) -2\tau \label{eq:weak_converse_B_independent_after_changing_measure} \\
    & = \frac{1}{n}\sum_{i=\sqrt{n} + 1 }^{n-\sqrt{n}} I(\tilde{Y}_i; \tilde{L}_i, Y^{i-1}, X^{i-1})   - 2\tau\\
    &\geq \frac{1}{n}\sum_{i=1}^{n} I(\tilde{Y}_i; \tilde{L}_i, Y^{i-1}, X^{i-1}) -\frac{2\sqrt{n}}{n}\log \abs{\mathcal{Y}}   - 2\tau, \label{eq:weak_converse_B_upper_bound_conditional_entropy} \\
    & \geq \frac{1}{n}\sum_{i=1}^{n} I(\tilde{Y}_i; \tilde{L}_i, Y^{i-1}, X^{i-1})   - 3\tau
\end{align}
where in \eqref{eq:weak_converse_B_independent_after_changing_measure}, we notice that \(\tilde{Y}_i\) is independent of \((Y^{i-1}, X^{i-1})\);
in \eqref{eq:weak_converse_B_upper_bound_conditional_entropy}, we make use of

\begin{equation}
    I(\tilde{Y}_i; \tilde{L}_i, Y^{i-1}, X^{i-1}) \leq H(\tilde{Y}_i) \leq \log \abs{\mathcal{Y}}.
\end{equation}

Overall, we conclude that for any \(\tau \in (0,1)\), when \(n\) is sufficiently large,
\begin{align}
    R & \leq \frac{1}{n}\sum_{i=1}^{n} I(\tilde{X}_i; \tilde{L}_i, Y^{i-1}, X^{i-1}) + 2\tau + \epsilon_n \\
    B & \geq \frac{1}{n}\sum_{i=1}^{n} I(\tilde{Y}_i; \tilde{L}_i, Y^{i-1}, X^{i-1}) - 3\tau.
\end{align}
Now let \(\tilde{U}_i \triangleq (\tilde{L}_i, Y^{i-1}, X^{i-1})\). Recall that
\begin{equation}
    \tilde{P}_{X_i, Y_i, L, Y^{i-1}, X^{i-1}} = P_{X^{i-1}} \times P_{X_i} \times P_{L, Y^{i-1}, Y_{i} | X_i, X^{i-1}},
\end{equation}
where
\begin{equation}
    P_{L, Y^{i-1}, Y_{i} | X_i, X^{i-1}} = P_{Y^{i-1} | X^{i-1} } \times P_{Y_i | X_i} \times P_{L | Y^{i-1}, Y_i}.
\end{equation}
Thus, we have the Markov chain \(\tilde{X}_i \to \tilde{Y}_i \to \tilde{U}_i\) for every \(i \in [n]\).
Let \(J\) be independently and uniformly distributed over \([n]\), i.e., the time sharing random variable.
Hence,
\begin{align}
    R & \leq \frac{1}{n}\sum_{i=1}^{n} I(\tilde{X}_i; \tilde{U}_i) + 2\tau + \epsilon_n \\
    & = \frac{1}{n}\sum_{i=1}^{n} I(\tilde{X}_J; \tilde{U}_J |  J =i) + 2\tau + \epsilon_n \\
    & = I(\tilde{X}_J; \tilde{U}_J |  J ) + 2\tau + \epsilon_n \\
    & = I(\tilde{X}_J; \tilde{U}_J ,  J ) + 2\tau + \epsilon_n \label{eq:weak_converse_X_J_independent} \\
    & = I(X; U ) + 2\tau + \epsilon_n, \label{eq:weak_converse_define_X_U}
\end{align}
where \eqref{eq:weak_converse_X_J_independent} is because \(\tilde{X}_i\) follows the distribution \(P_X\) for every \(i \in [n]\), i.e., \(\tilde{X}_J\) is independent of \(J\); in \eqref{eq:weak_converse_define_X_U} we write \(X = \tilde{X}_J \) and \(U \triangleq (\tilde{U}_J, J) \).
As for \(B\), we similarly have
\begin{align}
    B & \geq \frac{1}{n}\sum_{i=1}^{n} I(\tilde{Y}_i; \tilde{U}_i) - 3\tau \\
    & = \frac{1}{n}\sum_{i=1}^{n} I(\tilde{Y}_J; \tilde{U}_J | J=i) - 3\tau \\
    & = I(\tilde{Y}_J; \tilde{U}_J | J) - 3\tau \\
    & =I(\tilde{Y}_J; \tilde{U}_J , J) - 3\tau \\
    & = I(Y ; U) - 3\tau,
\end{align}
where \(\tilde{Y}_J\) follows the distribution \(P_Y = P_{X}\cdot P_{Y|X} \) since every \(\tilde{X}_i\) follows the same distribution \(P_X\), i.e., \(\tilde{Y}_J\) is independent of \(J\) and we write \(Y = \tilde{Y}_J\).
Note that the Markov chain \(X \to Y \to U\) holds, since
\begin{align}
    P_{X, Y, U} & = P_{\tilde{X}_J, \tilde{Y}_J, \tilde{U}_J, J} \\
    & = P_{J} P_{ \tilde{X}_J | J} P_{\tilde{Y}_J | \tilde{X}_J, J} P_{ \tilde{U}_J | \tilde{Y}_J, \tilde{X}_J , J } \\
    & = P_J P_{X}P_{Y | X } P_{\tilde{U}_J | \tilde{Y}_J , J } \label{eq:weak_converse_independence_relationship_X_Y_U} \\
    & = P_{X}P_{Y | X } P_{J | \tilde{Y}_J}  P_{\tilde{U}_J | \tilde{Y}_J , J }  \label{eq:weak_converse_independence_Y_J} \\
    & = P_{X}P_{Y | X} P_{\tilde{U}_J, J | \tilde{Y}_J } \\
    & = P_{X} P_{Y | X} P_{U|Y},
\end{align}
where \eqref{eq:weak_converse_independence_relationship_X_Y_U} is because for every \(J =i \in [n]\), we have \(P_{ \tilde{X}_J | J = i} = P_X\), \(P_{\tilde{Y}_J | \tilde{X}_J, J =i} = P_{Y|X}\), and the Markov chain \(\tilde{X}_i \to \tilde{Y}_i \to \tilde{U}_i\); \eqref{eq:weak_converse_independence_Y_J} is due to the independence between \(J\) and \(\tilde{Y}_J\).
Therefore, for any sequence of \((n, R, B)\)-codes such that \(\bar{\lambda} \to 0\), we must have
\begin{equation}
    R \leq \max_{P_X, P_{U|Y}} I(X;U) \qquad \text{s.t.} \quad I(Y;U) \leq B,
\end{equation}
where \(X \overset{P_{Y|X}}{\to} Y \overset{P_{U|Y}}{\to} U\) forms a Markov chain.
The proof of the cardinality bound for \(\mathcal{U}\) follows from a standard application of the support lemma \cite[Appendix C]{gamalNetworkInformationTheory2011}, and is provided in Appendix \ref{appendix:weak_converse_cardinality_bound_U}.
With this, the proof for Theorem \ref{thm:weak_converse_constant_composition_ensemble} is complete.
\begin{remark}
It can be seen that the requirement \(x^{i-1} \in \mathcal{T}_{i-1}^{\delta_n}(P_X)\) in the set \(\mathcal{E}_{i-1}\) does not play any role in the proof, i.e., the proof still holds if we define $  \mathcal{E}_{i-1} \triangleq \{x^{i-1} \in \supp(P_{X^{i}}): x^{n}_{i} \in \mathcal{T}_{n-i+1}^{\delta_n}(P_X)\}.$
The reason for not using such definition is to provide a slightly more general proof, i.e., there is no causality constraint and the same argument still applies if we instead start from
\begin{align}
    nR &\leq \sum_{i=1}^{n} I(X_i; L, Y^{i+1}, X^{i+1}) + n\epsilon_n \\
    nB &\geq \sum_{i=1}^{n} I(L;Y_i| Y^{i+1}, X^{i+1}),
\end{align}
as discussed in Remark \ref{rem:weak_converse_almost_independent_both_direction}.
\end{remark}

\section{Sphere Packing Bound}
\label{sec:sphere_packing_bound}
In this section, we prove Theorem \ref{thm:sphere_packing_bound}.
We fix a sequence of \((n, R, B)\)-codes, or equivalently a sequence of mappings \((f_n, \varphi_n, \phi_n)\) as defined in Section \ref{section:problem_setting},  where codewords have composition \(P_X\).
Next, we select an auxiliary (or test) channel \(Q_{Y|X}\) and a corresponding IB channel \((Q_{Y|X}, B)\).
We will specify \(Q_{Y|X}\) later on.
The same sequence of \((n, R , B)\)-codes can be applied to both channels \((P_{Y|X},B)\) and \((Q_{Y|X},B)\).
We use the subscript \(P\) or \(Q\) to differentiate all (random) variables and information measures induced under the two channels by the same codes.
For example, given a codebook \(\bm{C} = \mathcal{C}_n\), we denote by \(\lambda_{Q, m}(n, R, B, \mathcal{C}_n)\) the decoding error probability of message \(m\) under the IB channel \((Q_{Y|X}, B)\).
The ensemble-average decoding error probability will then be \(\bar{\lambda}_Q(n, R, B)\) or \(\bar{\lambda}_P(n, R, B)\).

For a codebook \(\bm{C} = \mathcal{C}_n\), define the decoding error region of message \(m\) as
\begin{equation}
    \mathcal{Y}^n(m)^c \triangleq \{ y^n \in \mathcal{Y}^n : \phi_n( \varphi_n (y^n),\mathcal{C}_n) \neq m \},
\end{equation}
i.e., all channel outputs at the relay that are not decoded to message \(m\) at the receiver.
Hence, we have
\begin{equation}
    \lambda_{Q, m}(n, R, B, \mathcal{C}_n) = Q_{Y|X}^n [ \mathcal{Y}^n(m)^c | x^n(m) ]
\end{equation}
as well as
\begin{equation}
    \lambda_{P, m}(n, R, B, \mathcal{C}_n) = P_{Y|X}^n [ \mathcal{Y}^n(m)^c | x^n(m) ].
\end{equation}
The task here is to find a lower bound for  \(\bar{\lambda}_P(n, R, B)\).
We instead find a lower bound for every \(\lambda_{P, m}(n, R, B, \mathcal{C}_n)\), which is accomplished through the test channel.

\subsection{Sphere Packing Bound}
Define the divergence typical set for codeword \(x^n(m)\) from codebook \(\mathcal{C}_n\) as
\begin{equation}
    \mathcal{D}_n^{\epsilon}(m) = \left\{y^n \in \mathcal{Y}^n : \abs{ \frac{1}{n} \log \frac{Q_{Y|X}^n( y^n | x^n(m) )}{ P_{Y|X}^n(y^n | x^n(m)) }  - D(Q_{Y|X} \| P_{Y|X} | P_X)  } \leq \epsilon \right\}
\end{equation}
Since the codeword composition is \(P_X\), under the test channel \(Q^n_{Y|X}\) we have
\begin{equation}
    Q_{Y|X}^n [ \mathcal{D}_n^{\epsilon}(m) | x^n(m)  ] \geq 1 - \alpha_n,
\end{equation}
where \(\alpha_n \) is a linear function of \(\frac{1}{n \epsilon^2 }\) due to the law of large numbers.
For any pair of sets \(\mathcal{A}\) and \(\mathcal{B}\), it holds that $ P[\mathcal{A} \cap \mathcal{B}] \geq  P[\mathcal{A}] + P[ \mathcal{B}] -1$.
Consequently,
\begin{align}
    Q_{Y|X}^n [ \mathcal{Y}^n(m)^c \cap \mathcal{D}_n^{\epsilon}(m) | x^n(m) ]
    & \geq Q_{Y|X}^n [ \mathcal{Y}^n(m)^c  | x^n(m) ] + Q_{Y|X}^n [ \mathcal{D}_n^{\epsilon}(m) | x^n(m) ]  - 1 \\
    & \geq Q_{Y|X}^n [ \mathcal{Y}^n(m)^c  | x^n(m) ] + (1-\alpha_n) - 1 \\
    & = \lambda_{Q, m}(n, R, B, \mathcal{C}_n) - \alpha_n.
\end{align}
Notice that by definition on the divergence typical set, we have
\begin{equation}
    P_{Y|X}^n( y^n | x^n(m) ) \geq Q_{Y|X}^n(y^n | x^n(m))e^{-n(D(Q_{Y|X} \| P_{Y|X} | P_X)+\epsilon)}, \qquad \forall y^n \in \mathcal{D}_n^{\epsilon}(m).
\end{equation}
Therefore, for every codebook in the ensemble \(\bm{C} = \mathcal{C}_n\), we have
\begin{align}
\lambda_{P,m}(n, R, B, \mathcal{C}_n) \nonumber & = P_{Y|X}^n [ \mathcal{Y}^n(m)^c | x^n(m) ] \\
    & \geq P_{Y|X}^n [ \mathcal{Y}^n(m)^c \cap \mathcal{D}_n^{\epsilon}(m) | x^n(m) ] \\
    & \geq Q_{Y|X}^n [ \mathcal{Y}^n(m)^c \cap \mathcal{D}_n^{\epsilon}(m) | x^n(m) ] \times e^{-n(D(Q_{Y|X} \| P_{Y|X} | P_X)+\epsilon)} \label{eq:sphere_packing_bound_applying_divergence_bound}\\
    & \geq ( \lambda_{Q, m}(n, R, B, \mathcal{C}_n) - \alpha_n )e^{-n(D(Q_{Y|X} \| P_{Y|X} | P_X)+\epsilon)}.
\end{align}
Hence, after averaging over the message set and ensemble, we arrive at
\begin{align}
    \bar{\lambda}_P(n, R, B) & \geq ( \bar{\lambda}_Q(n, R, B) - \alpha_n ) e^{-n(D(Q_{Y|X} \| P_{Y|X} | P_X)+\epsilon)}. \label{eq:sphere_packing_bound_test_channel}
\end{align}
The task now is reduced to finding a lower bound for the test channel's decoding error probability \(\bar{\lambda}_Q(n, R, B)\) under the \((n, R, B)-\)code, which will be done through Fano's inequality.
\begin{remark}
    It can be seen that \(\mathcal{D}_n^{\epsilon}(m)\) is roughly the same as \(\mathcal{T}_n^{\epsilon}(Q_{Y|X} | x^n(m))\), the set of all conditional typical sequences.
    Thus, the set \(\mathcal{Y}^n(m)^c \cap \mathcal{D}_n^{\epsilon}(m) \) can be interpreted as sequences \(y^n\) from the shell \(\mathcal{T}_n^{\epsilon}(Q_{Y|X} | x^n(m))\) that lead to a decoding error event at the relay.
    The ratio of such \(y^n\) in the shell is
    \begin{align}
        \frac{ | \mathcal{Y}^n(m)^c \cap \mathcal{D}_n^{\epsilon}(m) | }{ | \mathcal{D}_n^{\epsilon}(m) |   } & \approx \frac{  Q_{Y|X}^{n} [  \mathcal{Y}^n(m)^c \cap \mathcal{D}_n^{\epsilon}(m) | x^n(m)  ] }{  Q_{Y|X}^{n} [ \mathcal{D}_n^{\epsilon}(m)  | x^n(m) ] } \label{eq:sphere_packing_bound_interpretation_sequence_probability} \\
        & \approx  Q_{Y|X}^{n} [  \mathcal{Y}^n(m)^c \cap \mathcal{D}_n^{\epsilon}(m) | x^n(m)  ], \label{eq:sphere_packing_bound_interpretation_ratio}
    \end{align}
    where \eqref{eq:sphere_packing_bound_interpretation_sequence_probability} is because every sequence in the shell has roughly the same probability; \eqref{eq:sphere_packing_bound_interpretation_ratio} is due to \(Q_{Y|X}^{n} [ \mathcal{D}_n^{\epsilon}(m)  | x^n(m) ] \approx 1\).
    Hence, the lower bound in \eqref{eq:sphere_packing_bound_applying_divergence_bound} can be interpreted as the probability of the shell \(\mathcal{T}_n^{\epsilon}(Q_{Y|X} | x^n(m))\) under the channel \(P_{Y|X}^n\) (i.e., \(\exp\{-n(D(Q_{Y|X} \| P_{Y|X} | P_X)+\epsilon)\}\)) multiplied with the ratio of error sequences in the shell.
    The test channel \(Q_{Y|X}\) we selected determines which shell \(\mathcal{T}_n^{\epsilon}(Q_{Y|X} | x^n(m))\) is picked to constitute the lower bound.
\end{remark}
\begin{remark}
\label{rem:traditional_sphere_packing_bound}
    We can readily see that \eqref{eq:sphere_packing_bound_test_channel} leads to a sphere packing bound.
    In particular, for any IB channel \((Q_{Y|X}, B)\) whose capacity is less than \(R\), i.e., \(C(B) \leq R\), the weak converse for \((Q_{Y|X}, B)\) derived in the previous section suggests that for  sufficiently large \(n\), we will have $\bar{\lambda}_Q(n, R, B) \geq \tau$, where \(\tau \in (0,1)\) is a small constant.
    Since \eqref{eq:sphere_packing_bound_test_channel} holds for any test channel \((Q_{Y|X}, B)\), we can select the best test channel under the constraint \(C(B) \leq R\), resulting in
    \begin{align}
        \bar{\lambda}_P(n, R, B) & \geq \max_{ \substack{ (Q_{Y|X}, B): \\ C(B) \leq R } }( \tau - \alpha_n ) e^{-n(D(Q_{Y|X} \| P_{Y|X} | P_X)+\epsilon)} \\
        & \ndot{=} \max_{ \substack{ (Q_{Y|X}, B): \\ C(B) \leq R } }e^{-nD(Q_{Y|X} \| P_{Y|X} | P_X)}. \label{eq:sphere_packing_bound_traditional_bound}
    \end{align}
    \eqref{eq:sphere_packing_bound_traditional_bound} is established by following the Haroutunian's conventional approach of establishing sphere packing bounds \cite{haroutunianBoundsExponentProbability1968}.
    As we will see in the following sections, by considering Fano's lower bound, this conventional approach can be further refined in some cases, and we will derive an improved sphere packing bound for the IB channel. The refined approach is inspired by the work of Kelly and Wagner \cite{kellyReliabilitySourceCoding2012}.
\end{remark}
\subsection{Fano's Lower Bound}
We now find a lower bound for the test channel's ensemble-average error probability \(\bar{\lambda}_Q(n, R, B)\).
As discussed in Section \ref{sec:weak_converse_for_capacity}, conditioned on any codebook \(\bm{C} = \mathcal{C}_n\), we have the Markov chain
\begin{equation}
    M \to x^n(M) \to Y^n \to L \to \hat{M}.
\end{equation}
This Markov chain holds under any IB channel, e.g., both \((P_{Y|X},B)\) and \((Q_{Y|X},B)\).
Suppose the underlying channel for the Markov chain is the test channel \((Q_{Y|X}, B)\).
From Fano's inequality, conditioned on any codebook \(\bm{C} = \mathcal{C}_n\), we have
\begin{equation}
    H_Q(M|L, \bm{C} = \mathcal{C}_n ) \leq H_Q(M|\hat{M}, \bm{C} = \mathcal{C}_n ) \leq 1 + \bar{\lambda}_Q(n, R, B, \mathcal{C}_n)nR.
\end{equation}
After averaging over the ensemble \(\bm{C}\), we obtain
\begin{equation}
    H_Q(M|L, \bm{C} )  \leq 1 + \bar{\lambda}_Q(n, R, B)nR.
\end{equation}
Since \(M\) is uniformly distributed over \([e^{nR}]\), we see that \(H(M) = nR\) and hence
\begin{equation}
    I_Q(M;L,\bm{C}) = H(M) - H_Q(M | L, \bm{C}) \geq nR -1 -\bar{\lambda}_Q(n, R, B)nR,
\end{equation}
that is,
\begin{equation}
    \bar{\lambda}_Q(n, R, B) \geq 1 - \frac{I_Q(M;L,\bm{C}) + 1}{nR}.
\end{equation}
This is known as Fano's lower bound.
From the converse proof in Section \ref{sec:weak_converse_Fano_inequality}, we have
\begin{align}
    \frac{1}{n}I_Q(M;L,\bm{C})  & \leq \frac{1}{n}\sum_{i=1}^{n} I_Q(\tilde{X}_i; \tilde{L}_i, Y^{i-1}, X^{i-1}) + 2\tau \\
    & = I_Q(X;U) + 2\tau
\end{align}
and
\begin{align}
    B  & \geq \frac{1}{n}\sum_{i=1}^{n} I_Q(\tilde{Y}_i; \tilde{L}_i, Y^{i-1}, X^{i-1}) -\frac{2\sqrt{n}}{n}\log \abs{\mathcal{Y}}  - 3\tau\\
    & = I_Q(Y;U) -3\tau,
\end{align}
where we have the Markov chain \(X \overset{Q_{Y|X}}{\to} Y \overset{Q_{U|Y}}{\to} U\) with \(X = \tilde{X}_J\), \(Y = \tilde{Y}_J\), and \(U = ( \tilde{L}_J, Y^{J-1}, X^{J-1}, J  )\).
It is evident that \(Q_{U|Y}\) depends on \(Q_{Y|X}\), i.e., it varies for different DMCs \(Q_{Y|X}\).
We assume that the selected auxiliary channel \((Q_{Y|X}, B)\) satisfies
\begin{equation}
    I_Q(X;U)  \leq R - \nu - 2\tau, \label{eq:sphere_packing_bound_requirement_test_channel}
\end{equation}
for a small constant \(\nu > 0\).
Therefore, if we select the test channel \(Q_{Y|X}\) according to the requirement in \eqref{eq:sphere_packing_bound_requirement_test_channel}, we see that
\begin{align}
    \bar{\lambda}_Q(n, R, B) & \geq 1 - \frac{I_Q(M;L,F_n) + 1}{nR} \\
    & \geq 1- \frac{  n \big(I_Q(X;U) + 2\tau \big) + 1 }{ nR } \\
    & \geq 1- \frac{ nR - n\nu + 1  }{  nR  } \\
    & = \frac{ n\nu -1 }{nR}. \label{eq:sphere_packing_bound_lower_bound_for_test_channel}
\end{align}
We now can substitute \eqref{eq:sphere_packing_bound_lower_bound_for_test_channel} into \eqref{eq:sphere_packing_bound_test_channel} to obtain a lower bound for \(\bar{\lambda}_P(n, R, B)\).

\subsection{Optimizing over Test Channels}
Since we can freely select the test channel \((Q_{Y|X}, B)\), we can select the one that produces the tightest lower bound for \(\bar{\lambda}_P(n, R, B)\).
Recall the requirement that the selected channel \((Q_{Y|X}, B)\) must satisfy
\begin{equation}
    I(P_X, Q_{Y|X} \cdot Q_{U|Y})  \leq R - \nu - 2\tau, \label{eq:sphere_packing_bound_test_channel_requirement_1}
\end{equation}
where \(Q_{U|Y}\) is a certain channel depending on \(Q_{Y|X}\).
Moreover, for this specific \(Q_{U|Y}\), we must have
\begin{equation}
    I(Q_Y, Q_{U|Y}) \leq B + 3\tau, \label{eq:sphere_packing_bound_test_channel_requirement_2}
\end{equation}
where \(Q_Y = P_X \cdot Q_{Y|X}\).
Therefore, we can deduce that
\begin{align}
    & \bar{\lambda}_P(n, R, B) \nonumber \\
    &\geq \max_{(Q_{Y|X}, B)} \ \  ( \frac{ n\nu -1 }{nR}  - \alpha_n )e^{-n(D(Q_{Y|X} \| P_{Y|X} | P_X)+\epsilon)} \label{eq:sphere_packing_bound_best_test_channel} \\
    & = \max_{ Q_Y  }  \max_{ \substack{Q_{Y|X}: \\  I(Q_Y, Q_{U|Y}) \leq B +3\tau, \\  I(P_X, Q_{Y|X} \cdot Q_{U|Y})  \leq R - \nu - 2\tau} } ( \frac{  n\nu  -1 }{nR}  - \alpha_n )e^{-n(D(Q_{Y|X} \| P_{Y|X} | P_X)+\epsilon)} \label{eq:sphere_packing_bound_max_over_auxiliary_channel} \\
    & \geq  \max_{ Q_Y  } \min_{P_{U|Y}}  \max_{ \substack{Q_{Y|X}: \\  I(Q_Y, P_{U|Y}) \leq B +3\tau, \\  I(P_X, Q_{Y|X} \cdot P_{U|Y})  \leq R - \nu - 2\tau} } ( \frac{  n\nu  -1 }{nR}  - \alpha_n )e^{-n(D(Q_{Y|X} \| P_{Y|X} | P_X)+\epsilon)} \label{eq:sphere_packing_bound_min_over_P_U_Y} \\
    & = \max_{ Q_Y  } \min_{ \substack{  P_{U|Y}:  \\   I(Q_Y, P_{U|Y}) \leq B +3\tau  } }  \max_{ \substack{Q_{Y|X}: \\  I(P_X, Q_{Y|X} \cdot P_{U|Y})  \leq R - \nu - 2\tau} } ( \frac{  n\nu  -1 }{nR}  - \alpha_n )e^{-n(D(Q_{Y|X} \| P_{Y|X} | P_X)+\epsilon)}, \label{eq:sphere_packing_bound_independent_P_U_Y_and_Q_Y_X}
\end{align}
where in \eqref{eq:sphere_packing_bound_best_test_channel} the maximization means that we select the best test channel under the two requirements in \eqref{eq:sphere_packing_bound_test_channel_requirement_1} and \eqref{eq:sphere_packing_bound_test_channel_requirement_2};
in \eqref{eq:sphere_packing_bound_max_over_auxiliary_channel} we select the best test channel by first fixing a type \(Q_Y\) and then looking into all \(Q_{Y|X}\) such that \(P_{X} \cdot Q_{Y|X} = Q_Y\);
in \eqref{eq:sphere_packing_bound_min_over_P_U_Y} we recall that \(Q_{U|Y}\) depends on \(Q_{Y|X}\), so we can lower bound it by minimizing over \(P_{U|Y}\), which is now independent of \(Q_{Y|X}\);
and in \eqref{eq:sphere_packing_bound_independent_P_U_Y_and_Q_Y_X} we notice that the constraint \( I(Q_Y, P_{U|Y}) \leq B +3\tau\) is independent of \(Q_{Y|X}\).

In the achievability proof, \(P_{U|Y}\) represents the compress-forward scheme between the relay and receiver.
Hence, \eqref{eq:sphere_packing_bound_min_over_P_U_Y} can be interpreted as that we select the optimal compress-forward scheme such that the lower bound in \eqref{eq:sphere_packing_bound_min_over_P_U_Y}  is as small as possible, i.e., the error exponent is as large as possible.
Now recall that \(\alpha_n\) is a linear function of \(\frac{1}{n\epsilon^2}\).
Hence, it is guaranteed that for sufficiently large \(n\), we have
\begin{equation}
    \frac{ n\nu -1 }{nR}  - \alpha_n  = \frac{\nu}{R} - \frac{1}{nR} - \alpha_n > 0.
\end{equation}
Since \eqref{eq:sphere_packing_bound_independent_P_U_Y_and_Q_Y_X} holds for any \(\nu, \tau, \epsilon > 0\) as \(n \to \infty\), we conclude that
\begin{equation}
    \limsup_{n \to \infty} -\frac{1}{n} \log \bar{\lambda}(n, R, B) \leq E_{\textnormal{sp}}(R, B, P_X),
\end{equation}
where
\begin{equation}
    E_{\textnormal{sp}}(R, B, P_X) = \min_{Q_{Y}} \max_{ \substack{ P_{U|Y}: \\ I(Q_Y, P_{U|Y}) \leq B } }  \min_{ \substack{Q_{Y|X}: \\ P_X \cdot Q_{Y|X} = Q_{Y}, \\ I(P_X, Q_{Y|X} \cdot P_{U|Y} )  \leq R } } \! \! D(Q_{Y|X} \| P_{Y|X} | P_{X}).
\end{equation}
The proof of cardinality bound makes use of the idea in \cite[Theorem 2]{kellyReliabilitySourceCoding2012} through combining the support lemma with KKT conditions, and is provided in Appendix \ref{appendix:weak_converse_cardinality_sphere_packing_bound}.
With this, the theorem is established.


\section{Connections to the WAK problem}
\label{sec:source_coding_coded_side_information_exponent_lower_bound}
In this section, we establish a connection between coding for the IB channel and coding for the WAK problem, which we then utilize to prove Theorem \ref{thm:source_coding_coded_side_information_exponent_lower_bound}. To this end, we first present a few preliminary results on covering through permutations, which will be useful in our proof later on.
\subsection{Permutations and Type Class Covering}
We first revisit Ahlswede's Covering Lemma from \cite[Appendix I]{ahlswedeGoodCodesCan1982} (cf. \cite[Section 6]{ahlswedeColoringHypergraphsNew1980}).
To this end, we need to introduce some definitions and notation related to permutations.

Consider a permutation rule \(\pi\) on the set \([n]\), i.e., a one-to-one mapping \(\pi: [n] \to [n]\).
For a sequence \(\bm{x} = (x_1, x_2, \ldots, x_n)\), we denote by \(\pi [\bm{x}]\) the sequence obtained through permuting the entries of \(\bm{x}\) under \(\pi\).
We denote by \(\pi_1 \circ \pi_2\) the composition (or product) of two permutations, i.e.,
\begin{equation}
    \pi_1 \circ \pi_2 [ \bm{x} ] = \pi_1 \big[  \pi_2 [ \bm{x} ] \big].
\end{equation}
Note that in general \(\pi_1 \circ \pi_2 [ \bm{x} ] \neq \pi_2 \circ \pi_1 [ \bm{x} ]\).
For a set \(\mathcal{A} \subseteq \mathcal{X}^n\), we write
\begin{equation}
    \pi [ \mathcal{A} ] \triangleq \{ \pi [ \bm{x}] :  \bm{x} \in \mathcal{A} \}.
\end{equation}

\begin{lemma}[Ahlswede's Covering Lemma]
    \label{lemma:Ahlswede_covering_lemma}
    Fix a type \(Q_X \in \mathcal{P}_n(\mathcal{X})\) and a set of sequences \(\mathcal{A} \subseteq \mathcal{T}_n(Q_X)\). There exists a sequence of permutations \(\pi_1, \pi_2, \ldots, \pi_k \) such that
    \begin{equation}
        \bigcup_{i=1}^{k} \pi_i [ \mathcal{A} ] = \mathcal{T}_n(Q_X),
    \end{equation}
    if \(k > \abs{\mathcal{A}}^{-1} \abs{\mathcal{T}_n(Q_X)} \log \abs{ \mathcal{T}_n(Q_X) }\).
\end{lemma}
\begin{proof}
Ahlswede's original proof is established for a more general result in the context of graph covering.
In Appendix \ref{appendix:Ahlswede_covering_lemma_proof}, we present a specialized version of his proof, distilled from \cite{ahlswedeColoringHypergraphsNew1979} and \cite{ahlswedeColoringHypergraphsNew1980}, and also fill in some missing details he omitted.
Our specialized version of the proof also serves as an important first step towards an extension of this lemma discussed next.
\end{proof}

Ahlswede's covering lemma states that for every set \(\mathcal{A} \subseteq \mathcal{T}_n(Q_X)\), we can find a sequence of $k$ permutations such that the union of the permuted \(\mathcal{A}\)'s covers \(\mathcal{T}_n(Q_X)\), where $ k \ndot{=} \abs{\mathcal{A}}^{-1} \abs{\mathcal{T}_n(Q_X)}$.
However, this sequence of permutations may depend on the particular set \(\mathcal{A}\).
Now suppose that we have multiple distinct sets \(\mathcal{A}_1, \mathcal{A}_2, \ldots, \mathcal{A}_J\) from the same type class $\mathcal{T}_n(Q_X)$. We are interested in finding a sequence of $k$ permutations under which covering simultaneously holds for almost all sets, i.e.,
\begin{equation}
    \bigcup_{i=1}^{k} \pi_i [ \mathcal{A}_j ] = \mathcal{T}_n(Q_X)
\end{equation}
should hold for a large fraction of $j \in [J]$. The key question is, \emph{how small can $k$ be?}
In the following, we  provide an answer to this by extending Ahlswede's Covering Lemma.
\begin{lemma}[Simultaneous Covering]
\label{lemma:Ahlswede_covering_lemma_expurgation}
Fix a type \(Q_X \in \mathcal{P}_n(\mathcal{X})\) and consider an arbitrary collection of sets
\begin{equation}
    \mathcal{F} = \{ \mathcal{A}_1, \mathcal{A}_2, \ldots, \mathcal{A}_J\},
\end{equation}
where \(\mathcal{A}_j \subseteq \mathcal{T}_n(Q_X)\) for every \(\mathcal{A}_j \in \mathcal{F}\).
Let
\begin{equation}
    \mathcal{A}_{\min} = \argmin_{\mathcal{A}_j \in \mathcal{F}} \abs{\mathcal{A}_j}.
\end{equation}
Then, there exists a sequence of permutations \(\pi_1, \pi_2, \ldots, \pi_k \) such that
\begin{equation}
    \bigcup_{i=1}^{k} \pi_i [ \mathcal{A}_j ] = \mathcal{T}_n(Q_X)
\end{equation}
holds for at least half of \(\mathcal{A}_j \in \mathcal{F}\), if \(k > |\mathcal{A}_{\min}|^{-1}\abs{\mathcal{T}_n(Q_X)} \log 2\abs{ \mathcal{T}_n(Q_X) }\).
\end{lemma}
\begin{proof} In the proof, we build upon Lemma \ref{lemma:Ahlswede_covering_lemma} using an expurgation argument. See Appendix \ref{appendix:Ahlswede_covering_lemma_expurgation_proof}.
\end{proof}
\begin{remark}
With a slight modification in the proof, the fraction \(1/2\) of sets in Lemma \ref{lemma:Ahlswede_covering_lemma_expurgation} can be changed to any \(\delta \in (0,1)\), as long as  \(k > |\mathcal{A}_{\min}|^{-1}\abs{\mathcal{T}_n(Q_X)} \log  (1-\delta)^{-1} \abs{ \mathcal{T}_n(Q_X) }\).
This is the same for all the following results, where the corresponding fractions can be manipulated in a similar fashion.
\end{remark}
Given a constant composition codebook \(\mathcal{C}_n\) with codewords from the type class $\mathcal{T}_n(Q_X)$, let \(|\mathcal{C}_n|\) denote the number of its unique codewords (there may be repetitions of the same codeword within a codebook).
Even though  Lemma \ref{lemma:Ahlswede_covering_lemma_expurgation} is established for a collection of sets \(\{\mathcal{A}_j\}\), it is not difficult to modify it to hold for a collection of constant composition codebooks \(\{\mathcal{C}_n\}\), leading directly to the following corollary.
\begin{corollary}
    \label{cor:Ahlswede_covering_lemma_codebook_expurgation}
    Fix a type \(Q_X \in \mathcal{P}_n(\mathcal{X})\) and consider an arbitrary collection of codebooks \(\mathcal{F} = \{ \mathcal{C}_n\}\), where every codebook in \(\mathcal{F}\) has constant composition codewords from \( \mathcal{T}_n(Q_X) \).
    Let
    \begin{equation}
        \mathcal{C}_{\min} = \argmin_{\mathcal{C}_n \in \mathcal{F}} \abs{\mathcal{C}_n}.
    \end{equation}
    Then, there exists a sequence of permutations \(\pi_1, \pi_2, \ldots, \pi_k \) such that
    \begin{equation}
        \bigcup_{i=1}^{k} \pi_i [ \mathcal{C}_n ] = \mathcal{T}_n(Q_X)
    \end{equation}
    holds for at least half of \(\mathcal{C}_n \in \mathcal{F}\), if \(k > |\mathcal{C}_{\min}|^{-1}\abs{\mathcal{T}_n(Q_X)} \log 2\abs{ \mathcal{T}_n(Q_X) }\).
\end{corollary}
For reasons that will be clear later on, we seek to apply Corollary \ref{cor:Ahlswede_covering_lemma_codebook_expurgation} to the ensemble of constant composition codebooks of rate \(\tilde{R}\) and codeword composition \(Q_X\), i.e., the collection $\mathcal{T}_n(Q_X)^{e^{n\tilde{R}}} $.
The aim there is to demonstrate the existence of \(\pi_1, \pi_2, \ldots, \pi_{k}\), where \(k \ndot{=} e^{n(H(Q_X) - \tilde{R})}\), under which the covering of $\mathcal{T}_n(Q_X)$ is simultaneously achieved by at least half $\mathcal{C}_n \in \mathcal{T}_n(Q_X)^{e^{n\tilde{R}}}$.
However, we encounter a problem if we attempt to directly apply Corollary \ref{cor:Ahlswede_covering_lemma_codebook_expurgation}.
In particular, there are codebooks in the ensemble consisting of only a single unique codeword (i.e., all codewords are the same in the codebook), so we will have \(|\mathcal{C}_{\min}| = 1\). This results in \(k \ndot{=} e^{nH(Q_X)} \) which is trivially achieved and too large for our purpose.

To circumvent this issue, we first restrict our attention to a collection of codebooks \(\mathcal{C}_n\) from $\mathcal{T}_n(Q_X)^{e^{n\tilde{R}}} $ that satisfy \(|\mathcal{C}_n| > \frac{1}{2}e^{n\tilde{R}}\), i.e., \(|\mathcal{C}_{\min}| > \frac{1}{2}e^{n\tilde{R}}\) for this collection.
It turns out that for large $n$, this collection contains almost all codebooks in $\mathcal{T}_n(Q_X)^{e^{n\tilde{R}}} $, as seen through the proof of the following theorem.
\begin{theorem}
    \label{thm:Ahlswede_covering_lemma_ensemble}
    Fix a rate $\tilde{R} > 0$ and type \(Q_X \in \mathcal{P}_n(\mathcal{X})\) with \(H(Q_X) > \tilde{R}\), and consider the constant composition ensemble with rate \(\tilde{R}\) and codeword composition \(Q_X\).
    For sufficiently large \(n\), there exists a sequence of permutations \(\pi_1, \pi_2, \ldots, \pi_k \) such that
    \begin{equation}
        \bigcup_{i=1}^{k} \pi_i [ \mathcal{C}_n ] = \mathcal{T}_n(Q_X)
    \end{equation}
    holds for at least half of \(\mathcal{C}_n \in  \mathcal{T}_n(Q_X)^{e^{n\tilde{R}}}\), where \(k \ndot{=}e^{n(H(Q_X) - \tilde{R})}\).
\end{theorem}
\begin{proof}
Consider the collection of codebooks in which more than $1/2$ of codewords are unique, i.e.,
\begin{equation}
    \mathcal{F} = \Big\{ \mathcal{C}_n \in  \mathcal{T}_n(Q_X)^{e^{n\tilde{R}}} :  | \mathcal{C}_n |   > \frac{1}{2} e^{n\tilde{R}}  \Big\}.
\end{equation}
Applying Corollary \ref{cor:Ahlswede_covering_lemma_codebook_expurgation} to \(\mathcal{F}\), we see that there is a sequence of permutations \(\pi_1, \pi_2, \ldots, \pi_k \) such that
\begin{equation}
    \bigcup_{i=1}^{k} \pi_i [ \mathcal{C}_n ] = \mathcal{T}_n(Q_X)
\end{equation}
holds for at least a fraction $\delta = 2/3$ of \(\mathcal{C}_n \in \mathcal{F}\), where
\begin{equation}
    k =  2 e^{-n\tilde{R}}  \abs{\mathcal{T}_n(Q_X)} \log 3\abs{ \mathcal{T}_n(Q_X) } > |\mathcal{C}_{\min}|^{-1}\abs{\mathcal{T}_n(Q_X)} \log 3\abs{ \mathcal{T}_n(Q_X) },
\end{equation}
which holds since \(|\mathcal{C}_{\min}| > \frac{1}{2} e^{n\tilde{R}} \).
To complete the proof of the theorem, we show that as $n$ grows large, almost all constant composition codebooks in $\mathcal{T}_n(Q_X)^{e^{n\tilde{R}}}$ are also in $\mathcal{F}$.
To this end, recall that the random constant composition codebook $\bm{C} $ is uniformly distributed on $\mathcal{T}_n(Q_X)^{e^{n\tilde{R}}}$.
\begin{lemma}
    \label{prop:ratio_unique_codewords_ensemble}
    The probability $\P\{ | \bm{C} |   \leq \frac{1}{2} e^{n\tilde{R}} \}$ decays to \(0\) double exponentially. Hence, the ratio of codebooks in $\mathcal{T}_n(Q_X)^{e^{n\tilde{R}}}$ with less than \(\frac{1}{2} e^{n\tilde{R}}\) unique codewords decays to $0$ double exponentially.
\end{lemma}
\begin{proof}
See Appendix \ref{appendix:proof_ratio_unique_codewords_ensemble}.
\end{proof}
Since at least $2/3$ of codebooks in $\mathcal{F}$ satisfy the simultaneous covering property under \(\pi_1, \pi_2, \ldots, \pi_k \), and by Lemma \ref{prop:ratio_unique_codewords_ensemble} we have $|\mathcal{F}|/ |\mathcal{T}_n(Q_X)|^{e^{n\tilde{R}}} \to 1$ as $n \to \infty$, then for large enough  $n$, we see that
\begin{equation}
    \bigcup_{i=1}^{k} \pi_i [ \mathcal{C}_n ] = \mathcal{T}_n(Q_X)
\end{equation}
holds for at least half of \(\mathcal{C}_n \in \mathcal{T}_n(Q_X)^{e^{n\tilde{R}}}\), where \(k \ndot{=}e^{n(H(Q_X) - \tilde{R})}\). This completes the proof.
\end{proof}
Theorem \ref{thm:Ahlswede_covering_lemma_ensemble} will play an essential role in constructing good codes for the WAK problem from good codes for the IB channel through permutations, as we see next.
\subsection{Encoder at the Transmitter}
The transmitter describes \(X^n\) to the receiver through an encoder mapping \(f_n^{\prime} : \mathcal{X}^n \to [e^{nR}]\), chosen as follows.
For each type \(Q_X \in \mathcal{P}_n(\mathcal{X})\) satisfying \(H(Q_X) \geq R\), we consider the constant composition codebook ensemble  with rate \(\tilde{R} = H(Q_X) - R\) and codeword composition \(Q_X\).
Since \(H(Q_X) > \tilde{R}\) if \(R > 0\), we follow Theorem \ref{thm:Ahlswede_covering_lemma_ensemble} and find a sequence of permutations \(\pi_1, \ldots, \pi_{k}\) such that
\begin{equation}
    \bigcup_{i=1}^{k} \pi_i [ \mathcal{C}_n ] = \mathcal{T}_n(Q_X) \label{eq:source_coding_coded_side_information_assumption_on_permutation}
\end{equation}
holds for at least half of \(\mathcal{C}_n \in \mathcal{T}_n(Q_X)^{e^{n \tilde{R}}}\), where \(k \ndot{=} e^{nR}\).
Next, we select a codebook from \(\mathcal{T}_n(Q_X)^{e^{n \tilde{R}}} \) such that \eqref{eq:source_coding_coded_side_information_assumption_on_permutation} holds and denote it by \(\mathcal{C}_n(Q_X) \).
Let \( \{ \mathcal{C}_n(Q_X) \}\) denote the set of selected codebooks for different types.
Both the transmitter and receiver are assumed to have access to the sequence of permutations \(\pi_1, \ldots, \pi_{k}\) associated with each \(Q_X\), as well as the codebooks \( \{ \mathcal{C}_n(Q_X) \}\).

Given an observation \(\bm{X} = \bm{x}\), the transmitter first examines the type of \(\bm{x}\), and sends an index to describe \(\hat{P}_{\bm{x}}\) to the receiver.
Since there are at most \((1+n)^{\abs{\mathcal{X}}}\) possible types, including the type index does not break the rate limit \(R\) asymptotically.
Next, if \(H(\hat{P}_{\bm{x}}) < R\), then the transmitter sends an index from \([e^{nR}]\) to describe \(\bm{x}\).
Combined with the type index of \(\hat{P}_{\bm{x}}\), we see that the receiver can recover such \(\bm{x}\) losslessly even without the helper, as observed by Oohama and Han \cite{oohamaUniversalCodingSlepianWolf1994} in the Slepian-Wolf problem.
Therefore, we will ignore these \(\bm{x}\) from now on, since they do not contribute to the decoding error probability.
On the other hand,  if \(H(\hat{P}_{\bm{x}}) \geq R\), then the transmitter looks up the permutations \(\pi_1, \ldots, \pi_{k}\) associated with \(\hat{P}_{\bm{x}}\) and also the selected codebook \(\mathcal{C}_n( \hat{P}_{\bm{x}} )\).
It identifies a permutation index \(i \in [k]\) such that \(\bm{x} \in \pi_i [ \mathcal{C}_n (\hat{P}_{\bm{x}}) ]\).
Since \(\mathcal{C}_n( \hat{P}_{\bm{x}} )\) and the permutations are selected to satisfy \eqref{eq:source_coding_coded_side_information_assumption_on_permutation}, it is guaranteed that such index \(i\) must exist.
The transmitter selects one arbitrarily if there are multiple such \(i\).
It then sends the permutation index \(i\) to the receiver.
Since \(k \ndot{=} e^{nR}\), the rate limit \(R\) is satisfied asymptotically.
\subsection{Encoder at the Helper}
We now turn our attention to the helper, which describes the side information \(Y^n\) to the receiver through an independent encoder \(\varphi_{n}': \mathcal{Y}^n \to [e^{nB}]\).
As stated earlier, given \(\bm{X} = \bm{x}\), the transmitter sends the type index for \(\hat{P}_{\bm{x}}\) and the permutation index \(i\) to the receiver.
With knowledge of \(\hat{P}_{\bm{x}}\) and permutation index \(i\), the receiver finds \( \pi_i [ \mathcal{C}_{n}(\hat{P}_{\bm{x}})]\) that contains $\bm{x}$, since it also has access to \(\{ \mathcal{C}_n(Q_X) \}\) and the permutations associated with each codebook,  while the helper is oblivious to it.

Going forward, we may view \(\pi_i[  \mathcal{C}_n(\hat{P}_{\bm{x}}) ]\) as a codebook from the IB channel setting, where the sequence \(\bm{x} \in \pi_i[  \mathcal{C}_n(\hat{P}_{\bm{x}}) ] \) generated by the source can be regarded as a codeword in this  codebook.
The distribution of the random side information sequence \(\bm{Y}\) conditioned on \(\bm{X} = \bm{x}\) is
\begin{equation}
    P_{\bm{Y} | \bm{X}}(\bm{y} | \bm{x}) = \prod_{i=1}^{n}P_{Y|X}( y_i | x_i ),
\end{equation}
i.e., the channel from the transmitter to the helper is the DMC \(P_{Y|X}\).
Given the rate-limited description \(l \in [e^{nB}]\) from the helper, the receiver decides which source sequence in \(\pi_i[  \mathcal{C}_n(\hat{P}_{\bm{x}}) ] \) is observed at the transmitter, i.e.,  which codeword from the codebook \(\pi_i[  \mathcal{C}_n(\hat{P}_{\bm{x}}) ]\) is passing through \(P_{Y|X}\) to the helper.
We can hence regard the  transmitter-helper-receiver path as an instance of the IB channel \((P_{Y|X}, B)\), where the helper takes the oblivious relay's role.
We choose the helper's encoder \(\varphi_n'\) to be the same as the oblivious relay's compress-forward mapping  \(\varphi_n\) in Section \ref{sec:lower_bound_proof}, and therefore the construction of the bottleneck codebooks \(\{\mathcal{B}_n(Q_Y)\}\) at the helper is the same as the one stated in Section \ref{sec:lower_bound_bottleneck_codebook_construction}.

\subsection{Error Analysis}
We now show that  the coding scheme constructed by permuting good codes for the IB channel attains the best known achievable error exponent of the WAK problem, previously established in \cite{kellyReliabilitySourceCoding2012}.
Note that the decoding strategy at the receiver is same as the one in Section \ref{sec:lower_bound_encoding_decoding}, with the only difference being that the codebook in use, i.e., \(\pi_i[  \mathcal{C}_n(\hat{P}_{\bm{x}}) ] \), is communicated to the receiver through the forwarded type $\hat{P}_{\bm{x}}$ and permutation index \(i\).

Recall that under the encoder at the transmitter, if \(H(\hat{P}_{\bm{x}}) < R\), then the receiver can recover \(\bm{x}\) losslessly.
Thus, under the coding scheme we described, we have
\begin{align}
    & \P \{ \hat{\bm{X}} \neq \bm{X}\} \nonumber \\
    & = \sum_{ \substack{ \bm{x} \in \mathcal{X}^n : \\ H(\hat{P}_{\bm{x}}) \geq R}  } \P \{ \bm{X} = \bm{x}\} \times \P \{ \hat{\bm{X}} \neq \bm{x} | \bm{X} = \bm{x}\} \\
    & = \sum_{  \substack{Q_X \in \mathcal{P}_n(\mathcal{X}) : \\ H(Q_X) \geq R}  } \  \sum_{ \bm{x} \in \mathcal{T}_n(Q_X)} \P \{ \bm{X} = \bm{x}\} \times \P \{ \hat{\bm{X}} \neq \bm{x} | \bm{X} = \bm{x}\} \\
    & = \sum_{  \substack{Q_X \in \mathcal{P}_n(\mathcal{X}) : \\ H(Q_X) \geq R}  } e^{-n( D(Q_X\|P_X) + H(Q_X) )} \times \sum_{ \bm{x} \in \mathcal{T}_n(Q_X)} \P \{ \hat{\bm{X}} \neq \bm{x} | \bm{X} = \bm{x}\} \label{eq:source_coding_coded_side_information_type_sequence_probability} \\
    & \leq \sum_{  \substack{Q_X \in \mathcal{P}_n(\mathcal{X}) : \\ H(Q_X) \geq R}  } e^{-n( D(Q_X\|P_X) + H(Q_X) )} \times  \sum_{i=1}^{k}  \sum_{ \bm{x} \in \pi_i [\mathcal{C}_n(Q_X)] } \P \{ \hat{\bm{X}} \neq \bm{x} | \bm{X} = \bm{x} \} \label{eq:source_coding_coded_side_information_union_permutation} \\
    & =  \sum_{  \substack{Q_X \in \mathcal{P}_n(\mathcal{X}) : \\ H(Q_X) \geq R}  } e^{-n( D(Q_X\|P_X) + H(Q_X) )} \times  \sum_{i=1}^{k}  e^{n(H(Q_X) - R)} \times \bar{\lambda}(n, H(Q_X) - R, B, \pi_i[\mathcal{C}_n(Q_X)]) \label{eq:source_coding_coded_side_information_IB_codebook}
\end{align}
where \eqref{eq:source_coding_coded_side_information_type_sequence_probability} is because   $\P \{ \bm{X} = \bm{x}\} =  e^{-n( D(Q_X\|P_X) + H(Q_X) )}$ for every \(\bm{x} \in \mathcal{T}_n(Q_X)\);
\eqref{eq:source_coding_coded_side_information_union_permutation} is due to
\begin{equation}
    \bigcup_{i=1}^{k} \pi_i [ \mathcal{C}_n(Q_X) ] = \mathcal{T}_n(Q_X);
\end{equation}
which holds by codebook construction;
and \eqref{eq:source_coding_coded_side_information_IB_codebook} holds since we are using the IB channel's coding scheme, and hence for codebook $\pi_i [\mathcal{C}_n(Q_X)]$, the average error is given by
\begin{align}
\frac{1}{e^{n(H(Q_X) - R)}} \sum_{ \bm{x} \in \pi_i [\mathcal{C}_n(Q_X)] } \P \{ \hat{\bm{X}} \neq \bm{x} | \bm{X} = \bm{x} \} = \bar{\lambda}(n, H(Q_X) - R, B, \pi_i[\mathcal{C}_n(Q_X)]) \nonumber.
\end{align}
Now define the mean decoding error probability over the sequence of permuted codebooks as
\begin{align}
    \bar{\lambda}^{(\pi)}(n, H(Q_X) - R, B, \mathcal{C}_n(Q_X))  \triangleq \frac{1}{k} \sum_{ i=1 }^{k} \bar{\lambda}(n, H(Q_X) - R, B, \pi_i[\mathcal{C}_n(Q_X)]),
\end{align}
and recall that \(k \ndot{=} e^{nR}\). Plugging these back into \eqref{eq:source_coding_coded_side_information_IB_codebook}, we obtain
\begin{align}
    & \P \{ \hat{\bm{X}} \neq \bm{X}\} \nonumber \\
    & \leq \sum_{  \substack{Q_X \in \mathcal{P}_n(\mathcal{X}) : \\ H(Q_X) \geq R}  } e^{-n( D(Q_X\|P_X) + H(Q_X) )} \times  e^{n(H(Q_X) - R)} \times k \times  \bar{\lambda}^{(\pi)}(n, H(Q_X) - R, B, \mathcal{C}_n(Q_X)) \label{eq:source_coding_coded_side_information_codebook_permutation_vector} \\
    & \ndot{ = } \max_{ \substack{Q_X \in \mathcal{P}_n(\mathcal{X}) : \\ H(Q_X) \geq R} } e^{-n D(Q_X\|P_X) } \times \bar{\lambda}^{(\pi)}(n, H(Q_X) - R, B, \mathcal{C}_n(Q_X)). \label{eq:source_coding_coded_side_information_asymptotic_k}
\end{align}

We now wish to find an upper bound for \(\bar{\lambda}^{(\pi)}(n, H(Q_X) - R, B, \mathcal{C}_n(Q_X))\), whose value clearly depends on \(\mathcal{C}_n(Q_X)\) we selected.
To find a good codebook \(\mathcal{C}_n(Q_X)\), we use a random coding argument and take the ensemble average over $\bm{C}$, uniformly distributed on $\mathcal{T}_n(Q_X)^{e^{n(H(Q_X) - R)}}$.
Observe that
\begin{align}
\E [ \bar{\lambda}^{(\pi)}(n, H(Q_X) - R, B, \bm{C}) ] & = \E \bigg[  \frac{1}{k} \sum_{ i=1 }^{k} \bar{\lambda}(n, H(Q_X) - R, B, \pi_i[\bm{C}])  \bigg]  \label{eq:source_coding_coded_side_information_ensemble_lambda_pi_definition}\\
    & = \frac{1}{k} \sum_{ i=1 }^{k} \E[ \bar{\lambda}(n, H(Q_X) - R, B, \pi_i[\bm{C}]) ] \\
    & = \frac{1}{k} \sum_{ i=1 }^{k} \E[ \bar{\lambda}(n, H(Q_X) - R, B, \bm{C}) ] \label{eq:source_coding_coded_side_information_ensemble_permutation_invariant} \\
    & = \E[ \bar{\lambda}(n, H(Q_X) - R, B, \bm{C}) ] \\
    & = \bar{\lambda}(n, H(Q_X) - R, B),
\end{align}
where \eqref{eq:source_coding_coded_side_information_ensemble_lambda_pi_definition} is due to the definition of \( \bar{\lambda}^{(\pi)}(n, H(Q_X)-R, B, \mathcal{C}_n)\); and in \eqref{eq:source_coding_coded_side_information_ensemble_permutation_invariant} we observe that the constant composition random codebook \(\bm{C}\) is invariant under permutations, i.e., for any permutation \(\pi\), \(\pi[\bm{C}]\) has the same distribution as \(\bm{C}\).
Moreover, in Section \ref{sec:lower_bound_proof}, we have shown that for the constant composition ensemble  with codeword composition \(Q_X\), the compress-forward strategy under the MMI decoder produces an ensemble-average error probability satisfying
\begin{align}
    \bar{\lambda}(n, H(Q_X) - R, B) & \ndot{\leq} \max_{Q_Y \in \mathcal{P}_n(\mathcal{Y})} \max_{ \substack{Q_{X|YU} }  } \exp \big\{ -n \big(  D(Q_{Y|X}  \| P_{Y|X} | Q_X) + I_Q(X;U|Y) + \nonumber \\
    & \hspace{5.5cm} \big | R - H_Q(U|X) - |I_Q(Y;U) - B |^{+}  \big| ^{+}   \big) \big\}, \label{eq:source_coding_coded_side_information_type_error_exponent}
\end{align}
if we do not include the optimization over \(P_{U|Y}\).
In the conventional random coding argument, one would  argue that \eqref{eq:source_coding_coded_side_information_type_error_exponent} implies that there exists at least one codebook \(\mathcal{C}_n(Q_X)\) such that
\begin{align}
    & \bar{\lambda}^{(\pi)}(n, H(Q_X) - R, B, \mathcal{C}_n(Q_X)) \nonumber \\
    & \ndot{\leq} \max_{Q_Y \in \mathcal{P}_n(\mathcal{Y})} \max_{ \substack{Q_{X|YU}  }  }  \exp \big\{ -n \big(  D(Q_{Y|X}  \| P_{Y|X} | Q_X) + I_Q(X;U|Y) + \nonumber \\
    & \hspace{5.5cm} \big | R - H_Q(U|X) - |I_Q(Y;U) - B |^{+}  \big| ^{+}   \big) \big\}. \label{eq:source_coding_coded_side_information_codebook_error_exponent}
\end{align}
However, recall the assumption we made when constructing the encoder that \(\mathcal{C}_n(Q_X)\) must also satisfy
\begin{equation}
    \bigcup_{i=1}^{k} \pi_i [ \mathcal{C}_n ] = \mathcal{T}_n(Q_X). \label{eq:source_coding_coded_side_information_codebook_requirement_covering}
\end{equation}
Thus, we need to find a codebook \(\mathcal{C}_n(Q_X)\) such that both \eqref{eq:source_coding_coded_side_information_codebook_error_exponent} and \eqref{eq:source_coding_coded_side_information_codebook_requirement_covering} hold at the same time.
This is accomplished through the expurgation technique together with Theorem \ref{thm:Ahlswede_covering_lemma_ensemble}.

Recall that we selected the sequence of permutations \(\pi_1, \pi_2, \ldots, \pi_k\) according to Theorem \ref{thm:Ahlswede_covering_lemma_ensemble}, i.e., for this specific sequence of permutations, the covering property in \eqref{eq:source_coding_coded_side_information_codebook_requirement_covering} holds for at least half of codebooks \(\mathcal{C}_n(Q_X) \in \mathcal{T}_n(Q_X)^{e^{n(H(Q_X) - R)}}\) in the constant composition ensemble.
On the other hand, through the expurgation technique, we can show that by getting rid of the worst one third of codebooks in the ensemble, the remaining two thirds of codebooks satisfy \eqref{eq:source_coding_coded_side_information_codebook_error_exponent}.
Since \(\frac{1}{2} + \frac{2}{3} > 1\), there must be an overlap between the two sets of codebooks, i.e.,  there must exist a \(\mathcal{C}_n(Q_X)\) such that \eqref{eq:source_coding_coded_side_information_codebook_error_exponent} and \eqref{eq:source_coding_coded_side_information_codebook_requirement_covering} hold at the same time.
By selecting such \(\mathcal{C}_n(Q_X)\), we can substitute \eqref{eq:source_coding_coded_side_information_codebook_error_exponent} into \eqref{eq:source_coding_coded_side_information_asymptotic_k}, which leads to
\begin{align}
    & \P \{ \hat{\bm{X}} \neq \bm{X}\} \nonumber \\
    & \ndot{ \leq } \max_{ \substack{Q_X \in \mathcal{P}_n(\mathcal{X}) : \\ H(Q_X) \geq R} } \max_{Q_Y \in \mathcal{P}_n(\mathcal{Y})} \max_{ \substack{Q_{X|YU} }  }  \exp \big\{ -n \big(  D(Q_{XY}  \| P_{XY}) + I_Q(X;U|Y) + \nonumber \\
    & \hspace{6.5cm} \big | R - H_Q(U|X) - |I_Q(Y;U) - B |^{+}  \big| ^{+}   \big) \big\} \\
    & = \max_{Q_Y \in \mathcal{P}_n(\mathcal{Y})}  \max_{ \substack{Q_{X|YU} : \\  H(Q_X) \geq R}  }  \exp \big\{ -n \big(  D(Q_{XY}  \| P_{XY}) + I_Q(X;U|Y) + \nonumber \\
    & \hspace{6.5cm} \big | R - H_Q(U|X) - |I_Q(Y;U) - B |^{+}  \big| ^{+}   \big) \big\} \\
    & = \max_{Q_Y \in \mathcal{P}_n(\mathcal{Y})} \min_{P_{U|Y}}  \max_{ \substack{Q_{X|YU} : \\  H(Q_X) \geq R}  }  \exp \big\{ -n \big(  D(Q_{XY}  \| P_{XY}) + I_Q(X;U|Y) + \nonumber \\
    & \hspace{6.5cm} \big | R - H_Q(U|X) - |I_Q(Y;U) - B |^{+}  \big| ^{+}   \big) \big\} \label{eq:source_coding_side_information_optimization_over_P_U_Y},
\end{align}
where in \eqref{eq:source_coding_side_information_optimization_over_P_U_Y} we select the best \(P_{U|Y}\) for every \(Q_Y \in \mathcal{P}_n(\mathcal{Y})\) when constructing the scheme.
This completes the proof for Theorem \ref{thm:source_coding_coded_side_information_exponent_lower_bound}.


\subsection{Mismatched Decoding}
\label{sec:source_coding_coded_side_information_mismatch}
We now consider the WAK problem under a mismatched decoding rule. For every index \(l\) forwarded from the helper, the receiver is required to reconstruct a certain sequence \(\bm{u}\). Given an index \(i\) from the transmitter, it adopts the following decoding rule
\begin{equation}
    \hat{\bm{x}} = \argmax_{ \bm{x} \in f^{\prime}_n(i)^{-1}, \bm{u}  } g(\hat{P}_{\bm{x}}, \hat{P}_{\bm{u} | \bm{x}}),
\end{equation}
where \(f^{\prime}_n(i)^{-1} =\{  \bm{x} : f^{\prime}_n( \bm{x} ) = i\} \) and
\begin{equation}
    g(\hat{P}_{\bm{x}}, \hat{P}_{\bm{u} | \bm{x}} ) = \sum_{x,u} \hat{P}_{\bm{x} \bm{u}}(x,u)\log q(x, u),
\end{equation}
in which \(q(x, u)\) is some decoding metric.
Since the decoder in Section \ref{sec:lower_bound_proof} includes the mismatched decoder, we have this immediate result.
\begin{theorem}
    \label{thm:source_coding_coded_side_information_mismatch_exponent}
    For the DMS pair \((X^n, Y^n)\), under a mismatched decoding rule, we have the following achievable error exponent
    \begin{align}
        & \liminf_{n \to \infty} -\frac{1}{n} \log \lambda^{\prime}(n ,R, B) \nonumber \\
        & \geq \min_{Q_Y} \max_{Q_{U|Y}} \min_{ \substack{Q_{X|YU} : \\ H(Q_X) \geq R}  }  D(Q_{XY} \| P_{XY}) + I_Q(X;U|Y) + \nonumber \\
        & \hspace{6.5cm} \big | R +  E_0(Q_X, Q_{U|X}) - H_Q(X) - |I_Q(Y;U) - B |^{+}  \big| ^{+},
    \end{align}
    where \(E_0(Q_X, Q_{U|X})\) is given by  \eqref{eq:lower_bound_error_exponent_generalized_decoder_E_0}.
\end{theorem}
Following the proof of Corollary \ref{cor:achievable_rate}, it can be verified that this exponent leads to the following achievable rates of the mismatched WAK problem.
\begin{corollary}
    \label{cor:source_coding_coded_side_information_mismatch_rate}
    For the DMS pair \((X^n, Y^n)\), under a mismatched decoding rule, all rates up to \(R_{\textnormal{LM}}(B)\) are achievable, where
    \begin{equation}
        R_{\textnormal{LM}}(B) = H(P_X) - \max_{ P_{U|Y}} E_0(P_X, P_{U|X}) \qquad \text{s.t.} \qquad  I(P_{Y}, P_{U|Y}) \leq B,
    \end{equation}
    in which \(X \overset{P_{Y|X}}{\to} Y \overset{P_{U|Y}}{\to} U\) forms a Markov chain and \(E_0(P_X, P_{U|X})\) is given by \eqref{eq:lower_bound_error_exponent_generalized_decoder_E_0}.
\end{corollary}

\section{Concluding Remarks}
\label{sec:conclusion}
In this work, we studied the error exponent of the IB channel under constant composition codes.
We established an achievable error exponent, showed that employing constant composition codes does not improve the rates achieved with IID codes, and then provided an upper bound for all achievable error exponents under constant composition codes.
We further explored the connections between the IB channel and the WAK problem.
In particular, we demonstrated that the helper in the WAK problem can be viewed as an oblivious relay, and codes developed for the IB channel can be transformed into codes for the WAK problem, owing to the simultaneous covering lemma.
Achievable error exponents and rates for the IB channel and the WAK problem under mismatched decoding rules were also derived.
We now conclude with a discussion of potential future work.

\emph{1)} In our preliminary work \cite{wuAchievableErrorExponent2024}, we established an achievable error exponent for the IB channel under constant composition codes.
The achievable error exponent in \cite[Theorem 1]{wuAchievableErrorExponent2024} was established through random generation of compress-forward codebooks and showing that the compress-forward strategy can be modeled as a DMC.
The achievable error exponent in this work, i.e., Theorem \ref{thm:lower_bound}, was established by the type covering lemma as well as a more refined analysis through the method of types.
We believe that the achievable error exponent provided in Theorem \ref{thm:lower_bound} is generally superior to the one in \cite[Theorem 1]{wuAchievableErrorExponent2024}. However, the two achievable exponents are not directly comparable, due to the different philosophies of the proofs behind them.
It is of interest to provide a proof to support this claim.

\emph{2)} It is desirable to improve the sphere packing bound provided in this work, i.e., Theorem \ref{thm:sphere_packing_bound}.
The achievable error exponent and sphere packing bound established in this work are not easily comparable.
For example, it is certain that the two bounds will meet at the capacity \(C(B)\), but it is unclear whether there exists a critical rate, strictly below \(C(B)\), above which the two bounds coincide.
One reason for the lack of comparability is that the term \(I(X;U|Y)\) is missing from \(E_{\mathrm{sp}}\) in \eqref{eq:main_results_definition_E_sp}.
Recall that \(I(X;U|Y)\) appears in the achievable error exponent to measure the performance of the compress-forward strategy.
However, the sphere packing bound established in this work is built upon the weak converse, meaning that we are restricted to \(X \to Y \to U\), i.e., \(I(X;U|Y) = 0\).
It is plausible that our sphere packing bound can be strengthened to incorporate \(I(X;U|Y)\) as follows
\begin{equation}
    E_{\textnormal{sp}}(R, B, P_X) = \min_{Q_{Y}} \max_{ \substack{ P_{U|Y}: \\ I(Q_Y, P_{U|Y}) \leq B } }  \min_{ \substack{Q_{X|UY}: \\ Q_X = P_{X}, \\ I_{Q}(X;U)  \leq R } } \! \! D(Q_{Y|X} \| P_{Y|X} | P_{X}) + I_Q(X;U|Y), \label{eq:conclusion_improved_sphere_packing_bound}
\end{equation}
bearing closer resemblance to the achievable error exponent in \eqref{eq:achievability_definition_E_r_R_B_P_X}.
A similar improvement for the WAK problem appeared in \cite{kangUpperBoundError2018}, providing some affirmation for our claim that \eqref{eq:conclusion_improved_sphere_packing_bound} is valid.

\emph{3)} It is of interest to derive a strong converse and an exponential strong converse for the IB channel.
Recently, a tight exponential strong converse for the WAK problem was established in \cite{takeuchiTightExponentialStrong2025} by leveraging the change of measure method developed in \cite{tyagiStrongConverseUsing2020}.
Due to the deep connection between the two problems, it is conceivable that a tight exponential strong converse for the IB channel can also be derived.

\emph{4)} It may be useful to establish the typical error exponent for the IB channel under constant composition codes.
As discussed earlier, the error exponent considered in this work resembles the random coding error exponent.
Another important performance metric for random codebook ensembles is the typical error exponent \cite{merhavErrorExponentsTypical2018}, where the focus is on the expectation of error exponents within an ensemble.
For the IB channel, the random codebook ensemble is not an input strategy but rather represents the employed transmission codebook that the relay is oblivious to.
Thus, the typical error exponent of the random ensemble can be interpreted as the error exponent of a typical transmission codebook, i.e., the error exponent of a typical user, which may be of potential practical interest.
\appendices
\section{Proofs of Constant Composition Distribution Properties}
\subsection{Proof of Lemma \ref{lemma:marginal_ditribution_constant_composition_ensemble}}
\label{appendix:proof_marginal_distribution_lemma}
The lemma follows from a counting argument.
For every \(b \in \mathcal{X}\) and \(i \in [n]\),
\begin{align}
    P_{X_i}(b) &= \sum_{x^n} \idc{x_i = b} \times P_{X^n}(x^n) \\
    &= \frac{\abs{\{x^n \in \mathcal{T}_n(P_X): x_i = b\}}}{\abs{\mathcal{T}_n(P_X)} } \\
    &= \frac{ \frac{(n-1)!}{(nP_X(b)-1)! \times \prod_{a \in \mathcal{X}, a \neq b}(nP_X(a)!)} }{ \frac{n!}{\prod_{a \in \mathcal{X}}(nP_X(a)!)} } \\
    &= \frac{nP_X(b)}{n} \\
    &= P_X(b),
\end{align}
which completes the proof.

\subsection{Proof of Lemma \ref{lemma:conditional_distribution_constant_composition_ensemble}}
\label{appendix:proof_conditional_distribution_lemma}
The support of \(P_{X^i}\) follows immediately from \eqref{eq:weak_converse_constant_ensemble_Q_hat}.
As for the conditional distribution, observe that the total number of  \(x^n \in \mathcal{T}_n(P_X)\) with prefix being \(x^i\) is \(\abs{\mathcal{T}_{n-i} (Q_X^{\ast}) }\), i.e., the cardinality of all possible suffixes under the prefix \(x^i\).
Since each such sequence \(x^n\) is equally probable under \(P_{X^n}\),  we have
\begin{equation}
    P_{X^{i}}(x^i) = \frac{ \abs{\mathcal{T}_{n-i}(Q_X^{\ast})}   }{ \abs{\mathcal{T}_n(P_X)} }.
\end{equation}
Similarly, for every \(a \in \mathcal{X}\), the probability of all possible sequences \(x^n \in \mathcal{T}_n(P_X)\) with prefix being \(x^i\) as well as \(x_{i+1} = a\) is given by
\begin{equation}
    P_{X^{i}, X_{i+1} }(x^{i}, a) = \frac{  \abs{ \{x^n_{i+1} \in \mathcal{T}_{n-i}(Q_X^{\ast}): x_{i+1} = a \} } }{  \abs{\mathcal{T}_n(P_X)}   }.
\end{equation}
Therefore, we conclude that
\begin{align}
    P_{X_{i+1}|X^{i}}(a | x^{i}) &= \frac{P_{X^{i}, X_{i+1} }(x^{i+1}, a)}{ P_{X^{i}}(x^i)} \\
    & = \frac{\abs{ \{x^n_{i+1} \in \mathcal{T}_{n-i}(Q_X^{\ast}): x_{i+1} = a \} } }{\abs{\mathcal{T}_{n-i}(Q_X^{\ast})}} \\
    &= Q_X^{\ast}(a), \label{eq:appendix_proof_conditional_distribution_lemma_marginal_distribution}
\end{align}
where \eqref{eq:appendix_proof_conditional_distribution_lemma_marginal_distribution} follows from applying Lemma \ref{lemma:marginal_ditribution_constant_composition_ensemble} to the suffix distribution.

\subsection{Proof of Lemma \ref{lemma:weak_converse_constant_ensemble_window_typical}}
\label{appendix:proof_window_typical_constant_ensemble}
We first prove the lemma for \(k=1\), i.e.,
\begin{equation}
    \P\{X^n \in \mathcal{T}_n(P_X): X^{i} \notin \mathcal{T}_i^{\delta}(P_X)\} \leq 2\abs{\mathcal{X}}e^{\abs{\mathcal{X}}\log(n+1) - i\delta^2 P_{\min}^2}.
\end{equation}
Due to \eqref{eq:weak_converse_constant_ensemble_Q_hat}, for each prefix type \(Q_X \in \mathcal{S}_i(\mathcal{X})\), there exists a unique suffix type \(Q_X^{\ast} \in \mathcal{S}_{n-i}(\mathcal{X})\) with
\begin{equation}
    iQ_X(a) + (n-i)Q_X^{\ast}(a) = nP_X(a), \qquad \forall a \in \mathcal{X}. \label{eq:weak_converse_prefix_type_suffix_type_relationship}
\end{equation}
Given any prefix \(x^i\) satisfying \(\hat{P}_{x^i} = Q_X \in \mathcal{S}_i(\mathcal{X})\), the total number of  \(x^n \in \mathcal{T}_n(P_X)\) with prefix being \(x^i\) is \(\abs{\mathcal{T}_{n-i} (Q_X^{\ast}) }\).
Since the cardinality of such prefixes \(x^i\) satisfying \(\hat{P}_{x^i} = Q_X\) is \(|\mathcal{T}_i(Q_X)|\), we see that
\begin{equation}
    |\{x^n \in \mathcal{T}_n(P_X): \hat{P}_{x^i} = Q_X\}  | = \abs{\mathcal{T}_i(Q_X)} \times |\mathcal{T}_{n-i}(Q_X^{\ast})|.
\end{equation}
Thus, under \(P_{X^n}\), the probability of sequences \(x^n\) with prefix type being  \(Q_X \in \mathcal{S}_i(\mathcal{X})\) is
\begin{equation}
    \P\{X^n \in \mathcal{T}_n(P_X): \hat{P}_{X^i} = Q_X\} = \frac{ \abs{\mathcal{T}_i(Q_X)} \times |\mathcal{T}_{n-i}(Q_X^{\ast})| }{ \abs{\mathcal{T}_n(P_X)} }. \label{eq:weak_converse_probability_prefix_type_combinatorial}
\end{equation}
Now recall that the probability of any sequence \(x^n\) satisfying \(\hat{P}_{x^n} = P_X\) under the IID distribution \(P_X^n\) is
\begin{equation}
    \prod_{a \in \mathcal{X}}P_X(a)^{nP_X(a)} = e^{-nH(P_X)}.
\end{equation}
First, it is evident that
\begin{align}
    &\P\{X^n \in \mathcal{T}_n(P_X):  \hat{P}_{X^i} = Q_X\}  = \frac{ \abs{\mathcal{T}_i(Q_X)} \times |\mathcal{T}_{n-i}(Q_X^{\ast})| \times e^{-nH(P_X)}}{ \abs{\mathcal{T}_n(P_X)} \times e^{-nH(P_X)}}.
\end{align}
Due to \eqref{eq:weak_converse_prefix_type_suffix_type_relationship}, we have
\begin{align}
    & \abs{\mathcal{T}_i(Q_X)} \times |\mathcal{T}_{n-i}(Q_X^{\ast})| \times e^{-nH(P_X)} \nonumber \\
    & = \abs{\mathcal{T}_i(Q_X)} \times |\mathcal{T}_{n-i}(Q_X^{\ast})| \times \prod_{a \in \mathcal{X}}P_X(a)^{nP_X(a)} \\
    & = \abs{\mathcal{T}_i(Q_X)} \times \prod_{a \in \mathcal{X}}P_X(a)^{iQ_X(a)} \times |\mathcal{T}_{n-i}(Q_X^{\ast})| \times \prod_{a \in \mathcal{X}}P_X(a)^{(n-i)Q_X^{\ast}(a)} \\
    & = P_X^i[\mathcal{T}_i(Q_X)] \times P_X^{n-i}[\mathcal{T}_{n-i}(Q_X^{\ast})].
\end{align}
Hence, it follows that
\begin{align}
    \P\{X^n \in \mathcal{T}_n(P_X):  \hat{P}_{X^i} = Q_X\}
    & = \frac{ P_X^i[\mathcal{T}_i(Q_X)] \times P_X^{n-i}[\mathcal{T}_{n-i}(Q_X^{\ast})] }{ P_X^n[\mathcal{T}_n(P_X)] } \\
    & \leq \frac{ P_X^i[\mathcal{T}_i(Q_X)]}{ P_X^n[\mathcal{T}_n(P_X)] } \\
    & \leq (n+1)^{\abs{\mathcal{X}}}P_X^i[\mathcal{T}_i(Q_X)],
\end{align}
where we notice \(P_X^{n-i}[\mathcal{T}_{n-i}(Q_X^{\ast})]  \leq 1\) and \(P_X^n[\mathcal{T}_n(P_X)] \geq (n+1)^{-\abs{\mathcal{X}}}\) (see, e.g., \cite[Lemma 2.3]{csiszarInformationTheoryCoding2011}).
Thus, we can proceed with
\begin{align}
    &\P\{ X^n \in \mathcal{T}_n(P_X): X^i \notin \mathcal{T}_i^{\delta}(P_X)\} \nonumber \\
    &= \sum_{Q_X \in \mathcal{S}_i(\mathcal{X}): \exists x^i \notin \mathcal{T}_i^{\delta}(P_X), \hat{P}_{x^i} = Q_X}  \P\{X^n \in \mathcal{T}_n(P_X):  \hat{P}_{X^i} = Q_X\}  \\
    & \leq  \sum_{Q_X \in \mathcal{S}_i(\mathcal{X}): \exists x^i \notin \mathcal{T}_i^{\delta}(P_X), \hat{P}_{x^i} = Q_X} (n+1)^{\abs{\mathcal{X}}}P_X^i[\mathcal{T}_i(Q_X)]\\
    &\leq (n+1)^{\abs{\mathcal{X}}}(1 - P_X^i[\mathcal{T}_i^{\delta}(P_X)] )  \label{eq:weak_converse_constant_ensemble_prefix_type_to_all_type} ,
\end{align}
where \eqref{eq:weak_converse_constant_ensemble_prefix_type_to_all_type} is due to \(\mathcal{S}_i(\mathcal{X}) \subseteq \mathcal{P}_i(\mathcal{X})\).
Consider the following upper bound
\begin{align}
    &1 - P_X^i[\mathcal{T}_i^{\delta}(P_X)] \nonumber \\
    & = \sum_{x^i}P_X^i(x^i) \times \idc{ \exists a \in \mathcal{X}, |P_{x^i}(a)-P_X(a)| > \delta P_X(a)} \\
    & = \sum_{x^i}P_X^i(x^i) \times \idc{\exists a \in \supp(P_X), |P_{x^i}(a)-P_X(a)| > \delta P_X(a)}  \label{eq:weak_converse_constant_ensemble_enlarge_set_a} \\
    & \leq  \sum_{a \in \mathcal{X} : P_X(a)>0 } 2e^{-i\delta^2P_X^2(a)} \label{eq:weak_converse_constant_ensemble_Hoeffding_inequality} \\
    &\leq  2\abs{\mathcal{X}}e^{-i\delta^2P_{\min}^2},
\end{align}
where in \eqref{eq:weak_converse_constant_ensemble_enlarge_set_a}, if \(P_{X}(a) = 0\) then \(P_{X^i}(x^{i}) = 0\) for all \(x^{i}\) with \(\hat{P}_{x^i}(a) > 0\) , i.e., we only need to consider \(x^i\) whose entries are from \(\supp(P_X)\); in \eqref{eq:weak_converse_constant_ensemble_Hoeffding_inequality} we make use of the union bound and \(\P\{\abs{N-kq}>k\delta\} \leq 2 e^{-2\delta^2 k}\), where the latter follows from \cite[Problem 3.18(b)]{csiszarInformationTheoryCoding2011}.
Thus,  we can conclude that
\begin{equation}
    \P\{X^n \in \mathcal{T}_n(P_X): X^{i} \notin \mathcal{T}_i^{\delta}(P_X)\} \leq 2\abs{\mathcal{X}}e^{\abs{\mathcal{X}}\log(n+1) - i\delta^2 P_{\min}^2}.
\end{equation}
The proof is completed after noticing that the same reasoning applies to any \(k > 1\).
\subsection{Proof of Corollary \ref{cor:weak_converse_constant_ensemble_prefix_suffix_typical}}
\label{appendix:proof_prefix_suffix_typical_constant_ensemble}
Following from Lemma \ref{lemma:weak_converse_constant_ensemble_window_typical}, we have
\begin{align}
    \P\{X^n \in \mathcal{T}_n(P_X): X^{i} \notin \mathcal{T}_i^{\delta_n}(P_X)\} &\leq 2\abs{\mathcal{X}}e^{\abs{\mathcal{X}}\log(n+1) - i\delta_n^2P_{\min}^2 }.
\end{align}
By choosing \(\delta_n = n^{-\frac{1}{8}}\) and noticing \( i \geq \sqrt{n}\), we see that
\begin{align}
    \P\{X^n \in \mathcal{T}_n(P_X): X^i \notin \mathcal{T}_i^{\delta_n}(P_X)\} \leq 2\abs{\mathcal{X}}e^{\abs{\mathcal{X}}\log(n+1) - n^{\frac{1 }{4}}P_{\min}^2}.
\end{align}
In the same manner, we obtain
\begin{align}
    \P\{X^n \in \mathcal{T}_n(P_X): X^{n}_{i+1} \notin \mathcal{T}_{n-i}^{\delta_n}(P_X)\} \leq 2\abs{\mathcal{X}}e^{\abs{\mathcal{X}}\log(n+1) - n^{\frac{1 }{4}}P_{\min}^2} ,
\end{align}
where we notice \(n-i \geq \sqrt{n}\).
From the union bound, the probability for sequences \(x^n\) with either prefix \(x^i\) or suffix \(x^{n}_{i+1}\) being non-typical can be upper bounded through
\begin{align}
    & \P\{X^n \in \mathcal{T}_n(P_X): X^{i} \notin \mathcal{T}_i^{\delta_n}(P_X) \ \text{or} \ X^{n}_{i+1} \notin \mathcal{T}_{n-i}^{\delta_n}(P_X)\} \nonumber \\
    & \leq \P\{X^n \in \mathcal{T}_n(P_X): X^i \notin \mathcal{T}_i^{\delta_n}(P_X)\} + \P\{X^n \in \mathcal{T}_n(P_X): X^n_{i+1} \notin \mathcal{T}_{i+1}^{\delta_n}(P_X)\} \\
    & \leq 4\abs{\mathcal{X}}e^{\abs{\mathcal{X}}\log(n+1) - n^{\frac{1 }{4}}P_{\min}^2}.
\end{align}
Consequently, we have
\begin{align}
    &\P\{X^n \in \mathcal{T}_n(P_X): X^{i} \in \mathcal{T}_i^{\delta_n}(P_X), X^{n}_{i+1} \in \mathcal{T}_{n-i}^{\delta_n}(P_X)\} \nonumber \\
    &= 1- \P\{X^n \in \mathcal{T}_n(P_X): X^{i} \notin \mathcal{T}_i^{\delta_n}(P_X) \ \text{or} \ X^{n}_{i+1} \notin \mathcal{T}_{n-i}^{\delta_n}(P_X)\} \\
    &\geq 1- 4\abs{\mathcal{X}}e^{\abs{\mathcal{X}}\log(n+1) - n^{\frac{1 }{4}}P_{\min}^2},
\end{align}
which completes the proof.

\section{Proof of Lemma \ref{lemma:weak_converse_replacing}}
\label{appendix:weak_converse_replacing_lemma}

First, notice that
\begin{align}
H(Y|Z, X) & = \sum_{x \in \mathcal{X}} P_{X}(x) H(Y | Z, X = x)   \\
    & = \sum_{x \in \mathcal{E}} P_{X}(x)H(Y | Z, X = x) + \sum_{x \in \mathcal{X} - \mathcal{E}} P_{X}(x)H(Y | Z, X = x) \\
    & \leq \sum_{x \in \mathcal{E}} P_{X}(x)H(Y | Z, X = x)  + (1-P_X[\mathcal{E}])\log |\mathcal{Y}|. \label{eq:weak_converse_replacing_lemma}
\end{align}
Next, for every \(x \in \mathcal{E}\), it is easy to verify that \(P_{ZY | X}( \cdot | x) \overset{\delta}{\sim} \tilde{P}_{Z Y | X}(\cdot | x)\). Thus, after marginalizing, we have \(P_{Z|X}(\cdot| x)  \overset{\delta}{\sim} \tilde{P}_{Z| X}(\cdot | x) \), which means that for every \(x \in \mathcal{E}\)
\begin{align}
    H(Y | Z, X = x) & = \sum_{z \in \mathcal{Z} } P_{Z|X}(z|x) H(Y | Z = z, X =x) \\
    & \leq  \sum_{z \in \mathcal{Z} } (1+\delta) \tilde{P}_{Z}(z | x) H(Y | Z = z, X =x).\label{eq:weak_converse_replacing_lemma_x}
\end{align}
On the other hand, for every \(x \in \mathcal{E}\), \(y \in \mathcal{Y}\), and \(z \in \mathcal{Z}\), we have
\begin{align}
    P_{Y|Z,X}(y | z, x) & = \frac{P_{ZY|X}(z,y | x)}{P_{Z|X}(z|x)}  \\
    & \leq \frac{(1+\delta)\tilde{P}_{ZY|X}(z, y | x)}{(1-\delta)\tilde{P}_{Z|X}(z|x)} \\
    & = (1 + \frac{2\delta}{1-\delta})\tilde{P}_{Y|Z,X}(y|z,x).
\end{align}
Similarly, we also have
\begin{equation}
    P_{Y|Z,X}(y | z, x)  \geq (1 - \frac{2\delta}{1 + \delta})\tilde{P}_{Y|Z,X}(y|z,x).
\end{equation}
Conditioned on \(\delta \in (0,1)\), we have \( \frac{2\delta}{1-\delta} > \frac{2\delta}{1 + \delta} \).
We conclude that for every \(x \in \mathcal{E}\) and \(z \in \mathcal{Z}\) it holds that
\begin{equation}
    P_{Y|Z,X}(\cdot | z, x) \overset{ \frac{2\delta}{1-\delta}  }{ \sim } \tilde{P}_{Y|Z,X}( \cdot |z,x).
\end{equation}
Thus, through \cite[Lemma 2.7]{csiszarInformationTheoryCoding2011}, we can proceed from \eqref{eq:weak_converse_replacing_lemma_x} with for every \(x \in \mathcal{E}\)
\begin{align}
    H(Y | Z, X = x) & \leq  \sum_{z \in \mathcal{Z} } (1+\delta) \tilde{P}_{Z|X}(z|x) H(Y | Z = z, X =x) \\
    & \leq \sum_{z \in \mathcal{Z} } (1+\delta) \tilde{P}_{Z|X}(z|x) H(\tilde{Y} | Z = z, X =x) - \frac{2\delta(1+\delta)}{1-\delta} \log \frac{2\delta}{(1-\delta)\abs{\mathcal{Y}}} \\
    & \leq H(\tilde{Y} | \tilde{Z}, X =x)   + \delta \log |\mathcal{Y}| - \frac{2\delta(1+\delta)}{1-\delta} \log \frac{2\delta}{(1-\delta)\abs{\mathcal{Y}}} . \label{eq:weak_converse_replacing_lemma_x2}
\end{align}
Substituting \eqref{eq:weak_converse_replacing_lemma_x2} into \eqref{eq:weak_converse_replacing_lemma_x}, we see that
\begin{align}
    H(Y|Z, X) & \leq \sum_{x \in \mathcal{E}} P_{X}(x) H(Y | Z, X = x)  + (1-P_X[\mathcal{E}])\log |\mathcal{Y}| \\
    & \leq \sum_{x \in \mathcal{E}} P_{X}(x) H(\tilde{Y} | \tilde{Z}, X =x)  + \delta \log |\mathcal{Y}| - \frac{2\delta(1+\delta)}{1-\delta} \log \frac{2\delta}{(1-\delta)\abs{\mathcal{Y}}} + (1-P_X[\mathcal{E}])\log |\mathcal{Y}|\\
    & \leq H(\tilde{Y} | \tilde{Z}, X )  + \delta \log |\mathcal{Y}| - \frac{2\delta(1+\delta)}{1-\delta} \log \frac{2\delta}{(1-\delta)\abs{\mathcal{Y}}}  + (1-P_X[\mathcal{E}])\log |\mathcal{Y}|. \label{eq:weak_converse_replacing_lemma_final}
\end{align}
A similar lower bound between  \(H(Y|Z, X)\) and \( H(\tilde{Y} | \tilde{Z}, X )\) can also be obtained in the same fashion.

It is worthwhile noting that the cardinality of \(\mathcal{Z}\) can be very large when we employ this lemma.
Through the use of robust typicality, we avoid considering the cardinality of \(\mathcal{Z}\) when changing the underlying pmf for the conditional entropy, which is seen from \eqref{eq:weak_converse_replacing_lemma_x}. This may cause issues if strong typicality is used.
\section{Proofs of Cardinality Bounds}
\subsection{Capacity}
\label{appendix:weak_converse_cardinality_bound_U}
For every \((P_{X}, P_{U|Y})\), we obtain the following two distributions through the Markov chain: \(P_{Y} = P_{X} \cdot P_{Y|X}\) and \(P_{U} = P_{Y} \cdot P_{U|Y}\).
Then, we can rewrite the Markov chain as \(U \overset{P_{Y|U}}{\to}  Y \overset{P_{X|Y}}{\to} X \), where \(P_{Y|U}\) is the reverse channel induced by \(P_{Y}\) and \(P_{U|Y}\), while \(P_{X|Y}\) is the reverse channel induced by \(P_X\) and \(P_{Y|X}\).
Consider the following \(\abs{\mathcal{Y}} + 1\) continuous functions on \(\mathcal{P}(\mathcal{Y})\):
\begin{align}
    f_{y}(P_Y) &= P_{Y}(y) \qquad \text{for \(\abs{\mathcal{Y}} - 1\) elements \(y\) from \(\mathcal{Y}\) }, \\
    f_{Y}(P_Y) & = H(P_{Y}), \\
    f_X(P_Y) & = H(P_{Y} \cdot P_{X|Y}).
\end{align}
Note that we only need to consider \(\abs{\mathcal{Y}} - 1\) elements since \(\sum_{y \in \mathcal{Y}}P_Y(y) =1\).
Hence, under the Markov chain \(U \overset{P_{Y|U}}{\to}  Y \overset{P_{X|Y}}{\to} X \), we have
\begin{align}
    \sum_{u} P_{U}(u) f_{y}(P_{Y|U}(\cdot | u)) & = P_{Y}(y), \\
    \sum_{u} P_{U}(u) f_{Y}(P_{Y|U}(\cdot | u)) & = H(Y|U), \\
    \sum_{u} P_{U}(u) f_{X}(P_{Y|U}(\cdot | u)) & = H(X|U).
\end{align}

According to the support lemma \cite[Appendix C]{gamalNetworkInformationTheory2011}, there exist a random variable \(U^{\prime} \sim P_{U^{\prime}}\) with \(\abs{\mathcal{U}^{\prime}} \leq \abs{\mathcal{Y}} + 1\) and a collection of pmfs \(P_{Y|U^{\prime}}( \cdot | u^{\prime} ) \in \mathcal{P}(\mathcal{Y})\), indexed by \(u^{\prime} \in \mathcal{U}^{\prime}\), such that
\begin{align}
    \sum_{u^{\prime}} P_{U^{\prime}}(u^{\prime}) f_{y}(P_{Y|U^{\prime}}(\cdot | u^{\prime})) & = P_{Y}(y), \label{eq:appendix_weak_converse_cardinality_bound_Y_unchanged}\\
    \sum_{u^{\prime}} P_{U^{\prime}}(u^{\prime}) f_{Y}(P_{Y|U^{\prime}}(\cdot | u^{\prime})) & = H(Y|U), \\
    \sum_{u^{\prime}} P_{U^{\prime}}(u^{\prime}) f_{X}(P_{Y|U^{\prime}}(\cdot | u^{\prime})) & = H(X|U).
\end{align}
It follows from \eqref{eq:appendix_weak_converse_cardinality_bound_Y_unchanged} that under the new Markov chain \(U^{\prime} \overset{P_{Y|U^{\prime}}}{\to}  Y \overset{P_{X|Y}}{\to} X \) the distributions of \(Y\) and \(X\) remain unchanged.
Consider the reverse channel \(P_{U^{\prime} | Y}\) induced by \(P_{U^{\prime}} \) and \(P_{Y|U^{\prime}}\).
Thus, for every \((P_{X}, P_{U|Y})\), we can find a new pair \((P_X, P_{U^{\prime}|Y})\) with \(\abs{\mathcal{U}^{\prime}} \leq \abs{\mathcal{Y}} + 1\) such that
\begin{align}
    I(X;U) & = H(X) - H(X|U) \\
    & = H(X) - H(X | U^{\prime}) \\
    & = I(X; U^{\prime}),
\end{align}
and in the same fashion $ I(Y;U) = I(Y;U^{\prime})$, which completes the proof.
%
\subsection{Sphere Packing Bound}
\label{appendix:weak_converse_cardinality_sphere_packing_bound}

Consider an arbitrary alphabet \(\mathcal{U}\).
Assume \(P_X\) is given and fixed.
For every \(Q_Y\), define
\begin{equation}
    E_{\textnormal{sp}}(R, B, Q_Y) = \max_{ \substack{ P_{U|Y}: \\ I(Q_Y, P_{U|Y}) \leq B } }  \min_{ \substack{Q_{Y|X}: \\ P_X \cdot Q_{Y|X} = Q_{Y}, \\ I(P_X, Q_{Y|X} \cdot P_{U|Y} )  \leq R } } \! \! D(Q_{Y|X} \| P_{Y|X} | P_{X}). \label{eq:appendix_cardinality_bound_E_sp}
\end{equation}
Consider an alphabet \(\mathcal{U}^{\prime}\) with \( \abs{\mathcal{U}^{\prime}} \leq \abs{\mathcal{X}} \abs{\mathcal{Y}} + \abs{\mathcal{Y}} + 1 \), define
\begin{equation}
    E^{\prime}_{\textnormal{sp}}(R, B, Q_Y) =  \max_{ \substack{ P_{U^{\prime}|Y}: \\ I(Q_Y, P_{U^{\prime}|Y}) \leq B } }  \min_{ \substack{Q_{Y|X}: \\ P_X \cdot Q_{Y|X} = Q_{Y}, \\ I(P_X, Q_{Y|X} \cdot P_{U^{\prime}|Y} )  \leq R } } \! \! D(Q_{Y|X} \| P_{Y|X} | P_{X}).
\end{equation}
The task is to show
\begin{equation}
    \min_{Q_Y} E_{\textnormal{sp}}(R, B, Q_Y) = \min_{Q_Y} E^{\prime}_{\textnormal{sp}}(R, B, Q_Y).
\end{equation}
We will instead show that for every \(Q_Y\), we have
\begin{equation}
    E_{\textnormal{sp}}(R, B, Q_Y) = E^{\prime}_{\textnormal{sp}}(R, B, Q_Y).
\end{equation}
There are no limits on \(\abs{\mathcal{U}}\), unlike \(|\mathcal{U}^{\prime}|\), so it is clear that
\begin{equation}
    E_{\textnormal{sp}}(R, B, Q_Y) \geq E^{\prime}_{\textnormal{sp}}(R, B, Q_Y).
\end{equation}
Hence, we only need to establish
\begin{equation}
    E_{\textnormal{sp}}(R, B, Q_Y) \leq E^{\prime}_{\textnormal{sp}}(R, B, Q_Y). \label{eq:appendix_cardinality_sphere_packing_bound_task}
\end{equation}
For any \((R, B, Q_Y)\), assume \((P_{U|Y}^{\ast}, Q_{Y|X}^{\ast})\) is a solution to the RHS of \eqref{eq:appendix_cardinality_bound_E_sp}, i.e.,
\begin{align}
    P_X \cdot Q_{Y|X}^{\ast} & = Q_{Y} \\
    I(Q_Y, P_{U|Y}^{\ast}) & \leq B \label{eq:appendix_cardinality_convex_optimization_1}\\
    I(P_X, Q_{Y|X}^{\ast} \cdot P_{U|Y}^{\ast}) & \leq R, \label{eq:appendix_cardinality_convex_optimization_2}
\end{align}
and more importantly \((P_{U|Y}^{\ast}, Q_{Y|X}^{\ast})\) must satisfy
\begin{equation}
    Q_{Y|X}^{\ast} = \argmin_{ \substack{Q_{Y|X}: \\ P_X \cdot Q_{Y|X} = Q_{Y}, \\ I(P_X, Q_{Y|X} \cdot P_{U|Y}^{\ast} )  \leq R } } \! \! D(Q_{Y|X} \| P_{Y|X} | P_{X}). \label{eq:appendix_cardinality_convex_optimization_3}
\end{equation}
Consider the Markov chain \(X \overset{Q_{Y|X}^{\ast}}{\to} Y \overset{P_{U|Y}^{\ast}}{\to} U^{\ast}\), where we denote the distribution of \(U^{\ast}\) by \(P_{U}^{\ast}\).
Since the RHS of \eqref{eq:appendix_cardinality_convex_optimization_3} is a strictly convex optimization problem, we can solve it using the Lagrangian dual function under the KKT conditions, denoted by \(\mathcal{L}(P_{U}^{\ast})\).
Therefore, under the KKT conditions, \(Q_{Y|X}^{\ast}\) must be the solution to
\begin{equation}
    \pdv{\mathcal{L}(P_{U}^{\ast})}{Q_{Y|X}(y | x)} = 0, \qquad \forall (x, y) \in \mathcal{X} \times \mathcal{Y}. \label{eq:appendix_cardinality_convex_optimization_4}
\end{equation}
Notice that \(I(P_X, Q_{Y|X}\cdot P_{U|Y}^{\ast}) = H(P_X) - H(X|U^{\ast})\), which is a linear function of \(P_{U}^{\ast}\).
From this, we see that both \(\mathcal{L}(P_{U}^{\ast})\) and the LHS of \eqref{eq:appendix_cardinality_convex_optimization_4} are linear functions of \(P_{U}^{\ast}\).
As a result, \eqref{eq:appendix_cardinality_convex_optimization_4} contain \(|\mathcal{X}||\mathcal{Y}|\) linear functions of \(P_{U}^{\ast}\).
By Appendix \ref{appendix:weak_converse_cardinality_bound_U}, \eqref{eq:appendix_cardinality_convex_optimization_1} and \eqref{eq:appendix_cardinality_convex_optimization_2} result in \(|\mathcal{Y}| + 1\) linear functions.
Hence, we need to consider \(|\mathcal{X}||\mathcal{Y}| + |\mathcal{Y}| + 1\) functions in total.
From the support lemma, there exists a \(P_{U^{\prime}|Y}^{\ast}\) with \( \abs{\mathcal{U}^{\prime}} \leq \abs{\mathcal{X}} \abs{\mathcal{Y}} + \abs{\mathcal{Y}} + 1 \) satisfying
\begin{align}
    I(Q_Y, P_{U^{\prime}|Y}^{\ast}) & \leq B \\
    I(P_X, Q_{Y|X}^{\ast} \cdot P_{U^{\prime}|Y}^{\ast}) & \leq R,
\end{align}
and more importantly due to \eqref{eq:appendix_cardinality_convex_optimization_4} we have
\begin{equation}
    Q_{Y|X}^{\ast} = \argmin_{ \substack{Q_{Y|X}: \\ P_X \cdot Q_{Y|X} = Q_{Y}, \\ I(P_X, Q_{Y|X} \cdot P_{U^{\prime}|Y}^{\ast} )  \leq R } } \! \! D(Q_{Y|X} \| P_{Y|X} | P_{X}).
\end{equation}
Due to the maximization over \(P_{U^{\prime}|Y}\) in \(E^{\prime}_{\textnormal{sp}}(R, B, Q_Y)\), we then can see that
\begin{equation}
    E^{\prime}_{\textnormal{sp}}(R, B, Q_Y) \geq E_{\textnormal{sp}}(R, B, Q_Y),
\end{equation}
which completes the proof.


\section{Proofs of Covering Lemmas}
\subsection{Proof of Lemma \ref{lemma:Ahlswede_covering_lemma}}
\label{appendix:Ahlswede_covering_lemma_proof}
Given any length-$n$ sequence \(\bm{x}\), there are \(n!\) possible permutations, which however do not all necessarily lead to distinct outcomes.
Stirling's approximation states that
\begin{equation}
    n! \approx e^{n \log n - n}
\end{equation}
as \(n \to \infty\), i.e., there are plenty of permutations to consider.
Denote the sequence of all possible permutations by \(\pi_1, \pi_2, \ldots, \pi_{n!}\).
Then, for every \(\bm{x} \in \mathcal{T}_n(Q_X)\), we must have
\begin{equation}
    \bigcup_{i =1}^{n!} \pi_{i} [ \bm{x} ] = \mathcal{T}_n(Q_X),
\end{equation}
since for any \(\bm{x}^{\prime} \in \mathcal{T}_n(Q_X)\) and \(\bm{x}^{\prime} \neq \bm{x}\), there is a permutation \(\pi\) such that \(\pi [\bm{x}] = \bm{x}^{\prime}\).
Therefore, for every non-empty set \(\mathcal{A} \subseteq \mathcal{T}_n(Q_X) \), we have
\begin{equation}
    \bigcup_{i =1}^{n!} \pi_{i} [ \mathcal{A} ] = \mathcal{T}_n(Q_X),
\end{equation}
since \(\mathcal{A}\) contains at least one \(\bm{x} \in \mathcal{T}_n(Q_X)\).
Next, for a fixed \(\mathcal{A} \subseteq \mathcal{T}_n(Q_X)\), define
\begin{equation}
    \mathrm{deg}( \bm{x} ) \triangleq \sum_{i =1}^{n!} \idc{ \bm{x} \in \pi_{i} [ \mathcal{A}] } \qquad \forall \bm{x} \in \mathcal{T}_n(Q_X),
\end{equation}
i.e., we list the sequence of permuted sets \(\pi_1[ \mathcal{A}], \pi_2 [ \mathcal{A}], \ldots, \pi_{n!} [ \mathcal{A} ]\) and look at the number of them that contain the sequence  \(\bm{x}\).
We have the following result.
\begin{lemma}
    \label{lemma:Ahlswede_covering_lemma_degree_of_edge}
    For any \(\mathcal{A} \subseteq \mathcal{T}_n(Q_X)\), we have
    \begin{equation}
        \mathrm{deg} ( \bm{x} ) =  \frac{ \abs{  \mathcal{A} } \times n! }{ \abs{ \mathcal{T}_n (Q_X) }},  \qquad \forall \bm{x} \in \mathcal{T}_n(Q_X).
    \end{equation}
\end{lemma}
\begin{proof}
    Consider a \(\bm{x} \in \mathcal{T}_n(Q_X)\) and let \(\mathrm{deg} ( \bm{x} ) = d\).
    Assume without loss of generality that \(\bm{x}\) is contained in the sets \(\pi_1[ \mathcal{A}], \pi_2 [ \mathcal{A}], \ldots, \pi_{d} [ \mathcal{A} ]\).
    For any \(\bm{x}^{\prime} \in \mathcal{T}_n(Q_X)\) and \(\bm{x}^{\prime} \neq \bm{x}\), there is a permutation \(\pi\) such that \(\pi [ \bm{x}] = \bm{x}^{\prime}\).
    Thus, \(\bm{x}^{\prime}\) must be contained in the sets
    \begin{equation}
        \pi \circ \pi_1[ \mathcal{A}], \pi \circ \pi_2 [ \mathcal{A}], \ldots, \pi \circ \pi_{d} [ \mathcal{A} ].
    \end{equation}
    Hence, we have
    \begin{equation}
        \mathrm{deg}( \bm{x}^{\prime} ) \geq d = \mathrm{deg} ( \bm{x} ).
    \end{equation}
    In the same fashion, we can show that
    \begin{equation}
        \mathrm{deg}( \bm{x} ) \geq \mathrm{deg} ( \bm{x}^{\prime} ).
    \end{equation}
    Therefore, we have \(\mathrm{deg}( \bm{x} ) = \mathrm{deg} ( \bm{x}^{\prime} )\) for every \(\bm{x}, \bm{x}^{\prime} \in \mathcal{T}_n(Q_X)\).
    Now we observe that
    \begin{align}
        \sum_{  \bm{x} \in \mathcal{T}_n(Q_X) } \mathrm{deg} (\bm{x}) & = \sum_{  \bm{x} \in \mathcal{T}_n(Q_X) } \sum_{i =1}^{n!} \idc{ \bm{x} \in \pi_{i} [ \mathcal{A}] } \\
        & =  \sum_{i =1}^{n!}  \sum_{  \bm{x} \in \mathcal{T}_n(Q_X) } \idc{ \bm{x} \in \pi_{i} [ \mathcal{A}] } \\
        & = \sum_{i =1}^{n!} \abs{ \pi_i [ \mathcal{A}] } \\
        & = \sum_{i =1}^{n!} \abs{  \mathcal{A} } \\
        & = \abs{  \mathcal{A} } \times n!.
    \end{align}
    Since $\mathrm{deg} ( \bm{x} )$ is the same across all $\bm{x} \in \mathcal{T}_n(Q_X)$, it follows that
    \begin{equation}
        \mathrm{deg} ( \bm{x} ) =  \frac{ \abs{  \mathcal{A} } \times n! }{ \abs{ \mathcal{T}_n (Q_X) }}  \qquad \forall \bm{x} \in \mathcal{T}_n(Q_X),
    \end{equation}
    which completes the proof.
\end{proof}
The task is to show that for every \(\mathcal{A} \subseteq \mathcal{T}_n(Q_X)\), we can find a sequence of permutations \(\pi_1, \pi_2, \ldots, \pi_k\) such that $\bigcup_{i =1}^{k} \pi_i [ \mathcal{A} ] = \mathcal{T}_n(Q_X)$,
i.e.,
\begin{equation}
    \sum_{ \bm{x} \in \mathcal{T}_n(Q_X) } \mathbbmss{1} \Big\{ \bm{x} \not \in \bigcup_{i =1}^{k} \pi_i [ \mathcal{A} ]  \Big\} = 0.
\end{equation}
Let \(\bm{\pi}\) be a random permutation, uniform on the set of all permutations \(\{ \pi_1, \pi_2, \ldots, \pi_{n!} \} \), i.e.,
\begin{equation}
    \P \{ \bm{\pi} = \pi_i \} = \frac{1}{n!}, \qquad \forall i \in [ n!].
\end{equation}
The existence of a sequence of permutations \(\pi_1, \pi_2, \ldots, \pi_k\) that satisfies the desired property is proved through averaging over a random ensemble of \(k\) permutations: a length-\(k\) vector \(\bar{\bm{\pi}} \triangleq ( \bm{\pi}_1, \bm{\pi}_2, \ldots, \bm{\pi}_k)\) where every \(\bm{\pi}_i\) independently follows the same distribution as \(\bm{\pi}\).
It follows that
\begin{align}
    & \E_{\bar{\bm{\pi}}} \bigg[ \sum_{ \bm{x} \in \mathcal{T}_n(Q_X) } \mathbbmss{1} \Big\{ \bm{x} \not \in \bigcup_{i =1}^{k} \bm{\pi}_i [ \mathcal{A} ]  \Big\} \bigg] \nonumber \\
    & = \sum_{ \bm{x} \in \mathcal{T}_n(Q_X) } \E_{\bar{\bm{\pi}}} \bigg[  \mathbbmss{1} \Big\{ \bm{x} \not \in \bigcup_{i =1}^{k} \bm{\pi}_i [ \mathcal{A} ]  \Big\} \bigg] \\
    & = \sum_{ \bm{x} \in \mathcal{T}_n(Q_X) } \P \Big\{  \bm{x} \not \in \bigcup_{i =1}^{k} \bm{\pi}_i [ \mathcal{A} ]  \Big\} \\
    & = \sum_{ \bm{x} \in \mathcal{T}_n(Q_X) } \P \left\{  \bm{x} \not \in \bm{\pi}_i [ \mathcal{A} ] , \forall i \in [k] \right\} \\
    & = \sum_{ \bm{x} \in \mathcal{T}_n(Q_X) } \prod_{i=1}^{k} \P \left\{ \bm{x} \not \in \bm{\pi}_i [ \mathcal{A}] \right\} \label{eq:appendix_Ahlswede_covering_lemma_independent_selection}  \\
    & = \sum_{ \bm{x} \in \mathcal{T}_n(Q_X) } \prod_{i=1}^{k}( 1 - \P \left\{ \bm{x} \in \bm{\pi}_i [ \mathcal{A}] \right\} ) \\
    & = \sum_{ \bm{x} \in \mathcal{T}_n(Q_X) } ( 1 - \P \left\{ \bm{x} \in \bm{\pi} [ \mathcal{A}] \right\} )^k  \label{eq:appendix_Ahlswede_covering_lemma_same_distribution}\\
    & \leq \sum_{ \bm{x} \in \mathcal{T}_n(Q_X) } e^{ - k \P \left\{ \bm{x} \in \bm{\pi} [ \mathcal{A}] \right\}  } \label{eq:appendix_Ahlswede_covering_lemma_inequality} \\
    & = \sum_{ \bm{x} \in \mathcal{T}_n(Q_X) } e^{ - k \abs{\mathcal{A}} \abs{\mathcal{T}_n(Q_X)}^{-1}  } \label{eq:appendix_Ahlswede_covering_lemma_uniform_distribution_dege}\\
    & = e^{ - k \abs{\mathcal{A}} \abs{\mathcal{T}_n(Q_X)}^{-1} + \log \abs{\mathcal{T}_n(Q_X)}  },
\end{align}
where \eqref{eq:appendix_Ahlswede_covering_lemma_independent_selection} is due to the independent selection of permutations in the random ensemble;
\eqref{eq:appendix_Ahlswede_covering_lemma_same_distribution} is because every \(\bm{\pi}_i\) in the random ensemble has the same distribution as \(\bm{\pi}\);
in \eqref{eq:appendix_Ahlswede_covering_lemma_inequality}, we make use of \((1-x)^t \leq e^{-tx}\) for \(x \in [0,1]\) and \(t \geq 0\);
and \eqref{eq:appendix_Ahlswede_covering_lemma_uniform_distribution_dege} follows from
\begin{equation}
    \P \left\{ \bm{x} \in \bm{\pi} [ \mathcal{A}] \right\} = \frac{ \mathrm{deg}(\bm{x})}{n!} = \abs{\mathcal{A}} \abs{\mathcal{T}_n(Q_X)}^{-1},
\end{equation}
on account of the uniform distribution of \(\bm{\pi}\).

It immediately follows that if \(k > \abs{\mathcal{A}}^{-1} \abs{\mathcal{T}_n(Q_X)} \log \abs{ \mathcal{T}_n(Q_X)}\), we have
\begin{equation}
    \E_{\bar{\bm{\pi}}} \bigg[ \sum_{ \bm{x} \in \mathcal{T}_n(Q_X) } \mathbbmss{1} \Big\{ \bm{x} \not \in \bigcup_{i =1}^{k} \bm{\pi}_i [ \mathcal{A} ]  \Big\} \bigg] < 1.
\end{equation}
Since the cardinality of a set must be either \(0\) or a positive integer, there must exist a sequence of permutations \(\pi_1, \pi_2, \ldots, \pi_k \) such that
\begin{equation}
    \sum_{ \bm{x} \in \mathcal{T}_n(Q_X) } \mathbbmss{1} \Big\{ \bm{x} \not \in \bigcup_{i =1}^{k} \pi_i [ \mathcal{A} ]  \Big\} = 0,
\end{equation}
which completes the proof.

\subsection{Proof of Lemma \ref{lemma:Ahlswede_covering_lemma_expurgation}}
\label{appendix:Ahlswede_covering_lemma_expurgation_proof}
Following the same steps in the proof of Lemma \ref{lemma:Ahlswede_covering_lemma}, for every \(\mathcal{A}_j \in \mathcal{F}\), after averaging over the random ensemble \(\bar{\bm{\pi}}\), we have
\begin{align}
    \E_{\bar{\bm{\pi}}} \bigg[ \sum_{ \bm{x} \in \mathcal{T}_n(Q_X) } \mathbbmss{1} \Big\{ \bm{x} \not \in \bigcup_{i =1}^{k} \bm{\pi}_i [ \mathcal{A}_j ]  \Big\} \bigg] & \leq e^{ - k \abs{\mathcal{A}_j} \abs{\mathcal{T}_n(Q_X)}^{-1} + \log \abs{\mathcal{T}_n(Q_X)}  } \\
    & \leq e^{ - k |\mathcal{A}_{\min}| \abs{\mathcal{T}_n(Q_X)}^{-1} + \log \abs{\mathcal{T}_n(Q_X)}  } \\
    & < \frac{1}{2}, \label{eq:appendix_Ahlswede_covering_lemma_expurgation}
\end{align}
if \(k > |\mathcal{A}_{\min}|^{-1} \abs{\mathcal{T}_n(Q_X)} \log 2 \abs{ \mathcal{T}_n(Q_X)}\).
We now define the random set \(\bm{A}\) such that
\begin{equation}
    P_{\bm{A}}( \mathcal{A}_j ) = \frac{1}{\abs{\mathcal{F}}} \qquad \forall \mathcal{A}_j \in \mathcal{F},
\end{equation}
i.e., \(\bm{A}\) is uniformly distributed on the collection $\mathcal{F}$.
Since \eqref{eq:appendix_Ahlswede_covering_lemma_expurgation} holds for every \(\mathcal{A} \in \mathcal{F}\), it follows that
\begin{equation}
    \E_{\bm{A}} \bigg\{ \E_{\bar{\bm{\pi}}} \bigg[ \sum_{ \bm{x} \in \mathcal{T}_n(Q_X) } \mathbbmss{1} \Big\{ \bm{x} \not \in \bigcup_{i =1}^{k} \bm{\pi}_i [ \bm{A} ]  \Big\} \bigg] \bigg\} < \frac{1}{2}.
\end{equation}
Recall that every \(\bm{\pi}_i\) independently follows the uniform distribution over all possible \(n!\) permutations, so \(\bar{\bm{\pi}}\) and \(\bm{A}\) are independent.
Hence, we can exchange the order of expectations and obtain
\begin{equation}
    \E_{\bar{\bm{\pi}}} \bigg\{ \E_{\bm{A}} \bigg[ \sum_{ \bm{x} \in \mathcal{T}_n(Q_X) } \mathbbmss{1} \Big\{ \bm{x} \not \in \bigcup_{i =1}^{k} \bm{\pi}_i [ \bm{A} ]  \Big\} \bigg] \bigg\} < \frac{1}{2}.
\end{equation}
Thus, there must exist a sequence of permutations \(\pi_1, \pi_2, \ldots, \pi_k\) such that
\begin{equation}
    \E_{\bm{A}} \bigg[ \sum_{ \bm{x} \in \mathcal{T}_n(Q_X) } \mathbbmss{1} \Big\{ \bm{x} \not \in \bigcup_{i =1}^{k} \pi_i [ \bm{A} ]  \Big\} \bigg] < \frac{1}{2}.
\end{equation}
Therefore, for this particular sequence of permutations, at least half of the sets \(\mathcal{A}_j \in \mathcal{F}\) must satisfy
\begin{equation}
    \sum_{ \bm{x} \in \mathcal{T}_n(Q_X) } \mathbbmss{1} \Big\{ \bm{x} \not \in \bigcup_{i =1}^{k} \pi_i [ \mathcal{A}_j ]  \Big\} < 1.
\end{equation}
This happens only if for this half of sets \(\mathcal{F}\), we have
\begin{equation}
    \sum_{ \bm{x} \in \mathcal{T}_n(Q_X) } \mathbbmss{1} \Big\{ \bm{x} \not \in \bigcup_{i =1}^{k} \pi_i [ \mathcal{A}_j ]  \Big\} = 0,
\end{equation}
which completes the proof.
\subsection{Proof of Lemma \ref{prop:ratio_unique_codewords_ensemble}}
\label{appendix:proof_ratio_unique_codewords_ensemble}
Fix $\delta \in (0,1)$. If we select \(\delta e^{n\tilde{R}}\) unique codewords from \(\mathcal{T}_n(Q_X)\), then there are
\begin{equation}
    \label{eq:unique_codeword_selection}
    \binom{ |\mathcal{T}_n(Q_X)| }{\delta  e^{n\tilde{R}}}
\end{equation}
possible selections.
We denote by \(\mathcal{H}_i= \{\bm{x}_1, \bm{x}_2, \ldots, \bm{x}_{\delta  e^{n\tilde{R}}}\}\) a possible selection (set).
We will write \(\mathcal{C}_n \subset \mathcal{H}_i\) if all of the unique codewords in \(\mathcal{C}_n\) are contained in \(\mathcal{H}_i\).
If a codebook \(\mathcal{C}_n \in \mathcal{T}_n(Q_X)^{e^{n\tilde{R}}}\) has less than \(\delta  e^{n\tilde{R}}\) unique codewords, i.e., \(|\mathcal{C}_n| \leq \delta  e^{n\tilde{R}}\), then we can construct a possible selection \(\mathcal{H}_i\) 
from \(\mathcal{C}_n\), i.e., in \(\mathcal{H}_i\) we first select the unique codewords in  \(\mathcal{C}_n\) and then arbitrarily select the remaining codewords from \(\mathcal{T}_n(Q_X)\).
Thus, through contradiction, we see that for every \(\mathcal{C}_n \in \mathcal{T}_n(Q_X)^{e^{n\tilde{R}}}\) with \(|\mathcal{C}_n| \leq \delta  e^{n\tilde{R}}\), there must exist a selection \(\mathcal{H}_i\) 
such that \(\mathcal{C}_n \subset \mathcal{H}_i\), since otherwise we can construct a new selection.

For a selection \(\mathcal{H}_i= \{\bm{x}_1, \bm{x}_2, \ldots, \bm{x}_{\delta  e^{n\tilde{R}}}\}\), notice that \(\mathcal{C}_n \subset \mathcal{H}_i\) means that the codewords of \(\mathcal{C}_n\) are from the set \(\{\bm{x}_1, \bm{x}_2, \ldots, \bm{x}_{\delta  e^{n\tilde{R}}}\}\).
Consequently, we observe that
\begin{equation}
    | \{\mathcal{C}_n \in \mathcal{T}_n(Q_X)^{e^{n\tilde{R}}} : \mathcal{C}_n \subset \mathcal{H}_i \}| \leq (\delta  e^{n\tilde{R}})^{e^{n\tilde{R}}},
\end{equation}
where we upper bound \(| \{\mathcal{C}_n \in \mathcal{T}_n(Q_X)^{e^{n\tilde{R}}} : \mathcal{C}_n \subset \mathcal{H}_i \}|\) by the size of the product set over \(\{\bm{x}_1, \bm{x}_2, \ldots, \bm{x}_{\delta  e^{n\tilde{R}}}\}\).
Hence, we have
\begin{align}
    \left|\left\{ \mathcal{C}_n \in \mathcal{T}_n(Q_X)^{e^{n\tilde{R}}} :   | \mathcal{C}_n |  \leq \delta e^{n\tilde{R}} \right\}\right| & \leq \sum_{\mathcal{H}_i}  | \{\mathcal{C}_n \in \mathcal{T}_n(Q_X)^{e^{n\tilde{R}}} : \mathcal{C}_n \subset \mathcal{H}_i \}| \\
    & \leq \binom{ |\mathcal{T}_n(Q_X)| }{\delta  e^{n\tilde{R}}} \times (\delta  e^{n\tilde{R}})^{e^{n\tilde{R}}}.
\end{align}
Recall the distribution of the random ensemble
\begin{equation}
    \P\{ \bm{C} = \mathcal{C}_n \} = \left(  \frac{1}{ | \mathcal{T}_n(Q_X)| } \right)^{e^{n\tilde{R}}}.
\end{equation}
Therefore, we see that
\begin{align}
    & \P \left\{  | \bm{C} |   \leq \delta e^{n\tilde{R}}  \right\} \nonumber \\
    & \leq \binom{ |\mathcal{T}_n(Q_X)| }{\delta  e^{n\tilde{R}}} \times \left( \frac{ \delta  e^{n\tilde{R}}  }{ | \mathcal{T}_n(Q_X)| } \right)^{e^{n\tilde{R}}} \\
    & \leq \left(  \frac{ e \times  |\mathcal{T}_n(Q_X)| }{ \delta  e^{n\tilde{R}} }  \right)^{ \delta  e^{n\tilde{R}}} \times \left( \frac{ \delta  e^{n\tilde{R}}  }{ | \mathcal{T}_n(Q_X)| } \right)^{e^{n\tilde{R}}} \label{eq:binomial_coefficient_inequality}\\
    & = e^{ \delta  e^{n\tilde{R}} } \times |\mathcal{T}_n(Q_X)|^{(\delta -1 )e^{n\tilde{R}}} \times ( \delta  e^{n\tilde{R}})^{(1-\delta)e^{n\tilde{R}}} \\
    & \leq  e^{ \delta  e^{n\tilde{R}} } \times e^{nH(Q_X) \times (\delta -1)e^{n\tilde{R}}} \times e^{ (1-\delta)e^{n\tilde{R}} \log(\delta  e^{n\tilde{R}}) } \label{eq:binomial_coefficient_delta} \\
    & = \exp\Big\{ (1 - \delta)e^{n\tilde{R}}\Big( \frac{\delta}{ 1- \delta} - nH(Q_X) + n\tilde{R} + \log \delta  \Big) \Big\}, \label{eq:double_exponentially_decaying}
\end{align}
where in \eqref{eq:binomial_coefficient_inequality} we use this inequality on binomial coefficient
\begin{equation}
    \binom{n}{k} \leq \left(  \frac{e \times n}{k}  \right)^{k};
\end{equation}
in \eqref{eq:binomial_coefficient_delta} we notice \(\delta -1 < 0\) and \(|\mathcal{T}_n(Q_X)| \geq e^{nH(Q_X)}\).
Since \(\delta \in (0,1)\) and \(H(Q_X) > \tilde{R}\), it is clear that \eqref{eq:double_exponentially_decaying} decays to \(0\) double exponentially.

\section*{Acknowledgement}
The authors would like to thank the associate editor and two anonymous reviewers for their timely and constructive feedback that helped improve the paper.



\bibliography{ref}
\bibliographystyle{IEEEtran}

\end{document}